\documentclass[showpacs,twocolumn,pra,notitlepage]{revtex4-2}
\usepackage{graphics} 
\usepackage{tikz}
\usepackage[]{qcircuit}
\usetikzlibrary{arrows}
\usetikzlibrary{shapes,fadings,snakes}
\usetikzlibrary{decorations.pathmorphing,patterns}
\usetikzlibrary{calc}
\usepackage{hyperref}
\usepackage{amsfonts}
\usepackage{amssymb}
\usepackage{amscd}
\usepackage{amsmath}    
\usepackage{amsthm}
\usepackage{enumitem}
\usepackage{epsfig}
\usepackage{subfigure}
\usepackage{xcolor}
\usepackage{physics}
\usepackage{booktabs}
\usepackage{epstopdf}
\usepackage[most]{tcolorbox}
\usepackage{multirow}
\usepackage{bbm}
\usepackage{overpic}
\usepackage{diagbox}
\usepackage{lineno}
\usepackage[titletoc]{appendix}

\graphicspath{{./Figures/}}

\newtheorem{theorem}{Theorem}
\newtheorem{lemma}{Lemma}
\newtheorem{corollary}{Corollary}
\newtheorem{claim}{Claim}
\newtheorem{definition}{Definition}

\newtheorem{observation}{Observation}

\newtheorem{problem}{Problem}

\newtcolorbox[auto counter]{mybox}[2][]{
	enhanced,
	breakable,
	colback=blue!5!white,
	colframe=blue!75!black,
	fonttitle=\bfseries,
	title=Box \thetcbcounter: #2,#1
}

\newcounter{mybox}
\refstepcounter{mybox}

\begin{document}

\title{Quantum Complementarity Approach to Device-Independent Security}

\date{\today}

\author{Xingjian Zhang$^{1}$}
\author{Pei Zeng$^{1}$}
\author{Tian Ye$^{1}$}
\author{Hoi-Kwong Lo$^{2,\,3,\,4}$}
\author{Xiongfeng Ma$^{1}$}

\affiliation{$^1$Center for Quantum Information, Institute for Interdisciplinary Information Sciences, Tsinghua University, Beijing 100084, China}
\affiliation{$^2$Department of Physics, University of Hong Kong, Pokfulam, Hong Kong}
\affiliation{$^3$Centre for Quantum Information and Quantum Control (CQIQC), Dept. of Electrical \& Computer Engineering and Department of Physics, University of Toronto, Toronto, Ontario, M5S 3G4, Canada}
\affiliation{$^4$Quantum Bridge Technologies, Inc., 100 College Street, Toronto, ON M5G 1L5, Canada}

\begin{abstract}
Complementarity is an essential feature of quantum mechanics. The preparation of an eigenstate of one observable implies complete randomness in its complementary observable. In quantum cryptography, complementarity allows us to formulate security analyses in terms of phase-error correction. However, in the device-independent regime that offers security without device characterization, the concept becomes much subtler. Security proofs of device-independent quantum cryptography tasks are often complex and quite different from those of their more standard device-dependent cousins. The existing proofs pose huge challenges to experiments, among which large data-size requirement is a crux. Here, we show the complementarity security origin of the device-independent tasks. By linking complementarity with quantum nonlocality, we recast the device-independent scheme into a quantum error correction protocol. Going beyond the identical-and-independent-distribution case, we consider the most general attack. We generalize the sample entropy in classical Shannon theory for the finite-size analysis. Our method exhibits good finite-size performance and brings the device-independent scheme to a more practical regime. Applying it to the data in a recent ion-trap-based device-independent quantum key distribution experiment, one could reduce the requirement on data size to less than a third. Furthermore, the complementarity approach can be naturally extended to advantage key distillation to ease experiments by tolerating higher loss and lower transmittance.
\end{abstract}

\maketitle

\newpage

\emph{Introduction. ---}As demonstrated by Heisenberg, complementary observables, such as position and momentum, where the full knowledge of one property implies the ignorance of the other, exist~\cite{heisenberg1927anschaulichen}. The inability to measure multiple observables simultaneously indicates a fundamental limitation of learning physical properties, thus violating the tenet of determinism in Newtonian physics. In the quest to understand quantum physics, complementarity is found to be closely related to the nonlocal physical state --- entanglement~\cite{einstein1935can}. As the strongest method to verify its existence, the Bell test has drawn considerable attention~\cite{bell1964einstein}. In the famous Clauser-Horne-Shimony-Holt (CHSH) Bell test, space-like separated parties, Alice and Bob, each have a black box for random measurements. Without any communication, Alice and Bob each send a random input $X,Y\in\{0,1\}$ to their own devices, which then output a value $A,B\in\{+1,-1\}$, respectively. With respect to these random variables, the expected Bell value is given by the following expression~\cite{clauser1969proposed},
\begin{equation}\label{Eq:AveBellValue}
\begin{aligned}
S &= \sum_{x,y\in\{0,1\}}(-1)^{xy} \mathbb{E}(AB|X=x,Y=y),
\end{aligned}
\end{equation}
where the function $\mathbb{E}(\cdot)$ represents the expected value and lowercase letters represent the specific realizations of the associated random variables. Deterministic classical physics upper bounds the Bell value by $S\leq2$, but quantum physics allows for a higher value of $S$. Such a correlation, called Bell nonlocality, necessarily implies entanglement. In addition, to unveil the nonlocal correlation, the observables corresponding to the random inputs, $\hat{A}_0,\hat{A}_1$ on Alice's side and $\hat{B}_0,\hat{B}_1$ on Bob's side, must be non-commuting. Furthermore, the nonlocal correlation indicates the unpredictable randomness of the outputs. As the conclusions do not depend on any system characterization, the scenario enjoys a \textit{device-independent} feature~\cite{mayers1998quantum,acin2007device}.

Aside from fundamental interests, with recent progress in quantum information, complementarity and nonlocality have found novel applications. One of the most profound applications of complementary observables is quantum communication. In quantum key distribution (QKD), where two remote parties aim to obtain identical private keys, legitimate users can encode their information in the basis states of complementary observables at random~\cite{bennett1984quantum,ekert1991quantum}. With the aid of measurement complementarity and quantum error correction, QKD security can be established~\cite{lo1999unconditional,shor2000simple,koashi2009simple}. The connection between quantum nonlocality and measurement complementarity allows us to go even further --- giving rise to the device-independent quantum key distribution (DIQKD)~\cite{mayers1998quantum,acin2007device}. Legitimate users do not need to characterize their quantum devices to guarantee security. Instead, they take key generation rounds and Bell tests at random, where the Bell value provides a ``self-test'' of the devices in use. Quantum nonlocality solves the security loopholes caused by malfunctioning or imperfection of quantum devices in common device-dependent QKD protocols once and for all~\cite{xu2020secure}\footnote{Note, however, that DIQKD is still vulnerable to memory attacks and covert channels. A memory attack may occur when the devices are reused for multiple QKD sessions, where the untrusted devices may store the key information in one DIQKD session and leak it via the necessary public communication required by the protocol in new sessions~\cite{barrett2013memory}. A covert channel signals the key information to the outside in an unnoticed way other than the public communication allowed by the protocol. We do not discuss these issues in this work. Particularly, there is no memory attack if we only consider one QKD session. For discussions on practical countermeasures, one may refer to Ref.~\cite{curty2019foiling}. Note that in a different QKD protocol, measurement-device-independent quantum key distribution (MDIQKD)~\cite{lo2012measurement,braunstein2012side}, the problem of memory attacks and covert channels can be naturally avoided on the measurement site.}.
With recent development in the theory of quantum entropies, a few analyses finally arrive at full device-independent security against the most general attacks~\cite{vazirani2014fully,miller2016robust,arnon2018practical,knill2018quantum}.  On the experimental side, after decades of efforts, a complete loophole-free DIQKD experiment has been conducted for the first time with delicate implementation on an ion-trap platform~\cite{nadlinger2022experimental}. There is also progress in the photonic~\cite{liu2022toward} and cold-atom~\cite{zhang2022device} platforms. These works mark a major breakthrough toward a future quantum internet.

Despite the development in both theories and experiments of DIQKD, the role of complementarity remains vague in device-independent security. Rather than being based on the intuitive thinking of uncertainty relation, which is common practice in the trusted-device scenarios~\cite{koashi2009simple}, the existing security analyses are often abstract in an operational interpretation and share a much different conceptual flavor. Some works attempt to unite the device-independent scenario to its device-dependent cousins. Overall, the target is characterizing the key privacy by the complementary observable to the key generation measurement and linking it with the Bell value. Yet the picture of this approach is incomplete. Early efforts fail in providing a tight relation between nonlocality and complementarity~\cite{tsurumaru2016multi}. While a recent work finally gives the desired link~\cite{woodhead2021device}, the result does not go into the general adversarial statistics that are possibly not independent and identically distributed (i.i.d.), unless resorting to the existing methods of quantum entropic measures for a full DIQKD security analysis.

\emph{Security Analysis. ---}In this work, we present a simple security proof of sequential DIQKD against the most general attacks --- coherent attacks --- where the system behavior is not assumed to be i.i.d., and complete the analysis along the complementarity approach. We consider a variant of Ekert's entanglement-based QKD protocol~\cite{ekert1991quantum}, which is similar to the CHSH test setting with an additional measurement on Alice's side~\cite{acin2006efficient}. The users first perform random quantum measurements, including key generation measurements, where $X=2,Y=0$, and Bell tests, where $X,Y\in\{0,1\}$. After accumulating a sufficient quantity of data, based on quantum bit error rate $e_b$ and observed average Bell value $\bar{S}$, the users perform information reconciliation and privacy amplification to process the key to be identical and private, respectively.

We analyze the final secure key length after one complete run of the entire protocol. Different from common QKD, as the quantum devices are now uncharacterized or even untrusted, the measurements may have a memory effect and hence history-dependent performance~\cite{pironio2009device}. To deal with this, we can depict an $n$-round experiment with a system evolution on $n$ subsystems for both parties, where the memory effect from previous events is equivalently embedded in the state preparation. We leave the detailed modeling in Appendix~\ref{Supp:SysModel}. Unlike the existing approaches based on entropic measures, we recast the DIQKD security analysis into quantum phase error correction~\cite{lo1999unconditional,koashi2009simple}. In this scenario, information reconciliation and privacy amplification can be decoupled. The former task is essentially a classical procedure and relatively easy to analyze. In $m$ key generation rounds, the number of reconciled key bits can be lower bounded by $I_{AB}\geq m[1-fh(e_b)]$, with $f\geq1$ an appropriate information reconciliation efficiency parameter. The essence of security analysis is to estimate the amount of information leakage in key generation, reflected in the cost of privacy amplification.

The complementarity approach starts with the intuitive idea of quantum uncertainty. The secrecy is cast in Eve's ignorance of the measurement outcomes of $\hat{B}_0$, which form the reconciled key. As shown in Fig.~\ref{Fig:FlowSimple}(a), to quantify the adversary's uncertainty on the user side, Alice and Bob explore the complementary basis. That is, given that Bob has measured the complementary observable of $\hat{B}_0$, how well can Alice guess the outcome? The higher the guessing probability is, the less key information is leaked outside. An incorrect guess in a single round is called a phase error and the joint measurement result is called a phase-error pattern. The cost for privacy amplification is given by the number of most likely phase-error patterns, which we term phase-error cardinality. Then, the target of privacy estimation is two-fold: define the quantum measurement that gives rise to phase errors and estimate the phase-error cardinality from observed statistics.

It is challenging to tackle multiple rounds in bulk and deal with non-i.i.d. statistics directly. Instead,  we start with a single round, where we define the phase error and determine its probability from the Bell value. Then, we relate probabilities with frequencies in the non-i.i.d.~setting, as shown in Fig.~\ref{Fig:FlowSimple}(b). We aim at estimating key privacy in terms of the phase-error sample entropy through observed Bell test statistics.

\begin{figure}[hbt!]
\centering
\resizebox{8.7cm}{!}{\includegraphics{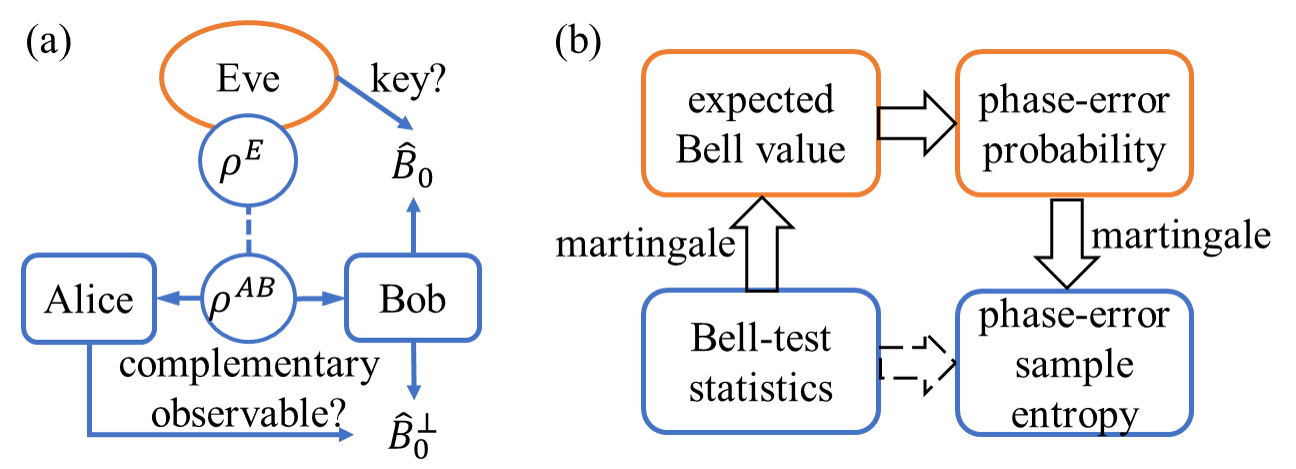}}
\caption{(a) The basic idea of the complementarity approach.
(b) The flowchart of privacy estimation. The two upper boxes indicate single-round probabilities and the two lower boxes indicate statistics over multiple rounds.}
\label{Fig:FlowSimple}
\end{figure}

We first focus on a single round. In the CHSH test, as there are only two measurement bases for each party and two outcomes for each measurement, one can apply Jordan's lemma and restrict the unknown quantum system to a low-dimensional Hilbert space~\cite{pironio2009device} (see Lemma~\ref{Lemma:DimRed} in Appendix). Conditioned on previous events, Eve squashes the quantum system of Alice and Bob into a pair of qubits and prepares their possible measurements to be projective observables. With a qubit subspace, we can define the complementary observable of $\hat{B}_0$, denoted as $\hat{B}_0^\perp$. The eigenvectors of the two observables form mutually unbiased bases. In a virtual experiment, Alice optimizes an auxiliary measurement, $\hat{A}_{\theta}$, to maximize its correlation with $\hat{B}_0^\perp$ over the pair of qubits, $\langle\hat{A}_{\theta}\otimes\hat{B}_0^\perp\rangle$, or equivalently, to minimize the phase-error probability, $(1-{\langle\hat{A}_{\theta}\otimes\hat{B}_0^\perp\rangle})/2$. Though the users cannot actually access $\hat{A}_{\theta}$ and $\hat{B}_0^\perp$, the relation between nonlocal behavior and measurement complementarity creates a bridge. Intuitively, a larger Bell value implies a smaller phase-error probability and more uncertainty from Eve's perspective. Note that different implementations can lead to the same Bell value, among which we need to consider the worst case for the legitimate users. In the end, the phase-error probability can be upper-bounded by
\begin{equation}\label{eq:PhaseErrorProb}
\begin{split}
e_{p}(S)&\equiv\dfrac{1}{2}\left(1-\min_{\rho,\hat{A}_0,\hat{A}_1,\hat{B}_0,\hat{B}_1}\max_{\hat{A}_{\theta}}{\langle\hat{A}_{\theta}\otimes\hat{B}_0^\perp\rangle}\right) \\
&=
\left\{
\begin{tabular}{ll}
$\dfrac{1-\sqrt{(S/2)^2-1}}{2}$, & $2<S\leq2\sqrt{2}$,\\
$\dfrac{1}{2}$, & $S\leq2$,
\end{tabular}
\right.
\end{split}
\end{equation}
where the pair of qubits, $\rho$, and the underlying Bell-test observables should give the expected Bell value, $S$. We leave its derivation in Appendix~\ref{Supp:PhaseErrorProb}.

Next, we extend the results to the case of multiple rounds. Though the behavior of untrusted quantum devices may not be i.i.d.~over the rounds under the most general attacks, fortunately, the sequential setting defines a time order for measurements. Hence, this process assures that the observables in each round depend only on the history but not the future. Furthermore, the correlations in both the quantum state and measurement devices from previous events are fully absorbed into the phase-error probability and the expected Bell value in each round. We can thus apply martingale analyses to their associate random variables. To guarantee that the martingale variables associated with the phase errors and Bell values share the same filtration, we also embed the random variables of measurement settings into the filtration. Then, the single-round relation between phase-error probability and expected Bell value bridges the martingales. In this way, we can estimate the phase-error cardinality in the key generation rounds from statistics in Bell-test rounds. Finally, we employ a modified version of Azuma's concentration inequality~\cite{azuma1967weighted}, namely, Kato's inequality~\cite{kato2020concentration}, to relate the probabilities of martingale variables with their corresponding frequencies.

Different from common complementarity-based QKD analyses, a difficult point here is that the phase-error rate, the average number of phase errors occurring in a round, is not a good measure of privacy. Note that Eve may set different expected Bell values over the rounds, while Alice and Bob can only estimate the average. Due to convexity issues, the relation between the phase-error probability and the Bell value in expectation does not imply a similar relation between the phase-error rate and the average Bell value over multiple rounds, where the former cannot be well upper-bounded via the latter. In Appendix~\ref{Sec:FiniteSizeProb}, we provide a concrete example showing this issue.

Note that the ultimate goal in complementarity-based privacy estimation is not the estimation of phase-error rate but the phase-error cardinality. To solve the convexity issue, we develop a sample-entropy method inspired by classical information theory~\cite{shannon1948mathematical} for the evaluation of the phase-error cardinality directly. The sample entropy is a random variable with respect to a phase-error pattern, the whose expected value is determined by the Shannon entropy. The direct use of the sample entropy of the phase error has a divergence problem. As a workaround, we apply a regularisation trick to the phase-error probability function.

With all the above ingredients, we arrive at the final security analysis result. We summarize it in the following informal theorem and leave the formal one and detailed derivation in Appendix~\ref{Supp:FiniteSize}.

\begin{theorem}[Informal]
After $n$ rounds of quantum measurements in DIQKD, let the observed average Bell value be $\bar{S}$ and the number of key generation rounds be $m$. Given fixed constants $\xi\in(0,\frac{1}{2}),\,\varepsilon_{S},\varepsilon_{pe},\varepsilon_{pc}\in(0,1)$, except for a failure probability no larger than $\varepsilon_f=\varepsilon_{S}+\varepsilon_{pe}+\varepsilon_{pc}$, the privacy amplification cost can be upper-bounded,
\begin{equation}\label{eq:priamp}
\begin{split}
I_{pa}&\leq \frac{1}{1+2\delta_{p}'}\left\{mh\left[e_{p}^{\xi} (S_{est})\right] + n\delta_{h}\right\}-\log\varepsilon_{pc}, \\
\end{split}
\end{equation}
where $h(x)=-x\log x-(1-x)\log(1-x)$ is the binary entropy function, $S_{est}$ is the estimation of the expected Bell value,
\begin{equation}\label{Eq:SimBellEst}
\begin{split}
S_{est}=(1+2\delta_{S}')\bar{S} - \delta_{S}, \\
\end{split}
\end{equation}
$\delta_{p}', \delta_{h}= O\left(\sqrt{\frac{-\log\varepsilon_{pe}}{n}}\right), \delta_{S}', \delta_{S}= O\left(\sqrt{\frac{-\log\varepsilon_{S}}{n}}\right)$, and $e_{p}^{\xi} \left(S\right)$ is the regularised phase-error function,
\begin{equation}\label{Eq:RelaxPhaseErrorProb}
e_{p}^{\xi}(S) = \dfrac{e_{p}(S) + \xi}{1+2\xi},
\end{equation}
with $e_p(S)$ given in Eq.~\eqref{eq:PhaseErrorProb}. Suppose the key length after information reconciliation is $I_{AB}$, the amount of generated key bits is given by $k=I_{AB}-I_{pa}$.
\label{Thm:Main}
\end{theorem}

The estimation failure probabilities $\varepsilon_{S}$ and $\varepsilon_{pe}$ are introduced by Kato's concentration inequality, for the Bell test and phase errors, respectively. The terms $\delta_{h},\delta_{p}',\delta_{S}$ and $\delta_{S}'$ can be optimised. The realization of privacy amplification has a failure probability $\varepsilon_{pc}$. The parameter $\xi$ is introduced to avoid divergence of the sample-entropy function. Under fixed protocol parameters, the value of $\xi$ can be optimized and approaches zero when the data size is large.

\emph{Simulation. ---}We consider an experimental set-up designed according to the ideal setting while suffering from noise and loss. Ideally, the entanglement source prepares Bell state $\ket{\Phi^+}=(\ket{00}+\ket{11})/\sqrt{2}$, the key generation measurements are $\hat{A}_2=\hat{B}_0=\hat{\sigma}_z$, and the rest measurements are the ones giving the largest expected Bell value $S=2\sqrt{2}$. In entanglement distribution, the Bell state undergoes a depolarizing channel before measurement, characterised by $\hat{\mathcal{E}}_d[\rho] = (1-e_d)\rho + e_d \hat{I}/{4}$ with $e_d$ being the depolarising factor. The fidelity of the noisy state $\rho$ to the ideal state $\ket{\Phi^+}$ is calculated by $F(\rho,\Phi^+)=\bra{\Phi^+}\rho\ket{\Phi^+}=1-3e_d/4$. This model is close to the realistic case of entanglement distribution with a fiber link~\cite{li2021experimental}.
Denote the total transmittance of the quantum signals on one side to be $\eta$. That is, with probability $(1-\eta)$, no detection is made. In a standard protocol, the users assign the undetected events a value of $1$. Alice and Bob randomly choose the measurement settings. We fix the total failure probability of the protocol to $\varepsilon_f=5\times10^{-11}$.

In the experiment, we need to determine the feasible regions of parameters $\eta$ and $e_d$ that render a positive key rate. The simulation result is shown in Fig.~\ref{Fig:Feasible}. In the attached tables, we list the tolerable channel-noise levels and the threshold transmittance for a positive key rate.
As shown by the feasible experimental regions in Fig.~\ref{Fig:Feasible}, the complementarity approach exhibits a fast convergence to the asymptotic limit of infinite data size. The asymptotic limit is consistent with the result assuming i.i.d.~attacks \cite{pironio2009device}.

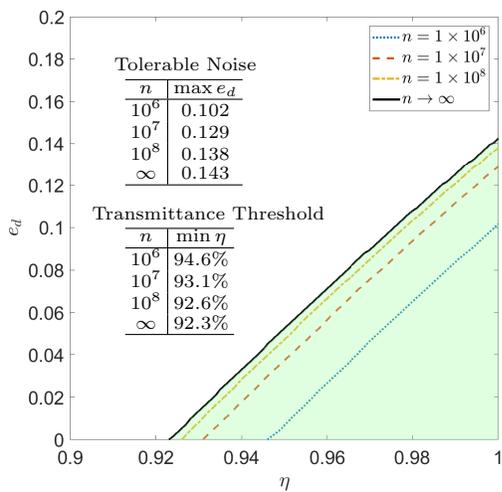
\begin{figure}[hbt!]
\centering
{\begin{overpic}[width=0.85\hsize]{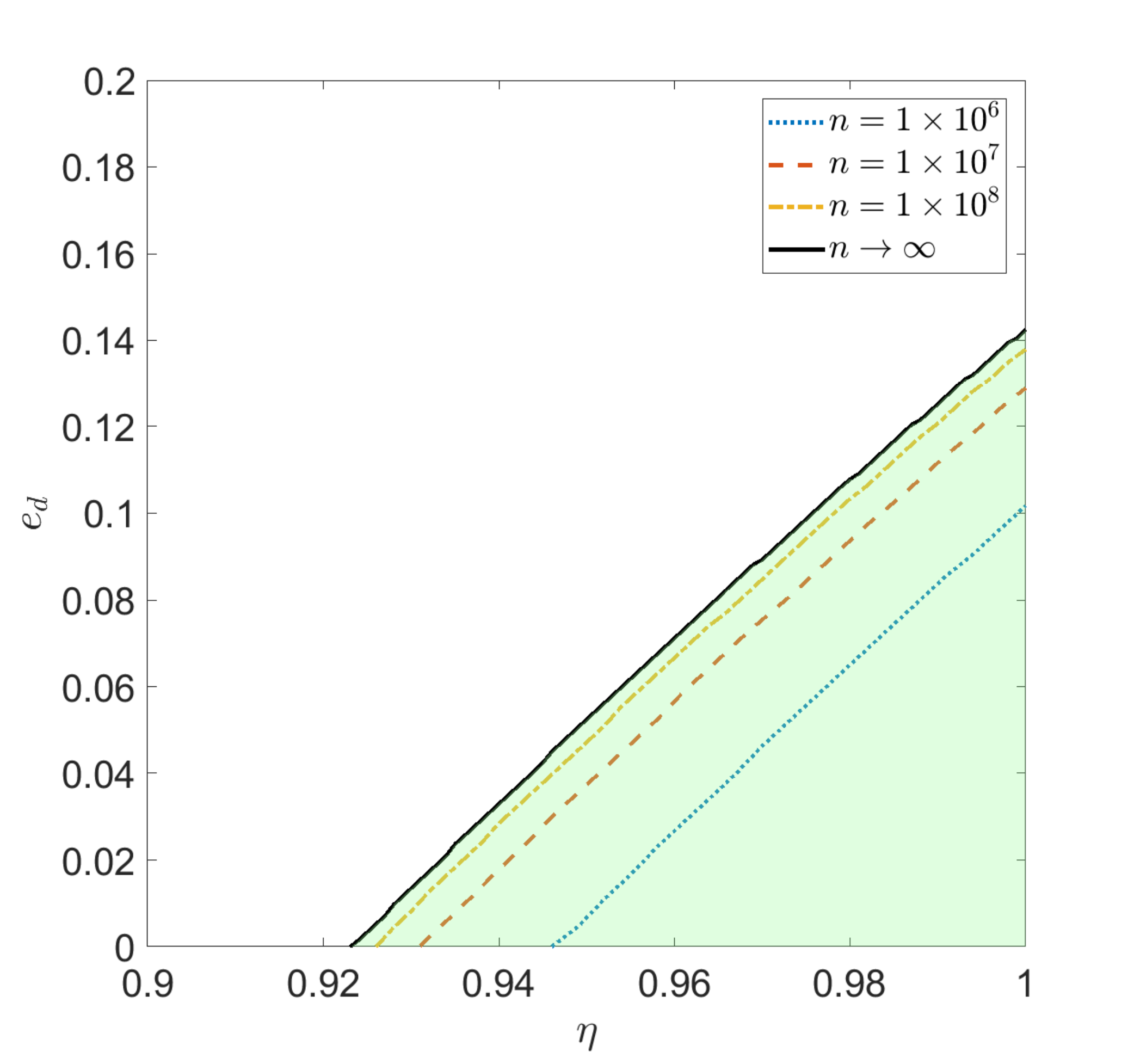}
\put(21,77)
{\fontsize{7}{7}\selectfont Tolerable Noise}
\put(23,65)
{\fontsize{7}{7}\selectfont
\begin{tabular}{c|c}
\hline $n$ & $\max{e_d}$ \\
\hline
$10^6$ & $0.102$ \\
$10^7$ & $0.129$ \\
$10^8$ & $0.138$ \\
$\infty$ & $0.143$ \\
\hline
\end{tabular}}
\put(17,50)
{\fontsize{7}{7}\selectfont Transmittance Threshold}
\put(23,38)
{\fontsize{7}{7}\selectfont
\begin{tabular}{c|c}
\hline $n$ & $\min{\eta}$ \\
\hline
$10^6$ & $94.6\%$ \\
$10^7$ & $93.1\%$ \\
$10^8$ & $92.6\%$ \\
$\infty$ & $92.3\%$ \\
\hline
\end{tabular}}
\end{overpic}}
\caption{Simulation results of key-generation performance with $\eta_A=\eta_B=\eta$.
The bottom right corner under each line represents the feasible region of experimental parameters $\eta,e_d$. For $n\rightarrow\infty$, we depict the feasible region in green.}
\label{Fig:Feasible}
\end{figure}

To better examine the finite-size performance, we consider an experimental setting close to the cold-atom and nitrogen-vacancy-center platforms. For these platforms, the efficiency can be close to 1 with heralding. Meanwhile, the depolarizing factor can be controlled down to a reasonable level. The bottleneck lies in the data size required for the experiment, where the successful heralding rate of entangled pairs is normally very low, typically, on the order of tens of pairs per minute. For this purpose, we simulate the smallest data sizes that are required under different values of depolarizing factors, as shown in Fig.~\ref{Fig:Simulation1LOCC}. In comparison, we simulate the data sizes required by the existing entropy-based analysis adopting the entropy accumulation theorem (EAT), including the original method~\cite{arnon2018practical} and a modified result~\cite{murta2019towards}. The min-tradeoff function in the EAT analyses shares a similar position as the phase-error sample entropy in the complementarity-based analysis. The simulation results show that our method is more advantageous in saving experimental time. For instance, at $e_d=10\%$, which corresponds to a state preparation fidelity of $92.5\%$, a reasonable experimental parameter achievable by state-of-the-art cold-atom~\cite{rosenfeld2017event,zhang2022device} and nitrogen-vacancy-centre~\cite{hensen2015loophole} platforms, the original EAT method $6.16\times10^8$ rounds and the modified EAT method requires $2.75\times10^7$ rounds for positive key generation. In comparison, our security analysis cuts down the experimental time by two orders, requiring $9.12\times10^5$ rounds, making a DIQKD realization on these experimental platforms plausible within a reasonable time under decent security parameters. With a larger depolarizing factor, the advantage of our result becomes more significant. We find a similar behavior for the relation between the least required data sizes and transmittance.

We also compare our method with the state-of-the-art EAT analysis, which utilizes second-order statistics to improve the finite-size performance~\cite{liu2021device}. Using the variance information of the distributions, the advanced EAT method drastically reduces the data size. With this advanced method, recently, a full DIQKD experiment was achieved on an ion-trap platform~\cite{nadlinger2022experimental}. Under an average Bell value of $\bar{S}=2.64$ and a quantum bit error rate of $e_b=1.8\%$, the experiment produces the first secure key bit after approximately $10^6$ rounds of measurements (see Fig.~4 in the cited reference). The correspondence between this method and the current complementarity-based method is less clear, as we have not utilized the variance information. Even without this technique, the complementarity-based method still exhibits better finite-size performance. The least data required for key generation could be shortened to approximately $3.16\times10^5$ rounds, less than one-third of that required by the EAT method.

\begin{figure}[hbt!]
\centering
{\begin{overpic}[width=1\hsize]{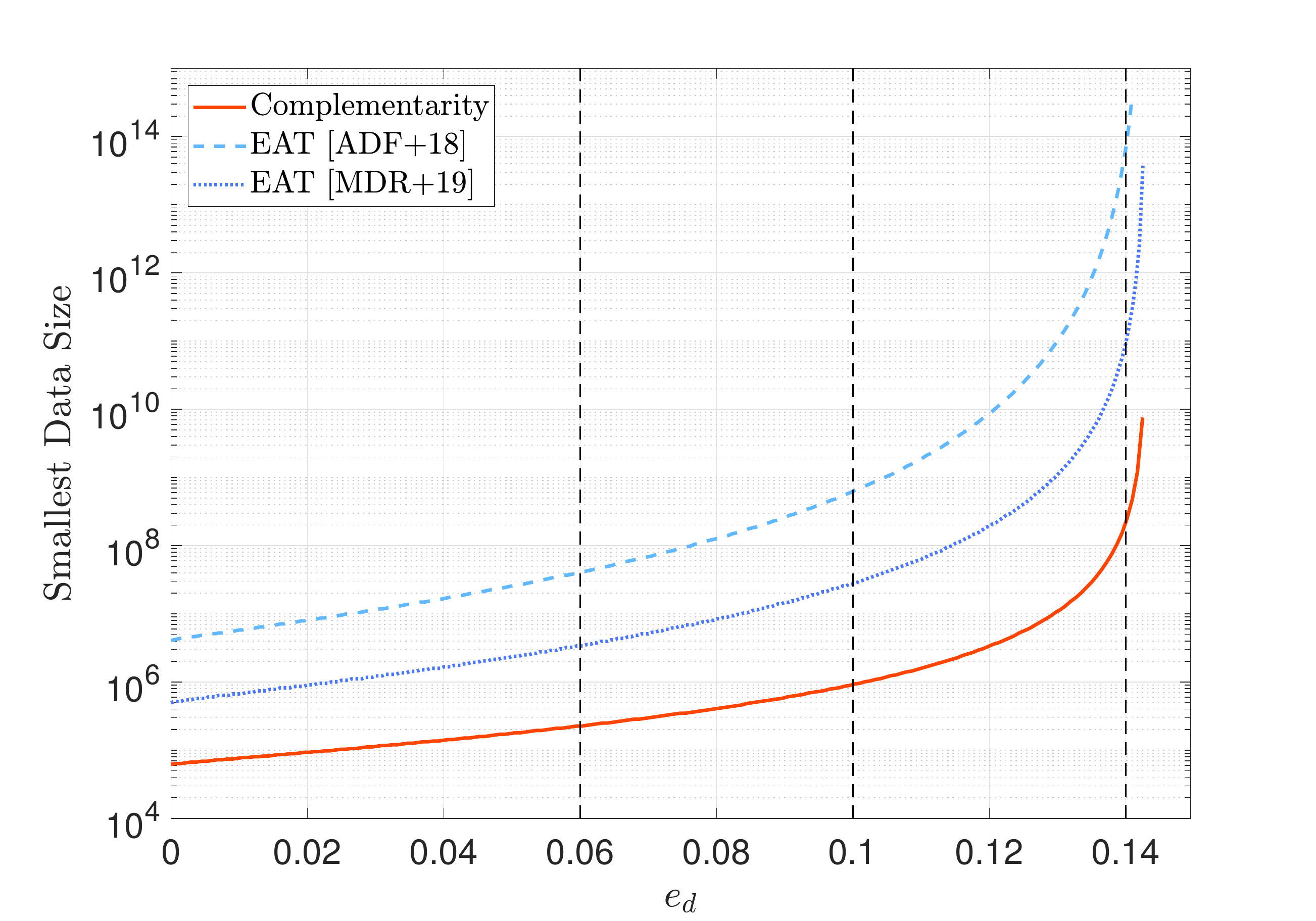}
\put(22,27)
{\fontsize{7}{7}\selectfont
\begin{tabular}{c}
$n=4.09\times10^{7}$
\end{tabular}
}
\put(22,21)
{\fontsize{7}{7}\selectfont
\begin{tabular}{c}
$n=3.39\times10^{6}$
\end{tabular}
}
\put(22,15)
{\fontsize{7}{7}\selectfont
\begin{tabular}{c}
$n=2.26\times10^{5}$
\end{tabular}
}
\put(43.5,34)
{\fontsize{7}{7}\selectfont
\begin{tabular}{c}
$n=6.16\times10^{8}$
\end{tabular}
}
\put(43.5,26)
{\fontsize{7}{7}\selectfont
\begin{tabular}{c}
$n=2.75\times10^{7}$
\end{tabular}
}
\put(43.5,18)
{\fontsize{7}{7}\selectfont
\begin{tabular}{c}
$n=9.12\times10^{5}$
\end{tabular}
}
\put(64,58)
{\fontsize{7}{7}\selectfont
\begin{tabular}{c}
$n=6.80\times10^{13}$ \\
\end{tabular}
}
\put(64,44)
{\fontsize{7}{7}\selectfont
\begin{tabular}{c}
$n=9.29\times10^{10}$ \\
\end{tabular}
}
\put(64,31)
{\fontsize{7}{7}\selectfont
\begin{tabular}{c}
$n=2.18\times10^{8}$ \\
\end{tabular}
}
\end{overpic}
}
\caption{The smallest data sizes required for successful key generation with respect to different depolarising factors at $\eta=1$.}
\label{Fig:Simulation1LOCC}
\end{figure}
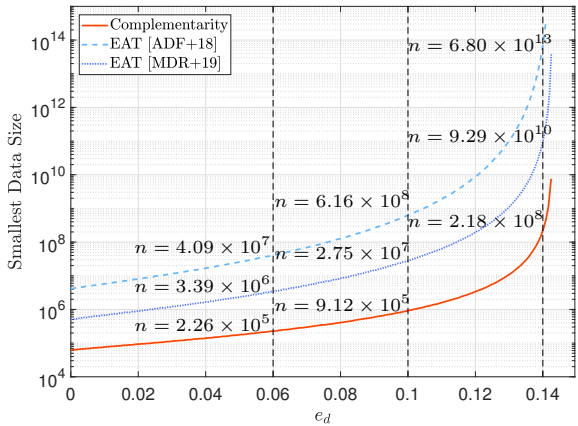

One main reason for the advantage in the finite-size regime is that our approach has a tighter bound on the information leakage of classical postprocessing. In the EAT-type analyses, the entropy that reflects the key secrecy is accumulated round by round in the experiment via a Markov chain~\cite{arnon2018practical,liu2021device}. The classical communication for Bell-value evaluation also contributes to the entropy accumulation. Consequently, the users need to cost extra secret keys to compensate for the information leakage. On the contrary, our method bypasses the above issue. Along the complementarity approach, the users can freely transmit classical information for parameter estimation as long as it does not affect the phase-error cardinality in key generation rounds. In fact, the information leakage in test rounds is not so pessimistic as in the EAT analysis. Without the additional consumption of private keys, the users can expect positive key generation in a shorter time. In Appendix~\ref{Sec:NumericalDiscussion}, we provide the simulation model, more numerical results, and more discussions comparing the two approaches.

\emph{Advantage Key Distillation. ---}Due to the operational meaning of measurement complementarity, we can employ three advantage key distillation methods, namely, noisy pre-processing (N) \cite{kraus2005lower}, usage of loss information (L) \cite{ma2012improved}, and B-step via two-way classical communication (B) \cite{gottesman2003proof,chau2002practical}, to DIQKD without the i.i.d.~assumption. We also consider the protocol using both noise pre-processing and loss information (C). Their improvement over the standard protocol (S) can be explicitly traced by the change in phase-error cardinality. These advantage methods can effectively increase the tolerable noise, reduce the minimum requirement on state fidelity, and lower the threshold of total transmittance, as shown in Table~\ref{tab:ThresholdEff}. Under the assumption of i.i.d.~attacks, the methods have been considered in the literature. We obtain the same threshold values for the methods of (L)~\cite{ma2012improved} and (N)~\cite{ho2020noisy}, and a similar result for (B) to the usage of repetition code~\cite{tan2020advantage}, where the reported threshold transmittance is $89.1\%$. We now strengthen the results as they become applicable to the most general attacks. We present the details for the advantage key distillation methods in Appendix~\ref{Supp:AdvKeyDistill} and simulation results in Appendix~\ref{Sec:AdvNum}.

\begin{table}[htb!]
\centering \caption{Thresholds of transmittance and fidelity for different postprocessing. We calculate the threshold transmittance with the fidelity set to unity and vice versa.}
\begin{tabular}{ccc}
\hline Protocol & Threshold Transmittance & Threshold Fidelity \\
\hline
(S) & $92.36\%$ & $89.28\%$\\
(N) & $91.27\%$ & $87.88\%$\\
(L) & $90.87\%$ & $89.28\%$\\
(C) & $90.47\%$ & $87.88\%$\\
(B) & $88.30\%$ & $88.28\%$\\
\hline
\end{tabular}
\label{tab:ThresholdEff}
\end{table}

\emph{Outlook. ---}Based on the complementarity approach, our security analysis explicitly shows the role of measurement complementarity in device-independent quantum cryptography. We focus on the sequential DIQKD protocol using the CHSH test. A natural extension is the parallel setting, where the legitimate users input all the measurement settings at once and then obtain all measurement results~\cite{jain2020parallel}. In practice, it would also be meaningful to explore whether Bell tests other than the CHSH test are more advantageous under certain realistic conditions~\cite{eberhard1993background}. In addition to the cryptography tasks, complementarity analysis could be immediately used for device-independent distillable entanglement quantification. In privacy estimation, we generalize the concept of sample entropy from classical Shannon theory. To the best of our knowledge, this is the first time that sample entropy has been used for QKD parameter estimation. We believe that the sample entropy will find more use in other areas of quantum science.

\emph{Acknowledgments. ---}We acknowledge Hongyi~Zhou, Go~Kato, Rong~Wang, and Chenyang~Li for the insightful discussions on concentration inequalities, Toyoshiro~Tsurumaru and Charles~L.-C.~Wen for discussions on related works, Ernest~Tan for discussions on the memory effect, Yizhi~Huang for independently checking the numerical results, Zhiqiang~Liu and Zhenhuan~Liu for contributions during early stages of this project, and Guoding~Liu and Junjie Chen for the careful proofread. This work is supported by the National Natural Science Foundation of China Grant No.~11875173, the National Key Research and Development Program of China Grant No.~2019QY0702 and No.~2017YFA0303903, NSERC, CFI, ORF, Connaught Innovation Award, MITACS, the Royal Bank of Canada, Huawei Technologies Canada, Inc.~and the University of Hong Kong start-up funding.

\onecolumngrid
\newpage

\appendix
\begin{appendices}
\tableofcontents
\newpage

\section{Notations}
\begin{table}[htbp!]
\caption{Frequently used parameters, random variables, and quantum operators in this work.}
\hspace{0.5cm}
\centering
\begin{tabular} {p{100pt}|p{220pt}|p{160pt}}
\hline
Notation & Meaning & Relation to Other Quantities \\

\hline

$n$ & number of rounds of quantum measurements & \\
$[n]$ & $\{1,\cdots,n\}$ & \\
$\mathbb{N}$ & set of natural numbers (including $0$) & \\
$\mathbb{C}$ & set of complex numbers & \\
$\mathcal{L}(\mathcal{H})$ & set of linear operators defined on the Hilbert space $\mathcal{H}$ \\
$\mathcal{D}(\mathcal{H})$ & set of density operators defined on the Hilbert space $\mathcal{H}$ \\

$f_{ec}$ & efficiency in information reconciliation & \\

$x_i\in\{0,1,2\}$ & Alice's input random variable in the $i$th round & \\
$y_i\in\{0,1\}$ & Bob's input random variable in the $i$th round & \\
$a_i\in\{1,-1\}$ & Alice's output random variable in the $i$th round & \\
$b_i\in\{1,-1\}$ & Bob's output random variable in the $i$th round & \\
$p_X$ & Alice's input probability distribution & \\
$p_Y$ & Bob's input probability distribution & \\
$\lambda_i$ & system random variable (value) of the $i$th round & \\

$\kappa^A$ & Alice's raw key bit string & \\
$\kappa^B$ & Bob's raw key bit string & \\

$\hat{\Pi}_{\lambda_i}^A$ & projector onto the subspace $\mathcal{H}_A^{\lambda_i}$ & \\
$\hat{\Pi}_{\lambda_i}^B$ & projector onto the subspace $\mathcal{H}_B^{\lambda_i}$ & \\
$\rho_{\lambda_i}^{AB}\in\mathcal{D}(\mathcal{H}_A^{\lambda_i}\otimes\mathcal{H}_B^{\lambda_i})$ & underlying subsystem in the $i$th round corresponding to the system random variable value $\lambda_i$ \\
$\hat{A}_x^{\lambda_i}\in\mathcal{L}(\mathcal{H}_A^{\lambda_i}),x\in\{0,1\}$ & observable in the $i$th round corresponding to Alice's input $x$ and system random variable value $\lambda_i$ & \\
$\hat{A}_2^{\lambda_i}\in\mathcal{L}(\mathcal{H}_{A(i)})$ & observable in the $i$th round corresponding to Alice's input $x=2$ & \\
$\hat{B}_y^{\lambda_i}\in\mathcal{L}(\mathcal{H}_B^{\lambda_i}),y\in\{0,1\}$ & observable in the $i$th round corresponding to Bob's input $y$ and system random variable value $\lambda_i$ & \\
$\hat{A}_{\theta_i}^{\lambda_i}\in\mathcal{L}(\mathcal{H}_A^{\lambda_i})$ & auxiliary observable on Alice's side in the $i$th round corresponding to the system random variable value $\lambda_i$ \\
$\hat{B}_{\bot}^{\lambda_i}$ & complementary observable of $\hat{B}_0^{\lambda_i}\in\mathcal{L}(\mathcal{H}_B^{\lambda_i})$ & $\{\hat{B}_{\bot}^{\lambda_i},\hat{B}_0^{\lambda_i}\}=0$ \\
$\hat{S}^{\lambda_i}$ & Bell observable in the $i$th round corresponding to the system random variable value $\lambda_i$ & $\hat{S}^{\lambda_i} = \sum_{x,y\in\{0,1\}}(-1)^{xy}\hat{A}_{x}^{\lambda_i}\otimes \hat{B}_{y}^{\lambda_i}$ \\
$S^{\lambda_i}$ & expected Bell value in the $i$th round corresponding to the system random variable value $\lambda_i$ & $S^{\lambda_i}=\Tr(\rho_{\lambda_i}^{AB}\hat{S}^{\lambda_i})$ \\
$E_{p}^{\lambda_i}\in\{0,1\}$ & phase-error random variable in the $i$th round corresponding to the system random variable value $\lambda_i$ & \\
$e_{p}^{\lambda_i}\in[0,\frac{1}{2}]$ & phase-error probability in the $i$th round corresponding to the system random variable value $\lambda_i$ & \\

$\Delta_{p}^{(i)}$ & martingale random variable constructed from phase-error random variables in the $i$th round & Eq.~\eqref{eq:DeltaRV}, see also the definition of $\zeta_{p}^{(i)}$ in Eq.~\eqref{Eq:RelaxSampleEntropy} \\
$\Delta_{xy}^{(i)},x,y\in\{0,1\}$ & martingale random variable with respect to the measurement setting $(x,y)$ constructed from Bell test random variables in the $i$th round & Eq.~\eqref{eq:DeltaBellRV}, see also the definition of $\zeta_{xy}^{(i)}$ in Eq.~\eqref{Eq:TestRoundsRV} \\

$\xi\in(0,\frac{1}{2})$ & regularisation parameter in the regularised sample entropy & \\

$\varepsilon_{pe}\in(0,1)$  & failure probability in the phase-error sample entropy estimation & \\
$\varepsilon_{xy}\in(0,1)$ & failure probability in the parameter estimation of the Bell test setting $(x,y)$ & \\
$\varepsilon_{S}\in(0,1)$ & failure probability in the parameter estimation of the Bell value & $\varepsilon_{S}=\sum_{x,y=0,1}\varepsilon_{xy}$ \\
$\varepsilon_{pc}\in(0,1)$ & failure probability in privacy amplification & \\
\end{tabular}
\label{tab:notation}
\end{table}

Before we commence, we explain the notations in this work. Some frequently used notations are listed in Table~\ref{tab:notation}.
\begin{enumerate}
  \item The letters $A,B,E$ are used to specify the quantum systems of Alice, Bob (legitimate users), and Eve (adversarial party). We denote the Hilbert spaces that the quantum operators act on as $\mathcal{H}$ and use subscripts and superscripts to specify subsystems.
  \item Except in Sec.~\ref{Supp:FiniteSize}, for simplicity, we usually express random variables by lowercase English or Greek letters. In most cases, we do not express the underlying probabilistic spaces explicitly. We shall make it clear if lowercase letters represent abstract but specific values or \emph{realisations} of random variables. Nevertheless, to keep in accordance with convention, we shall always use $E_p\in\{0,1\}$ to represent the phase-error random variable, $e_p$ to represent the phase-error probability and $S$ to represent the Bell value. In cases where the realisation of $E_p$ matters, we shall explicitly express it. For instance, we denote $E_p=1$ as the occurrence of a phase error.

      In Sec.~\ref{Supp:FiniteSize}, we develop a martingale-based finite-size analysis. For this part, we shall rigorously distinguish the concepts of random variables, the sample space of a random variable, and realisations of a random variable. For English letters (except $E_p$ and $S$), we use the expression in the form $X=x\in\mathcal{X}$ to represent the following meaning: a random variable, of which the sample space is $\mathcal{X}$, takes the realisation $x$ in an experiment. Whether the letters $A,B,E$ represent random variables or systems should be clear from the context. We use a similar expression for $\Lambda=\lambda$, where the random variable $\Lambda$ represents the system variable in this paper. For other Greek letters that are related to the martingale analysis, we shall explicitly explain whether they shall represent a random variable or a specific realisation of it.
  \item For quantum operators, $\rho$ is kept for the density matrix of the state of a quantum system. Quantum channels or measurements letters acting on a quantum system are represented with capital letters with a hat notation.
  \item If there is a group of quantum operators or random variables of similar meaning, we shall use their subscripts and (or) superscripts to specify them. A common case is to specify the round number. For example, $\rho_i$ represents the quantum state in the $i$th round. If the round number needs to be displayed together with other notations, we shall denote the round number by the subscript with a bracket, $(i)$, in case of ambiguity. If the round number value can be embedded in another random variable's value without confusion, we shall omit the subscript $(i)$. For example, in the notation $\rho_{\lambda_i}$, the round number is incorporated into the notation $\lambda$. In this case, $\rho$ needs to be specified by both $\lambda$ and $i$, and $i$ is also used for the specification of $\lambda$.
  \item There are notations that do not follow the above explanation. They will be explicitly explained when they appear for the first time.
\end{enumerate}

In addition, we list the representations of the Pauli operators used in this work,
\begin{equation}
\hat{\sigma}_x = \left(\begin{matrix}
0 & 1 \\
1 & 0
\end{matrix}\right),
\hat{\sigma}_y = \left(\begin{matrix}
0 & -i \\
i & 0
\end{matrix}\right),
\hat{\sigma}_z = \left(\begin{matrix}
1 & 0 \\
0 & -1
\end{matrix}\right),
\end{equation}
where we represent them on the computational basis determined by the eigenvectors of $\hat{\sigma}_z$, and $i^2=-1$. We use $H(\cdot)$ to represent the Shannon entropy function or the von Neumann entropy function. If the argument is a random variable, the Shannon entropy is given by
\begin{equation}
  H(X)=-\sum_{x\in\mathcal{X}}p_X(x)\log[p_X(x)],
\end{equation}
where $X\in\mathcal{X}$ denotes the random variable, and $p_X$ denotes its probability mass function. For $p_X(x)=0$, we take the convention of $0\log 0\equiv0$. If the argument is a density operator, i.e., a positive semi-definite operator with a unit trace, the von Neumann entropy is given by
\begin{equation}
  H(\rho)=-\sum_{i\in\mathcal{I}}\lambda_i\log(\lambda_i),
\end{equation}
where $\{\lambda_i\}_{i\in\mathcal{I}}$ are the eigenvalues of $\rho$. Throughout this work, the logarithm function $\log$ is assumed to be base 2 and $\ln$ is the logarithm function with base $e$. For a binary random variable, the Shannon entropy becomes the binary entropy, which is defined as
\begin{equation}
h(x)=-x\log x-(1-x)\log(1-x),
\end{equation}
with $x\in[0,1]$.

\section{Protocol and Security Statements}\label{Supp:SecDef}
In establishing the security analysis, we consider the device-independent quantum key distribution (DIQKD) protocol given in Box~\ref{box:ProtocolDetail}. The protocol was originally proposed in Ref.~\cite{acin2006efficient}. For simplicity, we denote this protocol as the `standard protocol'. Later in Sec.~\ref{Supp:AdvKeyDistill}, we shall consider several modified protocols with advantage key distillation based on the standard protocol.

\begin{mybox}[label={box:ProtocolDetail}]{{DIQKD Protocol (Standard Protocol)}}
\textbf{Arguments: }\\
$n$: the number of rounds of quantum measurements\\
$\{p_X(x)\}_x,\,\{p_Y(y)\}_y$: probability distributions for the basis choices

\tcblower
\begin{enumerate}[leftmargin=*]
\item\emph{Measurement: }For every round $i\in[n]$, Alice and Bob
    \begin{enumerate}[leftmargin=*]
    \item randomly set measurement bases $x_i\in \{0,1,2\},y_i\in \{0,1\}$ with probability distributions $p_X,p_Y$;
    \item record the outputs of the devices $a_i \in \{\pm1\},b_i \in \{\pm1\}$.
    \end{enumerate}
\item \emph{Parameter Acquisition: }Alice and Bob obtain the following parameters in the experiment
    \begin{enumerate}[leftmargin=*]
    \item the number of rounds $m_{xy}$ with inputs $(x_i,y_i) = (x,y)\in\{0,1\}\times\{0,1\}$ and $m$ with $(x_i,y_i) = (2,0)$;
    \item for $(x_i,y_i)=(2,0)$, the number of positions $q$ where $a_i\neq b_i$;
    \item for $(x_i,y_i) = (x,y)\in\{0,1\}\times\{0,1\}$, the number of positions $q_{xy}$ where $(-1)^{x\cdot y}a_i\cdot b_i=-1$.
    \end{enumerate}
\item \emph{Raw Key Acquisition: }Alice and Bob obtain the bit strings $\kappa^A,\kappa^B$ from $a_i,b_i$ with $(x_i,y_i)=(2,0)$.
\item \emph{Privacy Estimation: }Alice and Bob estimate the key privacy with statistics acquired in Step 2, (c).
\item \emph{Key Distillation: }Alice and Bob
    \begin{enumerate}[leftmargin=*]
    \item reconcile the bit strings to $\kappa^B$ through an encrypted classical channel except for failure probability $\varepsilon_{ec}$;
    \item  apply a privacy amplification to $\kappa^B$ and derive the final key except for failure probability $\varepsilon_{pc}$.
    \end{enumerate}
\end{enumerate}
\end{mybox}

Referring to the notations in Theorem~\ref{Thm:Main} of the main text, the bit error rate and the observed average Bell value are given as
\begin{equation}\label{Eq:biterrorrate}
\begin{aligned}
e_b&=\frac{f(a\neq b,x=2,y=0)}{p_{X}(2)p_{Y}(0)},
\end{aligned}
\end{equation}
\begin{equation}\label{eq:Sexpobs}
  \bar{S} = \sum_{x,y\in\{0,1\}}\frac{m_{xy}-2q_{xy}}{np_X(x)p_Y(y)}.
\end{equation}
where $f(\cdot)$ is the frequency of the realisation.

We clarify the notions \emph{round} and \emph{session} in this work. When referring to a round, we mean one complete quantum measurement in Step 1. A round starts when Alice and Bob set the measurement bases and ends when their devices generate outputs. A complete run of the protocol is referred to as a session. A session ends either Alice and Bob abort in advance, or they obtain a final key.

In the security analysis, we consider the most general sequential condition. That is, in Step 1 of the protocol, Alice and Bob shall not start the $(i+1)$th round until they have finished the $i$th round. The measurement results might be dependent on the previous events. In particular, the probabilities of the outputs might not be independent and identically distributed (i.i.d.). The inner working of the quantum devices is not characterised in prior, except for the following assumptions/requirements:
\begin{enumerate}
\item
\emph{Quantumness: }Quantum mechanics is correct.
\item
\emph{Non-Signalling: }No information is communicated between Alice and Bob's devices in a round.
\item
\emph{Trusted Inputs: }The users have trusted local private randomness for the bases choices.
\item
\emph{Reliable Classical Module: }The users have reliable key management and authenticated channels for classical communication.
\end{enumerate}


We impose the assumption of reliable key management to prevent the devices from leaking the key information to an adversary. To ensure this, we assume that except for the public classical information exchange necessary for the classical phase, the devices of Alice and Bob signal no other information to the outside. In practice, though, this assumption would technically guarantee the second assumption, as the inputs and outputs on the one side can be kept away from the other side in the same way they are kept secret from the outside. Still, we can relax the constraint on the communication between Alice and Bob's devices in the security analysis, where the devices are allowed to communicate with each other between two adjacent rounds.

In this work, we consider that the devices are only used for one complete session. If the devices are reused for multiple sessions of DIQKD, the \emph{memory loophole} might exist. That is, the untrusted devices may store the key information in one DIQKD session and leak it via the necessary public communication required by the protocol when they are reused in new sessions~\cite{barrett2013memory}. There exist practical defences against the memory attacks~\cite{barrett2013memory,curty2019foiling}. In this work, we do not discuss this point. It should be noted that in a single session, the memory attack does not take place. In particular, the quantum measurements can be dependent on all possible information generated in the previous rounds on both sides. Our security analysis applies to this general condition.

Suppose after running the protocol, Alice and Bob obtain their $k$-bit key strings, $\tilde{\kappa}^A$ and $\kappa^B$, respectively. We use $\tilde{\kappa}^A$ to distinguish Alice's final key from her raw key $\kappa^A$. The final key state is described by a classical-classical-quantum state,
\begin{equation}
\rho_{ABE}=\sum_{\tilde{\kappa}^A,\kappa^B}\Pr(\tilde{\kappa}^A,\kappa^B)\ketbra{\tilde{\kappa}^A}\otimes\ketbra{\kappa^B}\otimes\rho_E(\tilde{\kappa}^A,\kappa^B),
\end{equation}
where system $E$ represents the quantum side information. From an adversarial perspective, the key bits of the legitimate users are characterised by random variables of the systems $A,B$. Ideally, the final key state should be
\begin{equation}
\rho_{ABE}^{ideal}=2^{-k}\sum_{\kappa}\ketbra{\kappa}_A\otimes\ketbra{\kappa}_B\otimes\rho_E,
\end{equation}
where Alice and Bob's keys are perfectly identical and appear fully random to the outside. The protocol security is defined as the difference between the actual key state and the ideal key state.

\begin{definition}[Security~\cite{ben2005universal,renner2005Security}]
A QKD protocol is called $\varepsilon_{QKD}^s$-secret, if for a successful implementation of the protocol,
    \begin{equation}
    \frac{1}{2}(1-p_{abort})\|\rho_{ABE}-\rho_{ABE}^{ideal}\|_1\leq\varepsilon_{QKD}^s,
    \end{equation}
    where $p_{abort}$ is the probability that the protocol aborts, and $\|A\|_1=\Tr(\sqrt{A^\dag A})$ is the trace norm.
\label{Def:SecurityDef}
\end{definition}

Here, the parameter $\varepsilon_{QKD}^s$ is also called the soundness error. The parameter $p_{abort}$ is linked with another security parameter, namely, the completeness error. We remark that soundness and completeness are different notions. Roughly speaking, it is the soundness error that corresponds to the privacy of generated keys. On the other hand, completeness guarantees that the legitimate users can obtain keys (possibly not private) with a high probability. In the context of QKD, the completeness issue can often be neglected. We shall come to this point later in Sec.~\ref{Supp:EAT}.

Generally, after a complete session, the net increase of secure key bits shall be given by
\begin{equation}
k = m - I_{ir} - I_{pa},
\end{equation}
where $I_{ir}$ and $I_{pa}$ are the costs of pre-shared private randomness, or secret keys, for information reconciliation and privacy amplification, respectively. Security analysis aims at the determination of their values. The two terms can be correlated with each other. We analyse the DIQKD security under the framework of complementarity-based quantum error correction protocol~\cite{koashi2009simple}. Under this framework, we aim at establishing a virtual protocol as follows:

\begin{mybox}[label={box:KoashiVirtual}]{{Virtual Protocol in Complementarity Approach~\cite{koashi2009simple}}}
Alice and Bob run the following procedures to achieve a final secure key.
\begin{enumerate}[leftmargin=*]
\item
Initialise the system in joint state $\rho^{AB}$;
		
\item
Apply operation $\hat{\Pi}$ on $\rho^{AB}$ such that $\hat{\Pi}(\rho^{AB})\in\mathcal{D}(\mathcal{K}^{\otimes m}\otimes\mathcal{H}_R)$, where $\dim{\mathcal{K}}=2$;
		
\item
Measure system $\mathcal{K}^{\otimes m}$ on the computational basis, denoted as $Z$-basis, and obtain $\kappa_{rec}$;
		
\item
Measure system $\mathcal{H}_R$ with $\hat{M}_R$ and obtain result $\gamma$, which determines $I_{pa}$;

\item
Apply hashing on $\kappa_{rec}$ and obtain the final key $\kappa_{fin}$ with length of $(m-I_{pa})$ bits.
\end{enumerate}
\end{mybox}

If we can establish such a virtual protocol that meets certain security requirements, then there exists a data post-processing procedure for the raw data in the real experiment that produces a secret key with a negligible failure probability. This result is given in Ref.~\cite{koashi2009simple}, which we summarise as follows~\footnote{There is an error in the security parameter in the original work of Ref.~\cite{koashi2009simple} due to the definition of fidelity measures. This is later corrected in Ref.~\cite{fung2010practical}}.

\begin{lemma}[Security statement~\cite{koashi2009simple,fung2010practical}]
With respect to the actual QKD protocol in Box~\ref{box:ProtocolDetail}, if one can set up a corresponding virtual protocol in Box~\ref{box:KoashiVirtual}, where the chosen operation $\hat{\Pi}$ and measurement $\hat{M}_R$ in the virtual protocol meet the following requirements:
\begin{enumerate}[label={(}\arabic*{)}]
\item
the $Z$-basis measurement statistics $\kappa_{rec}$ in the virtual protocol is the same as reconciled key $\kappa^B$ in the actual protocol;

\item
given each outcome $\gamma$ of measurement $\hat{M}_R$, had a measurement on the complementary basis of the $Z$ basis been taken on $\mathcal{K}^{\otimes m}$, the number of measurement outcomes can be upper bounded by $2^{T}$ except for failure probability $\varepsilon_{pe}$,
\end{enumerate}
then the actual protocol is $\sqrt{\varepsilon(2-\varepsilon)}$-secret and $\varepsilon_{ec}$-correct, where $\varepsilon=\varepsilon_{pe}+\varepsilon_{pc}$ and $I_{pa}=T-\log\varepsilon_{pc}$.
\label{Lemma:SecDef}
\end{lemma}

In this security statement, the failure probability has the operational meaning of state fidelity. In Step~3 of the virtual protocol, suppose before the $Z$-basis measurement, the joint state becomes $\ket{0}_X^{\otimes m}$, an eigenstate of an observable complementary to the $Z$-basis measurement. Then, the $Z$-basis measurement result shall be intrinsically random and hence private from the outside. In establishing the security statement~\cite{koashi2009simple} for a general case, the privacy amplification is recast into a quantum error correction protocol that distills such a state with local operations and classical communication (LOCC). Conditioned on $\gamma$, the complementary measurement gives one of the eigenstates on $X$-basis. If the number of possible outcome patterns can be upper-bounded as stated by the second requirement in Lemma~\ref{Lemma:SecDef}, then there exists an error correction protocol for Alice and Bob to correct the state to $\ket{0}_X^{\otimes m}$. For the purpose of error correction, a general linear error correction code suffices~\cite{lo2003method,huang2021stream}. In practice, privacy amplification can be efficiently carried out by using the family of two-universal hashing functions~\cite{fung2010practical}. Suppose in the virtual quantum error correction protocol, the output state is $\rho_{\mathcal{K}}$. It is proved that the state fidelity of $\rho_{\mathcal{K}}$ to $\ket{0}_X^{\otimes m}$ is linked to the failure probability of the post-processing procedures,
\begin{equation}\label{Eq:FidelitySec}
  \Tr(\rho_{\mathcal{K}}\ketbra{0}_X^{\otimes m})\geq 1-\varepsilon.
\end{equation}
In Ref.~\cite{fung2010practical}, it is proved that the security parameter given by the fidelity measure can be linked to the trace-distance-based security definition. Measuring the state $\rho_{\mathcal{K}}$ with key generation measurement, the key is $\sqrt{\varepsilon(2-\varepsilon)}$-sound by Definition~\ref{Def:SecurityDef}. In the virtual quantum error correction protocol, the users need to communicate with each other the information of quantum error syndromes. They consume private randomness for this step, and the amount of private randomness corresponds to the privacy amplification cost, $I_{pa}$.

To simplify the statements in the complementarity-based security analysis, we introduce the notion of phase errors. We describe the measurement of complementary observable given $\gamma$ as the phase-error measurement and denote a specific outcome as a phase-error pattern. With respect to the ideal case $\ket{0}_X^{\otimes m}$, if the phase-error measurement result on the $i$th bit differs from $0$, then we say that a phase error occurs on this bit. Given an underlying system, there exists a smallest set of phase-error patterns that will occur except for a negligible failure probability. Denote the number of these patterns as the phase-error cardinality under the given failure probability. The second requirement in Lemma~\ref{Lemma:SecDef} represents an upper bound on the phase-error cardinality.

Another issue in this security statement is that information reconciliation and privacy amplification can be decoupled. This can be done by encrypting the error syndromes in classical communication of information reconciliation~\cite{lo2003method}. Information reconciliation, essentially the problem of key agreement, can be well solved efficiently via classical means~\cite{brassard1993secret}. At the moment, we assume that Alice and Bob apply one-way information reconciliation, and the outputs of $a_i,b_i$ in the key generation rounds are taken as the raw key bit strings $\kappa^A,\kappa^B$. Bob hashes his bit string and sends the hashing results to Alice through an encrypted classical channel. In this step, $I_{ir}$ bits of pre-shared keys are consumed for encryption. Based on the received message, Alice can reconcile her bit string $\kappa^A$ to $\kappa^B$. We denote the reconciled key length as
\begin{equation}
I_{AB} = m - I_{ir},
\end{equation}
which is the net increase of identical bits after the information reconciliation. In the Shannon limit, the cost of private randomness is
\begin{equation}
\lim_{n\rightarrow\infty} I_{ir} = H(\kappa^B|\kappa^A).
\end{equation}
In principle, the Shannon limit can be reached in the asymptotic limit of infinite data size with proper design of encoding and decoding schemes. In practice, the cost becomes
\begin{equation}\label{Supp:InfoRec}
I_{ir} = f_{ec}H(\kappa^B|\kappa^A),
\end{equation}
where $f_{ec}\geq1$ is an efficiency factor depending on the specific encryption code that is used. Much effort has been made in designing practical error-correcting codes and schemes for encoding and decoding. Except for the discussion in Section~\ref{Sec:IonSim}, in the numerical simulations in this work, for simplicity, we take the Shannon limit where $f_{ec}\rightarrow1$. For a more practical discussion, we would like to refer the readers to the recent state-of-the-art results~\cite{tang2021shannon}. If an error occurs with same probability for the two bit values, we have $H(\kappa^B|\kappa^A) = mh(e_b)$, where $e_b\equiv q/m$ is the bit error rate, which is the average difference between Alice and Bob's raw key bits. We note that this symmetric case can always be met, as Alice and Bob can have a public discussion and flip some of their raw bits. Nevertheless, this would result in the loss of the amount of reconciled key bits.

\section{Outline of Security Analysis}\label{Supp:SecOutline}
The evaluation of the privacy amplification cost, $I_{pa}$, is the essence of security analysis. To apply the complementarity framework to DIQKD, we shall
\begin{enumerate}
\item define phase-error measurement in the device-independent scenario;
\item estimate the phase-error cardinality via the observed statistics.
\end{enumerate}

The solutions to the two problems provide the ingredients for the virtual protocol in Box~\ref{box:KoashiVirtual}. The solution to the first problem gives the construction of $\hat{\Pi}$ and $\hat{M}_R$. To tackle the problem, we take advantage of Jordan's lemma and effectively reduce the quantum systems in each round to qubits~\cite{pironio2009device}. Then, we can apply the mature techniques in regular QKD to define phase-error measurement.

To solve the second problem, we adopt the idea of breaking the general non-i.i.d.~correlation into the analysis of single rounds. We divide the task into the following steps,
\begin{enumerate}[label=(\alph*)]
\item relating the phase-error probability with the expected Bell value in each round;
\item finding a measure of the phase-error cardinality that relates to the phase-error probability;
\item extending the single-round result to multiple rounds and deriving the estimation of the phase-error cardinality.
\end{enumerate}

For the first step, we derive an analytical relation between the expected values. For the second step, one might naively use the phase-error rate, a well-studied and commonly-used privacy measure in regular QKD. Unfortunately, as we shall explain later, this is not a good measure in DIQKD. Instead, we generalise the concept of sample entropy from classical Shannon theory: we prove that this provides a good measure of the phase-error cardinality. For the third step, we show that the sequential setting allows for a martingale-based analysis. In this step, the random basis setting in DIQKD plays a vital role. By combining the results, we obtain an analytical upper bound on the cost of privacy amplification in terms of the observed average Bell value.
In Figure \ref{Fig:flowchart} we demonstrate the key technical steps towards the key privacy estimation.


\begin{figure*}[tbh]
\centering
\resizebox{12cm}{!}{\includegraphics{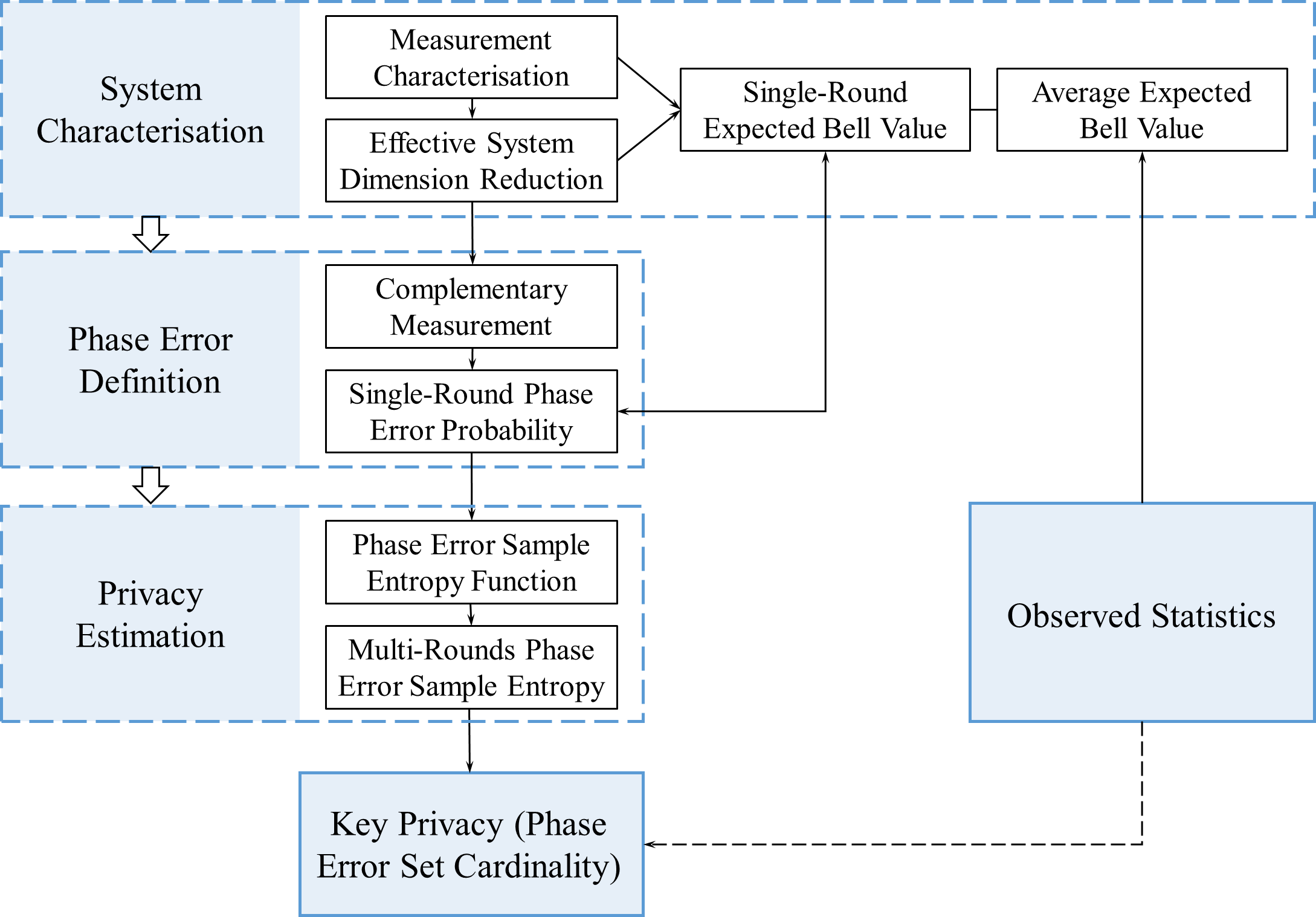}}
\caption{The flowchart of the security analysis steps in privacy estimation.
(1) System characterisation: In the first step, the quantum system needs to be defined. Due to the simplicity of our protocol, we can analyse the problem with an effective qubit system for each single round. We shall also model the possible correlations over rounds and define the underlying physical quantities, especially the expected Bell value.
(2) Phase error definition: Based on the system characterisation results, the measurement complementary to the key generation measurement is defined. This allows us to introduce a virtual protocol to estimate the key privacy. In the virtual protocol, the phase-error operation is defined, and the probability a phase error occurs is linked with the expected Bell value.
(3) Privacy estimation: We estimate the key privacy by upper bounding the number of possible phase-error patterns over multiple rounds, or, the cardinality of the phase error set. We construct the statistics of phase-error sample entropy for this task. As the random variables might not be independent and identically distributed, Azuma's inequality is used to link the expected values and probabilities with frequencies.
}
\label{Fig:flowchart}
\end{figure*}

\section{System Modelling}\label{Supp:SysModel}
In this section, we model the quantum system and measurement in the sequential DIQKD protocol. Also, we briefly review a few standard techniques developed in the field of quantum self-testing that are used in our security analysis.

\subsection{General sequential setting}\label{Supp:SequentialSetting}
In DIQKD, the legitimate users do not have access to the internal working of the uncharacterised devices but only the input-output statistics accumulated from the outside. Under the framework of quantum mechanics, the underlying physical processes can be described by quantum measurements. We analyse the security under the most general sequential measurement setting. Without loss of generality, DIQKD can be described by a sequential quantum process in Figure \ref{Fig:sequential}. In an adversarial picture, the quantum systems and measurement devices are provided by a potentially malicious party, Eve. After the protocol begins, Eve no longer has access to Alice's or Bob's devices. Nevertheless, she could prepare her quantum system correlated with Alice and Bob's joint system in prior. In the worst-case scenario, Eve could hold the purification of the initial state of Alice and Bob's joint system before measurement.

In sequential DIQKD, the measurement process is restricted to a specific time order noted by rounds. The state and measurement of a round cannot depend on the inputs and outputs of future rounds. From the $i$th round to the $(i+1)$th round of DIQKD, Alice and Bob's system evolves as follows.
\begin{enumerate}
\item
At the beginning of the $i$th round, Alice and Bob's joint system is in a bipartite state, $\rho_{i}^{AB}\in\mathcal{D}(\mathcal{H}_{A_i}\otimes\mathcal{H}_{B_i})$. In general, state $\rho_{i}^{AB}$ depends on all information generated in the previous rounds, including both the data accessible to Alice and Bob, $\{x_{j},y_{j},a_{j},b_{j}\}_{j=1}^{i-1}$, and any other possible hidden information of the untrusted devices inaccessible to the users, together denoted by $\lambda_{i}$ with $\lambda_{1}$ denoting the initial hidden variable of the system.

\item
Depending on the basis choices, $x_i,y_i\in\{0,1\}$, Alice and Bob perform general measurements on systems $\mathcal{H}_{A_i}$ and $\mathcal{H}_{B_i}$, independently, and then obtain classical outputs $a_i,b_i\in\{-1,1\}$. Alice and Bob's quantum operators could also depend on $\lambda_{i}$. Since there is no restriction on the system dimensions, the measurement operators can be projections due to the Neumark extension \cite{neumark1940spectral}, denoted by two observables, $\hat{A}_{x_i}$ and $\hat{B}_{y_i}$. Due to the non-signalling condition, they take a tensor-product form, $\hat{A}_{x_i}\otimes\hat{B}_{y_i}$. This is the main assumption in the Bell test, often guaranteed by space-like separation.

\item
A general channel $\hat{T}_i$ is applied to the post-measurement state. The state of the system becomes $\rho_{i+1}^{AB}\in\mathcal{D}(\mathcal{H}_{A_{i+1}}\otimes\mathcal{H}_{B_{i+1}})$ and is passed on to the $(i+1)$th round. Between two adjacent rounds, we do not restrict Alice's and Bob's devices from communicating with each other and sharing past events, but not with Eve. As a result, measurement $\hat{A}_{x_i}\otimes\hat{B}_{y_i}$ and operation $\hat{T}_i$ together transform $\rho_i^{AB}$ to $\rho_{i+1}^{AB}$.
\end{enumerate}

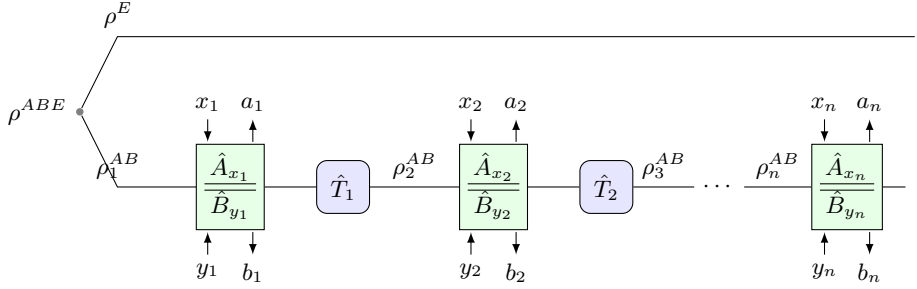
\begin{figure}[hbtp!]
	\begin{tikzpicture}[
		scale=1,
		>=latex
		]
		\node[circle,fill,gray,inner sep=1pt,label=left:$\rho^{ABE}$] (orig) at (0,0) {};
		\draw[-] (orig) -- (.5,1) node[right=.3,above]{$\rho^E$} -- ++(10.6,0);
		\draw[-] (orig) -- (.5,-1) node[right=.3,above]{$\rho^{AB}_1$} -- ++(1.5,0);
		\node[rectangle,draw=black,fill=green!10!white,minimum size=20pt] (op1) at (2,-1) {
			\begin{tabular}{c}
				$\hat{A}_{x_1}$ \\
				\hline
				\hline
				$\hat{B}_{y_1}$ \\
			\end{tabular}
		};
		\draw[->] ($ (op1) + (.3,.6) $) -- ++(0,0.3) node[above]{$a_1$};
		\draw[<-] ($ (op1) + (-.3,.6) $) -- ++(0,0.3) node[above]{$x_1$};
		\draw[->] ($ (op1) + (.3,-.6) $) -- ++(0,-0.3) node[below]{$b_1$};
		\draw[<-] ($ (op1) + (-.3,-.6) $) -- ++(0,-0.3) node[below]{$y_1$};
		\node[rectangle,draw=black,fill=blue!10!white,minimum size=20pt,rounded corners] (T1) at (3.5,-1) {$\hat{T}_1$};
		\draw[-] (op1.east) -- (T1.west) node[midway,above] {};
		\node[rectangle,draw=black,fill=green!10!white,minimum size=20pt] (op2) at (5.5,-1) {
			\begin{tabular}{c}
				$\hat{A}_{x_2}$ \\
				\hline
				\hline
				$\hat{B}_{y_2}$ \\
			\end{tabular}
		};
		\draw[->] ($ (op2) + (.3,.6) $) -- ++(0,0.3) node[above]{$a_2$};
		\draw[<-] ($ (op2) + (-.3,.6) $) -- ++(0,0.3) node[above]{$x_2$};
		\draw[->] ($ (op2) + (.3,-.6) $) -- ++(0,-0.3) node[below]{$b_2$};
		\draw[<-] ($ (op2) + (-.3,-.6) $) -- ++(0,-0.3) node[below]{$y_2$};
		\node[rectangle,draw=black,fill=blue!10!white,minimum size=20pt,rounded corners] (T2) at (7,-1) {$\hat{T}_2$};
		\draw[-] (T1.east) -- (op2.west) node[midway,above] {$\rho_{2}^{AB}$};
		\draw[-] (op2.east) -- (T2.west) node[midway,above] {};
		\node[] (dots) at (8.5,-1) {$\cdots$};
		\draw[-] (T2.east) -- (dots.west) node[midway,above] {$\rho_{3}^{AB}$};
		\node[rectangle,draw=black,fill=green!10!white,minimum size=20pt] (opn) at (10.2,-1) {
			\begin{tabular}{c}
				$\hat{A}_{x_n}$ \\
				\hline
				\hline
				$\hat{B}_{y_n}$ \\
			\end{tabular}
		};
		\draw[->] ($ (opn) + (.3,.6) $) -- ++(0,0.3) node[above]{$a_n$};
		\draw[<-] ($ (opn) + (-.3,.6) $) -- ++(0,0.3) node[above]{$x_n$};
		\draw[->] ($ (opn) + (.3,-.6) $) -- ++(0,-0.3) node[below]{$b_n$};
		\draw[<-] ($ (opn) + (-.3,-.6) $) -- ++(0,-0.3) node[below]{$y_n$};
		\draw[-] (dots.east) -- (opn.west) node[midway,above] {$\rho_{n}^{AB}$};
		\draw[-] (opn.east) -- ++(.3,0);
	\end{tikzpicture}
	\caption{System description of the sequential setting. Initially, the system of Alice and Bob is prepared in $\rho_1^{AB}\in\mathcal{D}(\mathcal{H}_A\otimes\mathcal{H}_B)$, which may be correlated or entangled with Eve's system, $\rho^E$. That is, $\rho_1^{AB}=\Tr_E(\rho^{ABE})$. In the $i$th round, the system of Alice and Bob to be measured is in the state of $\rho_i^{AB}$. Given the bases choices $x_i,y_i$ on each side, the measurement process is described by observables, $\hat{A}_{x_i}$ and $\hat{B}_{y_i}$, which yield the measurement results $a_i,b_i$. During the measurement process, no information can be communicated between Alice and Bob, as represented by the double lines. Between two adjacent rounds, communication is allowed between the measurement devices of Alice and Bob, characterised by a general quantum channel $\hat{T}_i$. Afterwards, the resulted state $\rho_{i+1}^{AB}$ becomes the input of the $(i+1)$th round.} \label{Fig:sequential}
\end{figure}

Now, we can make the criteria raised by the sequential setting mathematically rigorous. First, define the conditional probability,
\begin{equation}
	\begin{split}
\Pr(a_i,b_i|x_i,y_i,\lambda_{i}) &\equiv \Tr[\left(\hat{M}_{x_i}^{a_i}\otimes\hat{M}_{y_i}^{b_i}\right)\rho_i^{AB}],
	\end{split}
\end{equation}
where $\lambda_{i}$ is the hidden variable, $\hat{M}_{x_i}^{a_i}$ and $\hat{M}_{y_i}^{b_i}$ are the projectors of  observables $\hat{A}_{x_i}$ and $\hat{B}_{y_i}$ with specific outcomes, respectively. The probability distributions satisfy the non-signaling condition, $\forall x_i,x_i',y_i,y_i',a_i,b_i,\lambda_{i}$,
\begin{equation}
	\begin{split}
\sum_{b_i\in\{\pm1\}}\Pr(a_i,b_i|x_i,y_i,\lambda_{i}) &= \sum_{b_i\in\{\pm1\}}\Pr(a_i,b_i|x_i,y_i',\lambda_{i}), \\
\sum_{a_i\in\{\pm1\}}\Pr(a_i,b_i|x_i,y_i,\lambda_{i}) &= \sum_{a_i\in\{\pm1\}}\Pr(a_i,b_i|x_i',y_i,\lambda_{i}).
	\end{split}
\end{equation}
The sequential setting imposes further restrictions, such that for all $j>i$ and any given values of $x_i,y_i,a_i,b_i,\lambda_{i}$ at time $i$ and $x_j,y_j,x_j',y_j'$ at time $j$,
\begin{equation}
\begin{split}
\sum_{a_j,b_j}\Pr(a_i,b_i,a_j,b_j|x_i,y_i,x_j,y_j,\lambda_{i}) = \sum_{a_j,b_j}\Pr(a_i,b_i,a_j,b_j|x_i,y_i,x_j',y_j',\lambda_{i}).
\end{split}
\end{equation}

For the convenience of constructing phase-error measurement, we present an equivalent description of the above sequential processing. The $n$ experimental rounds correspond to $2n$ subsystems of Alice and Bob. Initially, these subsystems can be arbitrarily entangled. There is no limit on their dimensions. The measurements, $\hat{A}_{x_i}\otimes\hat{B}_{y_i}$, in the $i$th round are carried out on Alice's and Bob's $i$th subsystems. In the equivalent picture, as shown in Figure \ref{Fig:EffPar}, Alice and Bob's system evolves from the $i$th round to the $(i+1)$th round as follows.
\begin{enumerate}
\item
At the beginning of the $i$th round, the joint state of Alice's and Bob's subsystems is  $\rho_{i}^{AB}\in\mathcal{D}(\mathcal{H}_{A_i}\otimes\mathcal{H}_{B_i})$. Here, we use the same notations as the process of Figure \ref{Fig:sequential}, but we should understand $A_i$ and $B_i$ as subsystems.
	
\item
Alice and Bob perform the measurement, $\hat{A}_{x_i}\otimes\hat{B}_{y_i}$, on $\rho_{i}^{AB}$ and obtain outputs $a_i,b_i\in\{-1,1\}$. This is the same the process of Figure \ref{Fig:sequential}.
	
\item
The post-measurement state is transferred to the $(i+1)$th subsystems of Alice and Bob, which is done locally. We describe this `transfer' operation by a hidden updating quantum operation $\hat{K}_i^A$ acting on the $(i+1)$th subsystem of Alice as shown in Figure \ref{Fig:EffPar}, and similarly on Bob's side. For instance, $\hat{K}_i^A$ can be realised via teleportation from the $i$th subsystem to the $(i+1)$th subsystem \cite{bennett1993teleporting}. After this step, the $i$th subsystems are discarded. In a way, one can think of relabeling the subsystems from $i$ to $i+1$.

\item
A general joint channel $\hat{T}_i$ is applied to the $i+1$ subsystems, after which the state becomes $\rho_{i+1}^{AB}\in\mathcal{D}(\mathcal{H}_{A_{i+1}}\otimes\mathcal{H}_{B_{i+1}})$ and is passed on to the $(i+1)$th round. This is the same the process of Figure \ref{Fig:sequential}.
\end{enumerate}

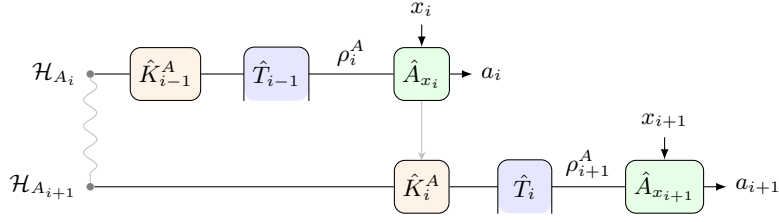
\begin{figure}[hbtp!]
	\begin{tikzpicture}[
		scale=1,
		>=latex
		]
		\node[circle,fill,gray,inner sep=1pt,label=left:$\mathcal{H}_{A_i}$] (orig) at (0,0) {};
		\node[rectangle,draw=black,fill=orange!10!white,minimum size=20pt,rounded corners] (K) at ($ (orig) + (1cm,0) $) {$\hat{K}_{i-1}^A$};
		\node[rectangle,fill=blue!10!white,minimum size=20pt,rounded corners] (TA) at ($ (K.east) + (1cm,0) $) {$\hat{T}_{i-1}$};
		\draw (TA.south west) {[rounded corners] -- (TA.north west)  -- (TA.north east) -- (TA.south east)};
		\node[rectangle,draw=black,fill=green!10!white,minimum size=20pt,rounded corners] (Ax) at ($ (TA.east) + (1.5cm,0) $) {$\hat{A}_{x_i}$};		
		\draw[<-] (Ax.north) --++ (0,0.3) node[above]{$x_i$};
		\draw[->] (Ax.east) --++ (.3,0) node[right] {$a_i$};
		\draw (orig) -- (K.west) (K.east) -- (TA.west) (TA.east) -- (Ax.west) node[midway,above] {$\rho_i^A$};
		\node[circle,fill,gray,inner sep=1pt,label=left:$\mathcal{H}_{A_{i+1}}$] (orig1) at (0,-1.5) {};
		\node[rectangle,draw=black,fill=orange!10!white,minimum size=20pt,rounded corners] (K1) at ($ (Ax) + (0,-1.5cm) $) {$\hat{K}_{i}^A$};
		\node[rectangle,fill=blue!10!white,minimum size=20pt,rounded corners] (TA1) at ($ (K1.east) + (1cm,0) $) {$\hat{T}_{i}$};
		\draw (TA1.south west) {[rounded corners] -- (TA1.north west)  -- (TA1.north east) -- (TA1.south east)};
		\node[rectangle,draw=black,fill=green!10!white,minimum size=20pt,rounded corners] (Ax1) at ($ (TA1.east) + (1.5cm,0) $) {$\hat{A}_{x_{i+1}}$};
		\draw[<-] (Ax1.north) --++ (0,0.3) node[above]{$x_{i+1}$};
		\draw[->] (Ax1.east) --++ (.3,0) node[right] {$a_{i+1}$};
		\draw (orig1) -- (K1.west) (K1.east) -- (TA1.west) (TA1.east) -- (Ax1.west) node[midway,above] {$\rho_{i+1}^A$};
		\draw[snake=snake,gray,opacity=.5] (orig) -- (orig1);
		\draw[-stealth,gray,opacity=.5] (Ax) -- (K1.north);
	\end{tikzpicture}
\caption{The $i$th and $(i+1)$th subsystems on Alice's side in a sequential manner. Bob's side is similar. For the $i$th subsystem of Alice, input state $\rho_i^{A}$ is measured by observable $\hat{A}_{x_i}$. Alice obtains classical output $a_i$. A hidden updating quantum operation $\hat{K}_i^A$ transfers the post-measurement state from subsystem $i$ to subsystem $(i+1)$. Afterwards, Alice and Bob's devices can communicate with each other to perform a joint operation $\hat{T}_i$, only half of which is shown in the figure. As a result, the state of the $(i+1)$th subsystem becomes $\rho_{i+1}^A$. We emphasise that before measurement, correlation or entanglement may already exist over the subsystems.} \label{Fig:EffPar}
\end{figure}

We remark that we do not impose any restriction on the initial joint state of Alice and Bob and operations $\hat{K}_i^A,\hat{T}_i$ in this description. The key point is that for any possible sequential system evolution in Figure~\ref{Fig:sequential}, one can find an equivalent description in the manner of Figure~\ref{Fig:EffPar}.

\subsection{Effective system dimension reduction}\label{Supp:EffDimReduction}
In the estimation of the raw key's privacy, we shall use the measurements $\hat{A}_0,\hat{A}_1$ and $\hat{B}_0,\hat{B}_1$ in each round. We can take advantage of Jordan's Lemma to effectively reduce the system dimension \cite{tsirelson1993some}. We restate the proof given in Ref.~\cite{pironio2009device}.

\begin{lemma}[Jordan's Lemma]\label{Lemma:DimRed}
For two Hermitian operators $\hat{A}_0,\hat{A}_1$ with eigenvalues $\pm1$ that act on the same Hilbert space $\mathcal{H}$ with an at most countable dimension, there exists a direct sum decomposition $\mathcal{H}=\bigoplus_\alpha\mathcal{H}^\alpha$ such that $\hat{A}_0,\hat{A}_1$ are simultaneously block diagonal, $\hat{A}_0 = \bigoplus_\alpha \hat{A}_0^\alpha,\,\hat{A}_1 = \bigoplus_\alpha \hat{A}_1^\alpha,\,\hat{A}_0^\alpha,\hat{A}_1^\alpha\in\mathcal{L}(\mathcal{H}^\alpha)$, and $\forall\alpha,\dim\mathcal{H}^\alpha\leq 2$.
\end{lemma}

\begin{proof}
Since the binary observables $\hat{A}_0,\hat{A}_1\in\mathcal{L}(\mathcal{H})$ have eigenvalues $\pm1$, they are unitary operators and do not have zero eigenspaces. The operator $\hat{A}_1\hat{A}_0$ is also unitary, and hence there exists a vector $\ket{\alpha}$, such that
\begin{equation}
  \hat{A}_1\hat{A}_0\ket{\alpha} = a_{\alpha}\ket{\alpha},
\end{equation}
with $a_{\alpha}\in\mathbb{C}$ and $|a_{\alpha}| = 1$. Then we have
\begin{equation}
  \hat{A}_1\hat{A}_0(\hat{A}_1\ket{\alpha}) = a_{\alpha}^*(\hat{A}_1\ket{\alpha}),
\end{equation}
indicating that $\ket{\alpha'}\equiv\hat{A}_1\ket{\alpha}$ is also an eigenvector of $\hat{A}_1\hat{A}_0$. Consider a subspace $\mathcal{H}^\alpha=\text{span}(\ket{\alpha},\ket{\alpha'})\subseteq \mathcal{H}$, if $\abs{\braket{\alpha}{\alpha'}}=1$, then $\dim{\mathcal{H}^\alpha}=1$; otherwise, $\dim{\mathcal{H}^\alpha}=2$. 
Since $\hat{A}_1\hat{A}_0$ is unitary, its eigenvectors span the whole Hilbert space. Then, we can iterate above steps and finally decompose the Hilbert space, $\mathcal{H}=\oplus_{\alpha}\mathcal{H}^\alpha$, where $\alpha$ denotes the index of subspace.


Next, we verify that $\mathcal{H}^\alpha$ are invariant subspaces of the operators $\hat{A}_0,\hat{A}_1$,
\begin{equation}
	\begin{split}
\hat{A}_1\ket{\alpha}&=\ket{\alpha'}, \\
\hat{A}_1\ket{\alpha'}&=\ket{\alpha}, \\
\hat{A}_0\ket{\alpha'}&=\hat{A}_0\hat{A}_1\ket{\alpha}=(\hat{A}_1\hat{A}_0)^{\dagger}\ket{\alpha}=a_{\alpha}^*\ket{\alpha}, \\
\hat{A}_0\ket{\alpha}&=\hat{A}_0\hat{A}_1\hat{A}_1\ket{\alpha}=(\hat{A}_1\hat{A}_0)^{\dagger}\ket{\alpha'}=a_{\alpha}\ket{\alpha'}. \\
	\end{split}
\end{equation}
Therefore, the lemma is proved.
\end{proof}

As a consequence of Jordan's lemma, on Alice's side, there exists a direct-sum decomposition of the space $\mathcal{H}_A=\bigoplus_\alpha\mathcal{H}_A^\alpha$ and $\forall\alpha, \dim{\mathcal{H}_A^\alpha}\le2$, such that the  quantum observables $\hat{A}_0,\hat{A}_1$ can be expressed as
\begin{equation}
  \hat{A}_i=\bigoplus_\alpha \hat{A}_i^\alpha=\sum_\alpha\hat{\Pi}_\alpha \hat{A}_i\hat{\Pi}_\alpha,
\end{equation}
where $i\in\{0,1\}$ and $\hat{\Pi}_\alpha$ is the projector onto the subspace $\mathcal{H}_A^\alpha$. The expectation of measuring the observable on the system $\rho^A$ becomes
\begin{equation}
	\begin{split}
\Tr(\hat{A}_i\rho^A) &= \sum_\alpha p_\alpha^A\Tr(\hat{A}_i^\alpha\rho_\alpha^A), \\
	\end{split}
\end{equation}
where $\rho^A_\alpha\in\mathcal{D}(\mathcal{H}_A^\alpha)$, $p^A_\alpha\ge0$, $\sum_\alpha p^A_\alpha=1$, and
\begin{equation}
	\begin{split}
\sum_\alpha\hat{\Pi}_\alpha\rho^A\hat{\Pi}_\alpha &=\bigoplus_\alpha p^A_\alpha\rho^A_\alpha. \\
	\end{split}
\end{equation}
Since $\dim{\mathcal{H}_A^\alpha}\le2$, we can effectively treat $\rho^A_\alpha$ as a qubit state.

A similar argument can be applied on Bob's side, where $\mathcal{H}_B=\bigoplus_\beta\mathcal{H}_B^\beta$. Therefore, the effective system of Alice and Bob in a round can be treated as a mixture of 2-qubit states, $\rho_{\alpha\beta}^{AB}\in\mathcal{D}(\mathcal{H}_A^\alpha\otimes\mathcal{H}_B^\beta)$. In DIQKD, we can absorb the variables, $\alpha$ and $\beta$, into the hidden variable $\lambda$. That is, given $\lambda$, 2-qubit state $\rho_{\lambda}^{AB}$, and Alice's and Bob's observables $\hat{A}_0^\lambda,\hat{A}_1^\lambda,\hat{B}_0^\lambda,\hat{B}_1^\lambda$ are fixed. So, we can also treat $\lambda$ as a system variable.

Note that the additional measurement setting on Alice's side, $\hat{A}_2$, cannot be guaranteed to take the block-diagonal form with the same direct-sum decomposition of the Hilbert space $\mathcal{H}_A$ as $\hat{A}_0,\hat{A}_1$. Nevertheless, we can further absorb the parameters for constructing the measurement $\hat{A}_2$ into $\lambda$,  $\hat{A}_2^{\lambda}$. In general, $\hat{A}_2^{\lambda}$ does not represent a qubit measurement on a qubit state. In the end, the system variable $\lambda$ would determine all the underlying states and  measurements.

In summary, the DIQKD protocol can be equivalently modelled by the quantum circuit shown in Figure \ref{Fig:EffParSum}. We summarise the system modelling results with the equivalent process description shown in Box~\ref{box:VirtualProtocol}, which generates the same classical statistics and quantum signals as the actual protocol. In this equivalent process description, it becomes clear how we should link the DIQKD security analysis to the complementarity-based security statement in Lemma~\ref{Lemma:SecDef}. In the key generation rounds, the projection operations $\{\hat{\Pi}_{\lambda_i}^B\}_{\lambda_i}$ prepare Bob's subsystems in the key generation rounds to qubit ones, corresponding to Step~2 in the protocol of Box~\ref{box:KoashiVirtual} and naturally sufficing the state preparation requirement in Lemma~\ref{Lemma:SecDef}.

\begin{mybox}[label={box:VirtualProtocol}]{{Equivalent Process Description in DIQKD}}
\textbf{Note: }The process described here is an equivalent description of the physical process in Step 1, Box~\ref{box:ProtocolDetail}.

\tcblower
\begin{enumerate}[leftmargin=*]
\item \emph{State Preparation:}
\begin{enumerate}[leftmargin=*]
  \item Alice and Bob share a bipartite quantum state $\rho^{AB}\in \mathcal{D}\left\{\bigotimes_{i=1}^n[\mathcal{H}_{A(i)}\otimes\mathcal{H}_{B(i)}]\right\}$.
  \item Quantum channels $\hat{K}_i^A,\hat{K}_i^B$ are applied to the $i$th subsystems $\mathcal{H}_{A(i)},\mathcal{H}_{B(i)}$ based on all the local information.
  \item An LOCC channel $\hat{T}_{i}$ is applied to the $(i+1)$th subsystems between $\mathcal{H}_{A(i+1)},\mathcal{H}_{B(i+1)}$.
\end{enumerate}
\item \emph{Projection to Qubits:}
\begin{itemize}[leftmargin=*]
  \item If $x_i = 2$,\\
    The projection $\{\hat{\Pi}_{\lambda_i}^B\}_{\lambda_i}$ is applied to Bob's $i$th subsystem. The projection result of $\lambda_i$ is recorded.
  \item If $x_i = 0,1$,\\
    The projection $\{\hat{\Pi}_{\lambda_i}^A\otimes\hat{\Pi}_{\lambda_i}^B\}_{\lambda_i}$ is applied to Alice's and Bob's $i$th subsystems. The projection result of $\lambda_i$ is recorded.
\end{itemize}
\item \emph{Measurements:}
\begin{itemize}[leftmargin=*]
  \item If $x_i = 2$,\\
    Alice and Bob measure $\hat{A}_{2}^{\lambda_i},\,\hat{B}_{y}^{\lambda_i}$ on their $i$th subsystems respectively, where $y_i = y$.
  \item If $x_i = 0,1$,\\
    Alice and Bob measure $\hat{A}_{x}^{\lambda_i},\,\hat{B}_{y}^{\lambda_i}$ on their $i$th subsystems respectively, where $x_i = x,y_i = y$.
\end{itemize}
\end{enumerate}
\end{mybox}

\begin{figure}[hbtp!]
	\begin{tikzpicture}[
		scale=1,
		>=latex
		]
		\node[circle,fill,gray,inner sep=1pt,label=left:$\mathcal{H}_{A_i}$] (orig) at (0,0) {};
		\node[rectangle,draw=black,fill=orange!10!white,minimum size=20pt,rounded corners] (K) at ($ (orig) + (1cm,0) $) {$\hat{K}_{i-1}^A$};
		\node[rectangle,fill=blue!10!white,minimum size=20pt,rounded corners] (TA) at ($ (K.east) + (1cm,0) $) {$\hat{T}_{i-1}$};
		\draw (TA.south west) {[rounded corners] -- (TA.north west)  -- (TA.north east) -- (TA.south east)};
		\node[rectangle,draw=black,fill=green!10!white,minimum size=20pt,rounded corners] (Ax) at ($ (TA.east) + (1.5cm,0) $) {$\hat{A}_{x_i}$};		
		\draw[<-] (Ax.north) --++ (0,0.3) node[above]{$x_i$};
		\draw[->] (Ax.east) --++ (.3,0) node[right] {$a_i$};
		\draw (orig) -- (K.west) (K.east) -- (TA.west) (TA.east) -- (Ax.west) node[midway,above] {$\rho_i^A$};
		\node[circle,fill,gray,inner sep=1pt,label=left:$\mathcal{H}_{B_i}$] (origB) at (0,-1.5) {};
		\node[rectangle,draw=black,fill=orange!10!white,minimum size=20pt,rounded corners] (KB) at ($ (origB) + (1cm,0) $) {$\hat{K}_{i-1}^B$};
		\node[rectangle,fill=blue!10!white,minimum size=20pt,rounded corners] (TB) at ($ (KB.east) + (1cm,0) $) {$\hat{T}_{i-1}$};
		\node[rectangle,draw=black,fill=green!10!white,minimum size=20pt,rounded corners] (Bx) at ($ (TB.east) + (1.5cm,0) $) {$\hat{B}_{y_i}$};		
		\draw[<-] (Bx.north) --++ (0,0.3) node[above]{$y_i$};
		\draw[->] (Bx.east) --++ (.3,0) node[right] {$b_i$};
		\draw (origB) -- (KB.west) (KB.east) -- (TB.west) (TB.east) -- (Bx.west) node[midway,above] {$\rho_i^B$};
		\draw[black,fill=blue!10!white,rounded corners] (TA.north west) rectangle node{$\hat{T}_{i-1}$} (TB.south east) ;
		\node[circle,fill,gray,inner sep=1pt,label=left:$\mathcal{H}_{A_{i+1}}$] (orig1) at (0,-3) {};
		\node[rectangle,draw=black,fill=orange!10!white,minimum size=20pt,rounded corners] (K1) at ($ (Ax) + (0,-3cm) $) {$\hat{K}_{i}^A$};
		\node[rectangle,fill=blue!10!white,minimum size=20pt,rounded corners] (TA1) at ($ (K1.east) + (1cm,0) $) {$\hat{T}_{i}$};
		\draw (TA1.south west) {[rounded corners] -- (TA1.north west)  -- (TA1.north east) -- (TA1.south east)};
		\node[rectangle,draw=black,fill=green!10!white,minimum size=20pt,rounded corners] (Ax1) at ($ (TA1.east) + (1.5cm,0) $) {$\hat{A}_{x_{i+1}}$};
		\draw[<-] (Ax1.north) --++ (0,0.3) node[above]{$x_{i+1}$};
		\draw[->] (Ax1.east) --++ (.3,0) node[right] {$a_{i+1}$};
		\draw (orig1) -- (K1.west) (K1.east) -- (TA1.west) (TA1.east) -- (Ax1.west) node[midway,above] {$\rho_{i+1}^A$};
		\node[circle,fill,gray,inner sep=1pt,label=left:$\mathcal{H}_{B_{i+1}}$] (origB1) at (0,-4.5) {};
		\node[rectangle,draw=black,fill=orange!10!white,minimum size=20pt,rounded corners] (KB1) at ($ (Bx) + (0,-3cm) $) {$\hat{K}_{i}^B$};
		\node[rectangle,fill=blue!10!white,minimum size=20pt,rounded corners] (TB1) at ($ (KB1.east) + (1cm,0) $) {$\hat{T}_{i}$};
		\node[rectangle,draw=black,fill=green!10!white,minimum size=20pt,rounded corners] (Bx1) at ($ (TB1.east) + (1.5cm,0) $) {$\hat{B}_{y_{i+1}}$};
		\draw[<-] (Bx1.north) --++ (0,0.3) node[above]{$y_{i+1}$};
		\draw[->] (Bx1.east) --++ (.3,0) node[right] {$a_{i+1}$};
		\draw (origB1) -- (KB1.west) (KB1.east) -- (TB1.west) (TB1.east) -- (Bx1.west) node[midway,above] {$\rho_{i+1}^B$};
		\draw[black,fill=blue!10!white,rounded corners] (TA1.north west) rectangle node{$\hat{T}_{i}$} (TB1.south east) ;
		\draw[snake=snake,gray,opacity=.5] (orig) -- (origB1);
	\end{tikzpicture}
\caption{System characterisation of the sequential setting. Two rounds are depicted for demonstration.
The subsystems for each round are denoted by their Hilbert spaces $\mathcal{H}_{A_i},\mathcal{H}_{B_i}$,
which are measured in a time sequence.
In general, the subsystems can be correlated or entangled \emph{a priori}, as shown by the curvy line. 
In each round, Alice and Bob randomly set the measurement bases $x_i,y_i$, respectively. The measurement on either party is independent of the other but can be influenced by all previous events on both sides.
After the measurements in a round end, the leftover quantum systems are transferred to the next subsystems of Alice and Bob via local hidden updating quantum operations $\hat{K}_i^A,\hat{K}_i^B$, respectively.
Afterwards, a general quantum channel $\hat{T}_i$ is applied to Alice's and Bob's systems, which prepares the joint system to be measured in the subsequent round in $\rho_{i+1}^{AB}$ and allows the untrusted devices to synchronise the history on both sides. For measurements corresponding to the settings $x_i,y_i\in\{0,1\}$, they can be described by qubit observables realised via the qubit projections $\hat{\Pi}_{\lambda_i}^A,\hat{\Pi}_{\lambda_i}^B$.
} \label{Fig:EffParSum}
\end{figure}
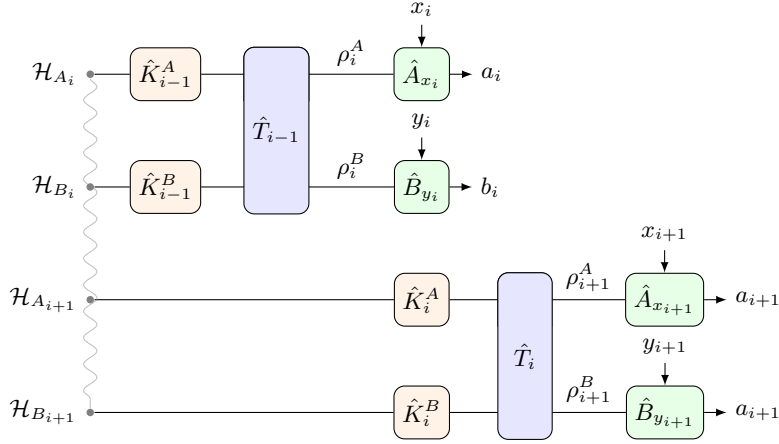


\section{Single-Round Phase-Error Probability}\label{Supp:PhaseErrorProb}
\subsection{Problem as an optimisation}
In this section, we shall derive the phase-error probability formula in a single round. Comparing to the virtual protocol in Box~\ref{box:KoashiVirtual}, the analysis in this section shall give the construction of the auxiliary measurement, $\hat{M}_R$. Along the complementarity-based quantum-error-correction approach, we virtually replace the key generation measurements with the phase-error measurement shown in Box~\ref{box:PhaseError} and remain all the other operations (both quantum and classical) the same as in the original protocol. Here, $\hat{B}_{\bot}^{\lambda_i}\in\mathcal{L}(\mathcal{H}_B^{\lambda_i})$ is the complementary observable of $\hat{B}_0^{\lambda_i}\in\mathcal{L}(\mathcal{H}_B^{\lambda_i})$. We denote the eigenvectors of $\hat{B}_{\bot}^{\lambda_i}$ as $\ket{b_{\bot}^{\lambda_i(+)}},\ket{b_{\bot}^{\lambda_i(-)}}$ and the eigenvectors of $\hat{B}_{0}^{\lambda_i}$ as $\ket{b_{0}^{\lambda_i(+)}},\ket{b_{0}^{\lambda_i(-)}}$. The sets of eigenvectors form mutually unbiased bases. Since the observables are now qubit ones, their eigenvectors satisfy
\begin{equation}
\left|\bra{b_{\bot}^{\lambda_i(s)}}\ket{b_{0}^{\lambda_i(t)}}\right|^2=\frac{1}{2},\forall s,t\in\{+,-\}.
\end{equation}
The measurement on Alice's side, $\hat{A}_{\theta_i}^{\lambda_i}\in\mathcal{L}(\mathcal{H}_A^{\lambda_i})$, is an auxiliary measurement. Without loss of generality, we can take it as a binary projection measurement with eigenvalues $\pm1$. The parameter $\theta_i$ denotes the freedom in the choice of the auxiliary measurement. The explicit parameterised representation of $\hat{A}_{\theta_i}^{\lambda_i}$ shall be made clear in the next section. In the prediction, Alice takes the measurement result of $\hat{A}_{\theta_i}^{\lambda_i}$ as her guess of the outcome of $\hat{B}_{\bot}^{\lambda_i}$. With probability $[1+\alpha(\hat{A}_{\theta_i}^{\lambda_i})]/2$, the outcome of $\hat{B}_\bot^{\lambda_i}$ can be predicted correctly, where
\begin{equation}\label{Eq:ComBasis}
\alpha(\hat{A}_{\theta_i}^{\lambda_i}) = \Tr[\rho_{\lambda_i}^{AB}({\hat{A}_{\theta_i}^{\lambda_i}}\otimes \hat{B}_\bot^{\lambda_i})],
\end{equation}
which we call the correlation between the measurements $\hat{A}_{\theta_i}^{\lambda_i}$ and $\hat{B}_\bot^{\lambda_i}$ on the state $\rho_{\lambda_i}^{AB}$. If the measurement results of $\hat{B}_{\bot}^{\lambda_i}$ and $\hat{A}_{\theta_i}^{\lambda_i}$ are perfectly correlated, that is, $\alpha(\hat{A}_{\theta_i}^{\lambda_i})=1$, the system would be certified on the complementary basis. In this case, the key bit $\kappa_i^B$ shall be private. In a general case, we introduce a random variable $E_{p}^{\lambda_i}$, which we call the phase-error variable in the $i$th round. An error occurs if the measurement results of $\hat{A}_{\theta_i}^{\lambda_i}$ and $\hat{B}_{\bot}^{\lambda_i}$ are not equal and $E_{p}^{\lambda_i}=1$, otherwise $E_{p}^{\lambda_i}=0$. Conditioned on the system variable, had the phase-error measurement been performed in the $i$th round, a phase error would occur with probability
\begin{equation}\label{def:ConPhaseError}
  \Pr{E_{p}^{\lambda_i} = 1|x_i = 2, y_i = 0, \lambda_i} = \dfrac{1-\alpha(\hat{A}_{\theta_i}^{\lambda_i})}{2}.
\end{equation}
In Figure~\ref{Fig:VirtualProtocol}, we depict the virtual phase-error measurement procedure. If we compare it to the well-known Bennett-Brassard-1984 QKD protocol \cite{bennett1984quantum}, $\hat{A}_{2}^{\lambda_i}$ and $\hat{B}_0^{\lambda_i}$ correspond to the $Z$-basis measurements, while $\hat{A}_{\theta_i}^{\lambda_i}$ and $\hat{B}_{\perp}^{\lambda_i}$ correspond to the $X$-basis measurements.

\begin{mybox}[label={box:PhaseError}]{{Phase-error Measurement}}
\textbf{Notes:}
\begin{itemize}[leftmargin=*]
\item Replace \emph{only} the key generation measurements with the following operations.
\item The information exchange with the outside needs to be the same as in the experiment.
\end{itemize}

\tcblower
\begin{enumerate}[leftmargin=*]
\item Alice applies the projection $\{\hat{\Pi}_{\lambda_i}^A\}_{\lambda_i}$ to her $i$th subsystem.
\item Alice and Bob measure $\hat{A}_{\theta_i}^{\lambda_i},\hat{B}_{\bot}^{\lambda_i}$ on their systems and record the measurement results, respectively.
\item Based on the measurement outcome of $\hat{A}_{\theta_i}^{\lambda_i}$, Alice predicts the corresponding measurement outcome of $\hat{B}_{\bot}^{\lambda_i}$.
\item Alice and Bob compare the measurement results. If Alice makes a wrong prediction, a phase error occurs with a binary variable $E_{p}^{\lambda_i}$ taking the value $1$. Otherwise, $E_{p}^{\lambda_i}=0$.
\end{enumerate}
\end{mybox}

\begin{figure}[htb]
\centering \includegraphics[width=7.5cm]{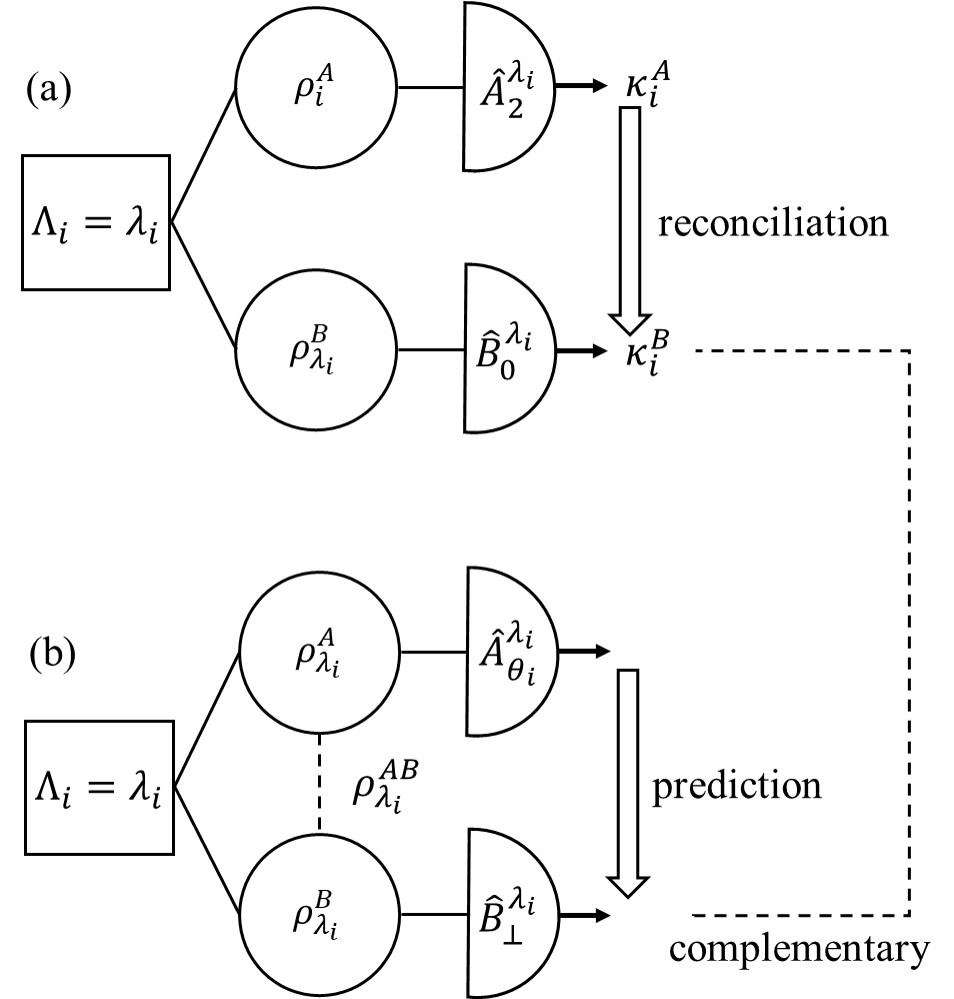}
\caption{A schematic diagram of the virtual phase-error measurement. The underlying quantum state $\rho_{\lambda_i}^{AB}$ depends on the system variable $\Lambda_i$. (a) In a key generation round, the key bit $\kappa_i^B$ is obtained by measuring $\hat{B}_0^{\lambda_i}$ on the state $\rho_{\lambda_i}^B$. (b) In a virtual phase-error measurement, Alice measures $\hat{A}_{\theta_i}^{\lambda_i}$ on her subsystem $\rho_{\lambda_i}^A$. Bob measures $\hat{B}_{\bot}^{\lambda_i}$ on his subsystem $\rho_{\lambda_i}^B$, where $\hat{B}_{\bot}^{\lambda_i}$ is complementary to $\hat{B}_{0}^{\lambda_i}$. Alice takes the measurement result of $\hat{A}_{\theta_i}^{\lambda_i}$ to predict the result of $\hat{B}_{\perp}^{\lambda_i}$.}
\label{Fig:VirtualProtocol}
\end{figure}

By our definition of the system variable $\lambda_i$ (see Sec.~\ref{Supp:EffDimReduction}), the correlated effects brought by the previous rounds have been considered.
We denote the corresponding underlying phase-error probability corresponding to the system variable value $\lambda_i$ in Eq.~\eqref{def:ConPhaseError} as $e_p^{\lambda_i}$. In this section, all the conditional probabilities should be considered as a value, where the terms being conditioned are specific values of random variables.

While the phase error is not directly measured in the protocol, the phase-error probability is an intrinsic property of the system. Intuitively, more entanglement between Alice and Bob implies more privacy in the generated key and hence the smaller the phase-error probability. On the other hand, entanglement can be quantified via the expected Bell value. Conditioned on the system variable, the expected Bell value of a round is given by
\begin{equation}\label{def:UnderBellValue}
\begin{split}
  S^{\lambda_i} &\equiv
  \sum_{x,y\in\{0,1\}}(-1)^{xy}\mathbb{E}(a_{i}b_{i}|x_i=x,y_i=y, \lambda_i) \\ &=
  \sum_{x,y\in\{0,1\}}(-1)^{xy}\Tr[\rho_{\lambda_i}^{AB}(\hat{A}_{x}^{\lambda_i}\otimes \hat{B}_{y}^{\lambda_i})] \\ &= \Tr(\rho_{\lambda_i}^{AB}\hat{S}^{\lambda_i}).
\end{split}
\end{equation}
Here, the operator
\begin{equation}\label{Eq:BellOperator}
\hat{S}^{\lambda_i}=\sum_{x,y\in\{0,1\}}(-1)^{xy}\hat{A}_{x}^{\lambda_i}\otimes \hat{B}_{y}^{\lambda_i}
\end{equation}
is the Bell operator of the round given the system variable value $\lambda_i$.

For brevity, we omit the system variable $\lambda$ and the label for round number in this section, and express the joint underlying quantum system in the round as $\rho\in\mathcal{D}(\mathcal{H}_A\otimes\mathcal{H}_B)$. Suppose given the system variable, the expected Bell value $S$ of the round is given. To obtain the worst phase-error probability under this constraint, we can solve the following optimisation problem of the correlation between $\hat{A}_\theta,\hat{B}_\perp$ on the state $\rho$,
\begin{equation}\label{Eq:OriginOpt}
\begin{split}
\alpha =& \min_{\rho,\hat{A}_0,\hat{A}_1,\hat{B}_0,\hat{B}_1}\max_{\hat{A}_\theta}\Tr[\rho({\hat{A}_\theta}\otimes \hat{B}_\perp)], \\
\text{s.t. }
&S = \sum_{x,y\in\{0,1\}}(-1)^{xy}\Tr[\rho(\hat{A}_x\otimes \hat{B}_y)].
\end{split}
\end{equation}
All the measurement observables are projective with eigenvalues $\pm1$, and $\hat{A}_0,\hat{A}_1,\hat{A}_{\theta}\in\mathcal{L}(\mathcal{H}_A),\hat{B}_0,\hat{B}_1,\hat{B}_\perp\in\mathcal{L}(\mathcal{H}_B)$. For convenience, we first state the optimisation result,
\begin{lemma}\label{lemma:Correlation}
The solution to the optimisation in Eq.~\eqref{Eq:OriginOpt} is given by,
\begin{eqnarray}\label{SuppEq:PhaseErrorProb}
\alpha\geq\alpha(S)\equiv
\left\{
\begin{tabular}{ll}
$\sqrt{(S/2)^2-1}$, & $2<S\leq2\sqrt{2}$,\\
$0$, & $S\leq2$.
\end{tabular}
\right.
\end{eqnarray}
Therefore, the corresponding phase-error probability is upper bounded by
\begin{eqnarray}\label{Eq:SuppPhaseError}
e_p\leq e_p(S)\equiv\dfrac{1-\alpha(S)}{2}=
\left\{
\begin{tabular}{ll}
$\dfrac{1-\sqrt{(S/2)^2-1}}{2}$, & $2<S\leq2\sqrt{2}$,\\
$\dfrac{1}{2}$, & $S\leq2$.
\end{tabular}
\right.
\end{eqnarray}
\end{lemma}

If we restrict Eve's hacking strategy to i.i.d.~attacks, where the system variables $\Lambda_i$ are independent of $i$, the lemma is consistent with the previous result \cite{pironio2009device}. Nevertheless, the result in Ref.~\cite{pironio2009device} has a different operational meaning as Holevo information and takes a different approach. To our best knowledge, in the sense of phase-error probability, the first attempt to construct a complementarity-based device-independent security analysis starts in Ref.~\cite{tsurumaru2016multi}, where the authors make the same definition for a complementary basis as in Eq.~\eqref{Eq:ComBasis}. Yet the authors did not reach a tight relation between phase-error probability and expected Bell value. Upon finishing our paper, we notice the work of Ref.~\cite{woodhead2021device}, which arrives at a similar result as ours in a single round with a slightly different derivation. However, the work does not further go into the non-i.i.d. statistics. Instead, to obtain a full security analysis, the authors transform the result to the quantum entropy conditioned on Eve's system via the entropic uncertainty relation~\cite{berta2010uncertainty} and resorts to the entropic approach in the end~\cite{arnon2018practical}.

We solve the optimisation problem in two steps. In the first step, we parameterise the quantum system and operators using Pauli operators. Afterwards, we derive a lower bound on the optimal value of the optimisation in Eq.~\eqref{Eq:OriginOpt} from the dual problem.

%

\subsection{Step 1: operator parametrisation}\label{Supp:RedStep1}
Without loss of generality, the observables can be taken as expanded by the Pauli operators $\hat{\sigma}_x,\hat{\sigma}_z$ on the corresponding spaces by using additional degrees of freedom. The measurements can be parameterised as
\begin{equation}\label{Eq:EffMeasure}
\begin{split}
&\hat{A}_0 = \sin{\theta_0}\hat{\sigma}_x + \cos{\theta_0}\hat{\sigma}_z, \\
&\hat{A}_1 = \sin{\theta_1}\hat{\sigma}_x + \cos{\theta_1}\hat{\sigma}_z, \\
&\hat{B}_0 = \hat{\sigma}_z, \\
&\hat{B}_1 = \sin{\gamma}\hat{\sigma}_x + \cos{\gamma}\hat{\sigma}_z. \\
\end{split}
\end{equation}
The complementary measurement of $\hat{B}_0$, $\hat{B}_\perp$, and the auxiliary measurement observable on Alice's side, $\hat{A}_\theta$, can also be spanned by $\hat{\sigma}_x,\,\hat{\sigma}_z$,
\begin{equation}
\begin{split}
  \hat{A}_\theta &= \sin{\theta}\hat{\sigma}_x+\cos{\theta}\hat{\sigma}_z, \\
  \hat{B}_\perp &= \hat{\sigma}_x.
\end{split}
\end{equation}

Since we are focused on the measurement results of the observables, where the system $\rho$ and its complex conjugate $\rho^*$ give the same statistics, the density operator of the shared state $\rho$ can be taken as real positive semi-definite. We can write the density operator of $\rho$ as the matrix
\begin{equation}\label{SuppEq:state}
\rho=
\left(\begin{matrix}
\rho_{11} & \rho_{12} & \rho_{13} & \rho_{14}\\
\rho_{21} & \rho_{22} & \rho_{23} & \rho_{24}\\
\rho_{31} & \rho_{32} & \rho_{33} & \rho_{34}\\
\rho_{41} & \rho_{42} & \rho_{43} & \rho_{44}
\end{matrix}\right)
\end{equation}
under the Bell-state basis $\left\{\ket{\Phi^+},\ket{\Psi^-},\ket{\Phi^-},\ket{\Psi^+}\right\}$, with $\rho_{ij}=\rho_{ji},\forall i,j\in\{1,2,3,4\}$. The Bell states are expressed as
\begin{equation}
\begin{split}
\ket{\Phi^+}&=\frac{\ket{00}+\ket{11}}{\sqrt{2}},\\ \ket{\Psi^-}&=\frac{\ket{01}-\ket{10}}{\sqrt{2}},\\ \ket{\Phi^-}&=\frac{\ket{00}-\ket{11}}{\sqrt{2}},\\ \ket{\Psi^+}&=\frac{\ket{01}+\ket{10}}{\sqrt{2}},
\end{split}
\end{equation}
where $\{\ket{0},\ket{1}\}$ are the eigenvectors of the Pauli operator $\hat{\sigma}_z$ acting on the corresponding local system. We denote $c_{ij}=\Tr\left[\rho (\hat{\sigma}_i\otimes\hat{\sigma}_j)\right],\forall i,j\in\{x,z\}$. After some simple calculation, we could get
\begin{equation}\label{SuppEq:correlation}
\begin{split}
c_{xx}&=\rho_{11}-\rho_{22}-\rho_{33}+\rho_{44},\\
c_{xz}&=-\rho_{12}-\rho_{21}+\rho_{34}+\rho_{43},\\
c_{zx}&=\rho_{12}+\rho_{21}+\rho_{34}+\rho_{43},\\
c_{zz}&=\rho_{11}-\rho_{22}+\rho_{33}-\rho_{44}.\\
\end{split}
\end{equation}
The value of the objective function becomes
\begin{equation}
\Tr\left[\rho (\hat{A}_\theta\otimes\hat{B}_\perp)\right]=\sin{\theta}c_{xx} + \cos{\theta}c_{zx}.
\end{equation}
We note that the objective function has a trivial lower bound of zero. In the optimisation of the objective function, the minimisation and the maximisation can be exchanged. Now we focus on the optimisation over $\hat{A}_{\theta}$. If $c_{xx}=c_{zx}=0$, the objective function is zero. If $c_{xx},c_{zx}$ do not simultaneously take the value of zero, say $c_{xx}\neq0$, then by taking $\theta=\arctan{-\frac{c_{zx}}{c_{xx}}}$, the objective function takes the value of zero. Therefore, we have $e_p\leq\frac{1}{2}$. Intuitively, this means that Alice can always find a strategy such that her prediction is correct with a probability of at least $1/2$.

For the case where $c_{xx},c_{zx}$ do not simultaneously take the value of zero, over the optimisation of $\hat{A}_\theta$, the optimum value of the objective function is taken when
\begin{equation}
\begin{split}
\sin{\theta}&=\frac{|c_{xx}|}{\sqrt{c_{xx}^2+c_{zx}^2}},\\
\cos{\theta}&=\frac{|c_{zx}|}{\sqrt{c_{xx}^2+c_{zx}^2}}.
\end{split}
\end{equation}
The optimisation problem becomes
\begin{equation}\label{SuppEq:primal}
\begin{split}
\alpha &= \min_{\rho_{ij},\theta_0,\theta_1,\gamma}\sqrt{c_{xx}^2+c_{zx}^2}, \\
\text{s.t. }
S &= S(\rho_{ij},\theta_0,\theta_1,\gamma) = \sum_{x,y\in\{0,1\}}(-1)^{xy}\Tr[\rho(\hat{A}_x\otimes \hat{B}_y)].
\end{split}
\end{equation}
The case of $c_{xx}=c_{zx}=0$ can be unified to Eq.~\eqref{SuppEq:primal}. The function $S(\rho_{ij},\theta_0,\theta_1,\gamma)$ represents the expected Bell value given by the shared state with elements in the density matrix $\rho_{ij}$, and measurement observables parameterised by $\theta_0,\theta_1,\gamma$. We consider a unitary operator $\hat{R}(\vartheta)=\cos{\frac{\vartheta}{2}}\hat{I}+i\sin{\frac{\vartheta}{2}}\hat{\sigma}_y$ acting on $\mathcal{H}_A$. By applying the unitary operation of $\left(\hat{R}(\vartheta)\otimes\hat{I}\right)[\cdot]\left(\hat{R}(\vartheta)\otimes\hat{I}\right)^{-1}$ to both $\rho$ and $\hat{A}_x\otimes\hat{B}_y$, the solution to the optimization does not change. A simple computation shows
\begin{equation}
\begin{split}
\left(R(\vartheta)\otimes\hat{I}\right)\rho\left(R(\vartheta)\otimes\hat{I}\right)^{-1}=\left(\begin{matrix}\rho'_{11} & \rho'_{12} & \rho'_{13} & \rho'_{14}\\
\rho'_{21} & \rho'_{22} & \rho'_{23} & \rho'_{24}\\
\rho'_{31} & \rho'_{32} & \rho'_{33} & \rho'_{34}\\
\rho'_{41} & \rho'_{42} & \rho'_{43} & \rho'_{44}\end{matrix}\right),
\end{split}
\end{equation}
where
\begin{equation}
\begin{split}
  \rho'_{12}&=\frac{1}{2}\left(\sin{\vartheta}\rho_{11}+2\cos{\vartheta}\rho_{12}-\sin{\vartheta}\rho_{22}\right),\\
  \rho'_{34}&=\frac{1}{2}\left(-\sin{\vartheta}\rho_{33}+2\cos{\vartheta}\rho_{34}+\sin{\vartheta}\rho_{44}\right).
\end{split}
\end{equation}
Let $\vartheta = \arctan\left(\frac{-2\rho_{12}-2\rho_{34}}{\rho_{11}-\rho_{22}-\rho_{33}+\rho_{44}}\right)$. Then we have $c'_{zx}=\rho'_{12}+\rho'_{21}+\rho'_{34}+\rho'_{43}=0$. For simplicity, we still represent the measurements using the notations in Eq.~\eqref{Eq:EffMeasure}. The state can be parameterised as
\begin{equation}\label{Eq:EffState}
\begin{split}
\rho=
\left(\begin{matrix}
\rho_{11} & \rho_{12} & \rho_{13} & \rho_{14}\\
\rho_{12} & \rho_{22} & \rho_{23} & \rho_{24}\\
\rho_{13} & \rho_{23} & \rho_{33} & -\rho_{12}\\
\rho_{14} & \rho_{24} & -\rho_{12} & \rho_{44}
\end{matrix}\right),
\end{split}
\end{equation}
where $\sum_{i=1}\rho_{ii}^{\lambda}=1,\rho_{ij}\geq0$. In this parametrisation, $c_{zx}=0$. Then, the optimisation problem in Eq.~\eqref{SuppEq:primal} can be simplified,
\begin{equation}\label{SuppEq:SimPrimal}
\begin{split}
\alpha &= \min_{\rho_{ij},\theta_0,\theta_1,\gamma}\left|c_{xx}\right|, \\
\text{s.t. }
S &= S(\rho_{ij},\theta_0,\theta_1,\gamma).
\end{split}
\end{equation}
Here, the function $S(\rho_{ij},\theta_0,\theta_1,\gamma)$ is parameterised by
\begin{eqnarray}
S(\rho_{ij},\theta_0,\theta_1,\gamma)
&=&
\left(\cos{\theta_0}+\cos{\theta_0}\cos{\gamma}+\cos{\theta_1}-\cos{\theta_1}\cos{\gamma}\right)c_{zz}
+ \left(\sin{\theta_0}\sin{\gamma}-\sin{\theta_1}\sin{\gamma}\right)c_{xx} \notag\\ && + \left(\sin{\theta_0}+\sin{\theta_0}\cos{\gamma}+\sin{\theta_1}-\sin{\theta_1}\cos{\gamma}\right)c_{xz}.
\end{eqnarray}

\subsection{Step 2: lower bound from the dual problem}\label{Supp:RedStep2}
If we consider the dual problem of Eq.~\eqref{SuppEq:SimPrimal},
\begin{equation}\label{SuppEq:dual}
\begin{split}
S^* &= \max_{\rho_{ij},\theta_0,\theta_1,\gamma}S(\rho_{ij},\theta_0,\theta_1,\gamma), \\
\text{s.t. }
\alpha&=|c_{xx}|,
\end{split}
\end{equation}
to which the solution is of the form $S^*(\alpha)$, with a duality argument, we can get a lower bound on the solution to the primal problem by taking the value $S$ into the inverse function $[S^*(\alpha)]^{-1}$. For the dual problem, the solution is given by Lemma \ref{Lem:Smax}. In the proof of the lemma, we employ the following observation.

\begin{observation}
The parameters defined in Eq.~\eqref{SuppEq:correlation} satisfy the following inequality
\begin{equation}\label{eq:cxzczz}
\begin{split}	
|c_{zz}|\leq\sqrt{1-c^2_{xz}}.
\end{split}
\end{equation}
\end{observation}

\begin{proof}
	\begin{equation}
		\begin{split}
			|c_{xz}|&=|-\rho_{12}-\rho_{21}+\rho_{34}+\rho_{43}| \\
			&\leq 2|\rho_{12}|+2|\rho_{34}| \\
			&\leq 2\sqrt{\rho_{11}\rho_{22}} + 2\sqrt{\rho_{33}\rho_{44}} \\
			&\leq 2\sqrt{(\rho_{11}+\rho_{33})(\rho_{22}+\rho_{44})} \\
			&= 2\sqrt{\left(\frac{1+c_{zz}}{2}\right)\left(\frac{1-c_{zz}}{2}\right)} \\
			&= \sqrt{1-c^2_{zz}}.
		\end{split}
	\end{equation}
	In the second line, we have used the fact that the state $\rho$ in Eq.~\eqref{Eq:EffState} is Hermitian and the triangle inequality. In the third line, we have used the fact that $\rho$ is positive semi-definite. In the fourth line, we have used the Cauchy-Schwarz inequality. In the fifth line, we have used the expression of $c_{zz}$ in Eq.~\eqref{SuppEq:correlation} and the fact that $\Tr(\rho)=1$. Rearranging this inequality shall we obtain Eq.~\eqref{eq:cxzczz}.
\end{proof}

\begin{lemma}\label{Lem:Smax}
The solution to Eq.~\eqref{SuppEq:dual} is given by
\begin{equation}
S^*= 2\sqrt{1+\alpha^2}.
\end{equation}
\end{lemma}

\begin{proof}
From the Cauchy inequality, the value of $S(\rho_{ij},\theta_0,\theta_1,\gamma)$ can be upper bounded by
\begin{eqnarray}\label{SuppEq:Sbound}
S(\rho_{ij},\theta_0,\theta_1,\gamma)
&\leq&\sqrt{\left(c_{zz}+\cos{\gamma}c_{zz}\right)^2 + \left(\sin{\gamma}c_{xx}+c_{xz}+\cos{\gamma}c_{xz}\right)^2}\notag\\
&& + \sqrt{\left(c_{zz}-\cos{\gamma}c_{zz}\right)^2 + \left(-\sin{\gamma}c_{xx}+c_{xz}-\cos{\gamma}c_{xz}\right)^2}.
\end{eqnarray}

With the inequality Eq.~\eqref{eq:cxzczz}, the RHS of Eq.~\eqref{SuppEq:Sbound} can be upper bounded by
\begin{eqnarray}\label{Supp:UppGeoFunction}
&&\sqrt{\left(1+\cos{\gamma}\right)^2\left(1-c^2_{xz}\right) + \left(\sin{\gamma}c_{xx}+c_{xz}+\cos{\gamma}c_{xz}\right)^2}\notag\\
&& + \sqrt{\left(1-\cos{\gamma}\right)^2\left(1-c^2_{xz}\right) + \left(-\sin{\gamma}c_{xx}+c_{xz}-\cos{\gamma}c_{xz}\right)^2}\notag\\
&=& \sqrt{(1+\cos{\gamma})^2+(\sin{\gamma}c_{xx})^2+2\sin{\gamma}(1+\cos{\gamma})c_{xx}c_{xz}}\notag\\
&& + \sqrt{(1-\cos{\gamma})^2+(\sin{\gamma}c_{xx})^2-2\sin{\gamma}(1-\cos{\gamma})c_{xx}c_{xz}}.
\end{eqnarray}
Without loss of generality, we assume $c_{xx}\geq0$. Then, we can visualize Eq.~\eqref{Supp:UppGeoFunction} in Figure \ref{Fig:CHSHgeometry}, from which we can see that the maximum of this function reaches when the derivative of $\tau$ is 0, where $\cos{\tau}\equiv c_{xz}$.

\begin{figure}[!hbtp]
\centering \resizebox{5cm}{!}{\includegraphics{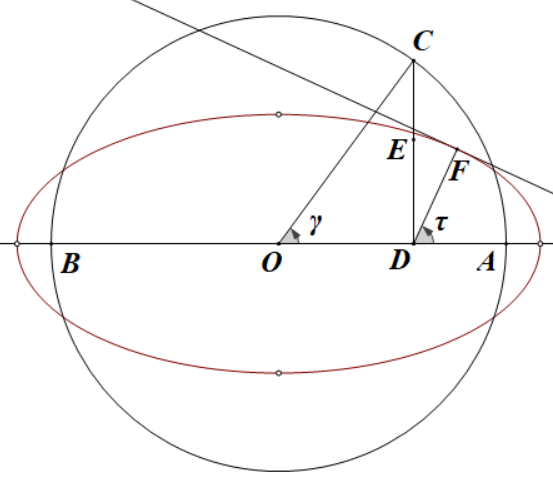}}
\caption{The geometric representation of Eq.~\eqref{Supp:UppGeoFunction}. As shown in Figure \ref{Fig:CHSHgeometry}, $\odot O$ is a circle centered at $O$ with a unit radius and $\angle AOC=\gamma$. Let $\frac{ED}{CD}=c_{xx},\cos{\tau}=c_{xz}$, and $DE=DF$. This is well-defined, since $c_{xz}\in[-1, 1]$. Then the formula in Eq.~\eqref{Supp:UppGeoFunction} equals $AF+BF$.} \label{Fig:CHSHgeometry}
\end{figure}

By taking derivative of $\tau$, we could have
\begin{eqnarray}
\frac{1+\cos{\gamma}}{\sqrt{(1+\cos{\gamma})^2+(\sin{\gamma}c_{xx})^2+2\sin{\gamma}(1+\cos{\gamma})c_{xx}\cos{\tau}}}\notag\\ =\frac{1-\cos{\gamma}}{\sqrt{(1-\cos{\gamma})^2+(\sin{\gamma}c_{xx})^2-2\sin{\gamma}(1-\cos{\gamma})c_{xx}\cos{\tau}}}.
\end{eqnarray}
With this equation, we can derive another formula,
\begin{eqnarray}
&&\prod\limits_{s\in\{-1,1\}}\sqrt{(1+s\cos{\gamma})^2+(\sin{\gamma}c_{xx})^2+2s\sin{\gamma}(1+s\cos{\gamma})c_{xx}\cos{\tau}}\notag\\
&=&(1+\cos{\gamma})(1-\cos{\gamma}) + (\sin{\gamma}c_{xx})^2,\
\end{eqnarray}
which immediately implies
\begin{eqnarray}
&&\sqrt{(1+\cos{\gamma})^2+(\sin{\gamma}c_{xx})^2+2\sin{\gamma}(1+\cos{\gamma})c_{xx}c_{xz}}\notag\\
&& + \sqrt{(1-\cos{\gamma})^2+(\sin{\gamma}c_{xx})^2-2\sin{\gamma}(1-\cos{\gamma})c_{xx}c_{xz}}\notag\\
&=& 2\sqrt{
\frac{\prod\limits_{s\in\{-1,1\}}\sqrt{(1+s\cos{\gamma})^2+(\sin{\gamma}c_{xx})^2+2s\sin{\gamma}(1+s\cos{\gamma})c_{xx}\cos{\tau}}}
{(1+\cos{\gamma})(1-\cos{\gamma})}
}\notag\\
&=&2\sqrt{\frac{(1+\cos{\gamma})(1-\cos{\gamma}) + (\sin{\gamma}c_{xx})^2}{(1+\cos{\gamma})(1-\cos{\gamma})}}\notag\\
&=&2\sqrt{1+c_{xx}^2},
\end{eqnarray}
giving the result
\begin{equation}
  S^*= \max_{\rho_{ij},\theta_0,\theta_1,\gamma}S(\rho_{ij},\theta_0,\theta_1,\gamma)= 2\sqrt{1+\alpha^2}.
\end{equation}
\end{proof}

For $S\in[2,2\sqrt{2}]$, we derive a valid lower bound to the original primal problem from the inverse of this function,
\begin{equation}
\alpha\geq\alpha(S)\equiv\sqrt{\frac{S^2}{4}-1}.
\end{equation}
For $S<2$, we take the lower bound of $\alpha\geq\alpha(S)\equiv0$. Therefore, we derive the lower bound on the correlation between $\hat{A}_\theta$ and $\hat{B}_\perp$, as given in Lemma~\ref{lemma:Correlation}.

\section{Finite-Size Analysis}\label{Supp:FiniteSize}
In Sec.~\ref{Supp:PhaseErrorProb}, we have established the relationship between the phase-error probability and the expected values of measurable observables in the experiment. While in practice, the expected values of observables cannot be accessed directly but only frequencies accumulated in a finite number of rounds. The key also has a finite length, hence the key privacy should be expressed in a way of frequency as well. In addition, statistical fluctuations and arbitrary correlations may exist over multiple rounds in the experiment. Here, we shall deal with these problems and derive an estimation of the key privacy in terms of observed statistics, as shown Figure \ref{Fig:FiniteSizeFlowchart}. This section is organised as follows.
\begin{enumerate}
  \item In Sec.~\ref{Sec:FiniteSizeProb}, we formulate the privacy estimation problem in an explicit optimisation. The key privacy is characterised by the number of probable phase-error patterns, or the cardinality of the smallest phase-error probable set (up to a negligible failure probability). We also introduce the necessary statistical tools for solving the problem, including sample entropy, which will be used to estimate the number of probable phase-error patterns, and martingale-based concentration inequalities, which are used for tackling the non-i.i.d.~condition.
  \item In Sec.~\ref{Sec:epRelax}, we introduce the regularised phase-error probability based on the single round analysis.      This subsection provides a technical preparation.
  \item Based on the ingredients introduced in Sec.~\ref{Sec:FiniteSizeProb} and~\ref{Sec:epRelax}, we define the (regularised) phase-error sample entropy in Sec.~\ref{Sec:PhaseErrorSample}. By constructing a corresponding martingale and using Azuma's concentration inequality, we show how the total phase-error sample entropy in an experiment can be related to the expected Bell value. This part is the essence of the finite-size analysis.
  \item In Sec.~\ref{Sec:BellValueEst}, we apply Azuma's concentration inequality to the estimation of the average expected Bell value in the experiment.
  \item In Sec.~\ref{Sec:KeyPrivacyEst}, we combine the results in Sec.~\ref{Sec:PhaseErrorSample} and~\ref{Sec:BellValueEst} to derive the final privacy estimation result with observed statistics in the experiment. This finishes the finite-size analysis.
  \item In Sec.~\ref{Sec:KatoResult}, we present a refined finite-size result using a modified version of Azuma's inequality, namely, Kato's inequality. We explain the necessary replacement and deliver the final result.
\end{enumerate}

\begin{figure}[hbtp!]
	\begin{tikzpicture}[
		scale=1,
		every text node part/.style={align=center},
		every rectangle node/.style={rounded corners},
		>=latex
		]
		\node[rectangle,draw,text width=6cm] (Obs) at (0,0) {Observed statistics in Bell-test rounds, $X_i,Y_i,A_i,B_i$};
		\node[rectangle,draw,text width=6cm] (ExpS) at ($ (Obs.south) + (0,-1cm) $) {Average expected Bell value};
		\draw[->] (Obs.south) -- (ExpS.north) node[midway,left] {Concentration inequality} node[midway,right] {Lemma~\ref{Lemma:BellValueBound}, Corollary~\ref{Corollary:EstBellValue}};
		\node[rectangle,draw,text width=6cm] (Sam) at ($ (ExpS.south) + (0,-1cm) $) {Expected $\xi$-regularised phase-error sample entropy};
		\draw[->] (ExpS.south) -- (Sam.north) node[midway,left] {Bell value $\Rightarrow$ phase error} ;
		\node[rectangle,draw,text width=6cm] (USam) at ($ (Sam.south) + (0,-1cm) $) {Upper bound on $\xi$-regularised phase-error sample entropy};
		\draw[->] (Sam.south) -- (USam.north) node[midway,left] {Concentration inequality} node[midway,right] {Lemma~\ref{Lemma:SmoothProperty},~\ref{Lemma:MartingaleConstruct},~\ref{Lemma:SampleEntropyBound}};
		\node[rectangle,draw,text width=6cm,label=right:Theorem~\ref{thm:PrivacyAmpCost}] (Key) at ($ (USam.south) + (0,-1cm) $) {Key privacy (phase-error $\varepsilon$-smallest probable set cardinality)};
		\draw[->] (USam.south) -- (Key.north) node[midway,left] {Pattern number counting} node[midway,right] {Lemma~\ref{Lemma:PatternCount}};
		\node[rectangle,draw,text width=3.5cm] (Regu) at ($ (Sam.east) + (2.5,0) $) {$\xi$-regularised phase-error probability};
		\node[rectangle,draw,text width=3.5cm] (Single) at ($ (Regu.east) + (2.5,0) $) {Single-round phase-error probability};
		\draw[->] (Single.west) -- (Regu.east) node[midway,above=.4cm]{Definition~\ref{Def:RelaxedSampleEntropy}} (Regu.west) -- (Sam.east) node[midway,above=.4cm]{Definition~\ref{Def:SampleEntropy}};
		\node[rectangle,draw,text width=3.5cm] (SingleS) at (ExpS-|Single) {Single-round expected Bell value $S_{\lambda_i}$};
		\draw[->] (SingleS.south) -- (Single.north) node[midway,right] {Lemma \ref{lemma:Correlation}};
		\draw[->] (SingleS.west) -- (ExpS.east) node[midway,above] {Averaging over $i$ in Eq.~\eqref{eq:avgexpS}};
	\end{tikzpicture}
	\caption{The flowchart of the finite size analysis. To estimate the key privacy using the observed statistics in the Bell-test rounds, we take the following steps. (1) By using Azuma's inequality, the average expected Bell value over rounds can be estimated from the observed Bell test statistics in Lemma~\ref{Lemma:BellValueBound} and Corollary~\ref{Corollary:EstBellValue}. (2) In each single round, the expected Bell value can be linked with the phase-error probability in Lemma~\ref{lemma:Correlation}. (3) Based on the phase error, the statistics of its $\xi$-regularised sample entropy is constructed via Definition~\ref{Def:SampleEntropy} and \ref{Def:RelaxedSampleEntropy}. Combined with (1), the expected value of $\xi$-regularised phase-error sample entropy over all key generation rounds can be estimated. (4) By using Azuma's inequality one more time, an upper bound on the value of $\xi$-regularised phase-error sample entropy is derived via Lemma~\ref{Lemma:SmoothProperty},~\ref{Lemma:MartingaleConstruct},~\ref{Lemma:SampleEntropyBound}. (5) The $\xi$-regularised phase-error sample entropy can be used to upper bound the cardinality of the phase-error $\varepsilon$-smallest probable set via Lemma~\ref{Lemma:PatternCount}. This gives the key privacy estimation result in Theorem~\ref{thm:PrivacyAmpCost}.} \label{Fig:FiniteSizeFlowchart}
\end{figure}
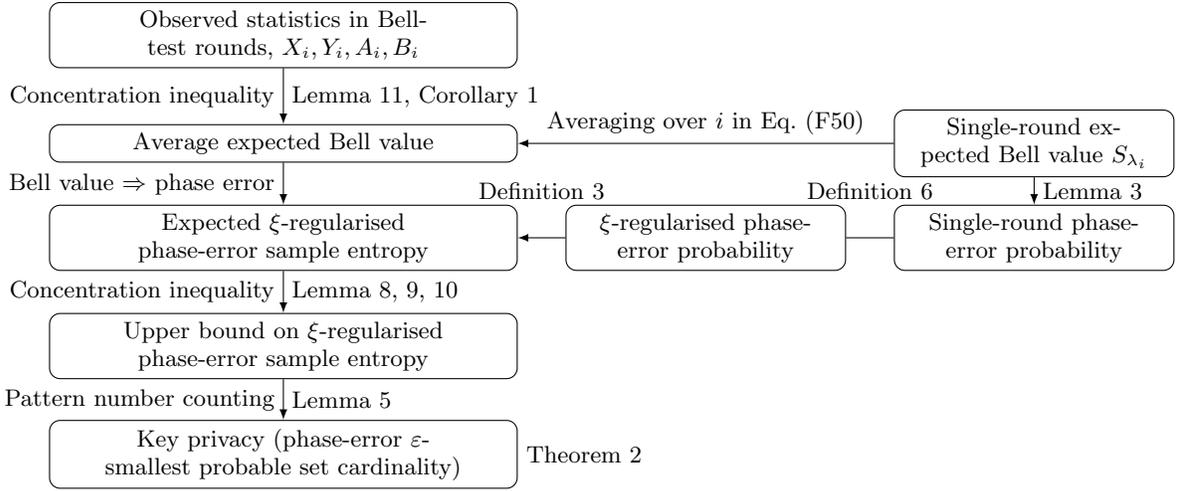

%

\subsection{Problem formulation and preliminaries}\label{Sec:FiniteSizeProb}
The target of QKD is essentially a key agreement problem.
The problem can be expressed as an error correction protocol based on the entanglement distillation protocol, where inconsistency of the keys and information leakage are characterised by bit error and phase error to be corrected, respectively~\cite{lo1999unconditional,shor2000simple}.
The objective is to determine the error pattern via a minimal amount of (bidirectional) communication.
Under this framework, the key privacy is determined by the amount of communication needed for this phase error correction.
The problem of key privacy estimation can be mathematically formulated into bounding the cardinality of the $\varepsilon$-smallest probable (typical) set of phase-error patterns. First, we give the definition of the $\varepsilon$-smallest probable set.

\begin{definition}[$\varepsilon$-Smallest Probable Set]
Given $\varepsilon\in[0,\frac{1}{2})$ and a random variable $X\in\mathcal{X}$ defined over a finite alphabet set $|\mathcal{X}|<\infty$ with a probability mass function $p_{X}$, the $\varepsilon$-smallest probable set of $X$ is defined as the smallest set of values $\mathcal{T}_X^\varepsilon$ such that a realisation of $X$ lies in it with a failure probability no larger than $\varepsilon$,
  \begin{equation}
  \begin{split}
  &\mathcal{T}_X^\varepsilon=\arg\min_{\mathcal{T}\subseteq\mathcal{X}}|\mathcal{T}|,\\
  &\text{s.t. }\Pr(X\in \mathcal{T})\geq 1-\varepsilon.
  \end{split}
  \end{equation}
\end{definition}



Now we formulate the problem of privacy estimation of the bits from the key generation rounds in DIQKD.

\begin{problem}\label{OriginalProblem}
In an $n$-round experiment, suppose a set of events $\{F_i = f_i\}_{i=1}^n$ is observed and $m$ rounds are randomly chosen for key generation $(X_{i_k},Y_{i_k}) = (2,0)$, where $i_k \in \{i_{1},i_{2},\cdots,i_{m}\}$. Given $\varepsilon_{pe}\in[0,\frac{1}{2})$ and the observed data $\{f_i\}$, upper bound the cardinality of the $\varepsilon_{pe}$-smallest probable set of the phase-error patterns $\vec{{E}}_{p}^{\lambda_{i_k}} = \left(E_{p}^{\lambda_{i_1}},\cdots,E_{p}^{\lambda_{i_{m(2,0)}}}\right)$.
\end{problem}


After the cardinality of the $\varepsilon_{pe}$-smallest probable set of the phase-error random variables is determined, a privacy amplification procedure can be efficiently carried out by using the family of two-universal hashing functions, corresponding to the phase error correction in an entanglement distillation protocol~\cite{lo1999unconditional,shor2000simple}. Given a fixed failure probability of phase error correction $\varepsilon_{pc}\in(0,1)$, the cost for privacy amplification is the logarithm of the $\varepsilon_{pe}$-smallest set cardinality $\left|\mathcal{T}_{\vec{{E}}_{p}^{\lambda_{i_k}}}^{\varepsilon_{pe}}\right|$ plus a constant of $-\log\varepsilon_{pc}$~\cite{fung2010practical}.
We emphasize that $\varepsilon_{pe}$ and $\varepsilon_{pc}$ are two different failure probabilities, corresponding to privacy estimation and phase error correction, respectively. In the following we shall focus on the first term and add the second one in the end. Then, the privacy estimation problem becomes the following one,
\begin{equation}\label{Eq:PrivacyEstProb}
\begin{split}
I_{pa} &= \sup_{\Pr{\vec{{E}}_{p}^{\lambda_{i}}}_i, \,\Lambda_i = \lambda_i} \log \left| \mathcal{T}_{\vec{{E}}_{p}^{\lambda_{i_k}}}^{\varepsilon_{pe}}\right| -\log\varepsilon_{pc}, \\
\text{s.t. }F_i &= f_i.
\end{split}
\end{equation}
The optimisation takes over all possible joint probability distributions of $\vec{{E}}_{p}^{\lambda_{i}}$ and system variables values $\Lambda_i=\lambda_i$ in each round that are compatible with the observed statistics, $F_i = f_i$. In the sequential DIQKD protocol, the relationship between the probability a phase error occurs and the expected Bell value in a round conditioned on the past events, as given by Eq.~\eqref{Eq:SuppPhaseError}, can be utilised for analysis.

In practice, a good upper bound on the solution to the optimisation Eq.~\eqref{Eq:PrivacyEstProb} suffices for privacy estimation.
A common method in QKD security analysis is to estimate the total number of phase errors, which can be used to derive an upper bound on the privacy amplification cost. In the DIQKD protocol, an upper bound on the total number of phase errors can be derived from the relationship between phase-error probability and expected Bell value in each round given by Eq.~\eqref{Eq:SuppPhaseError}. However, under a fixed value of the system variable $\Lambda_i=\lambda_i$, the phase-error probability is convex in the expected Bell value. Since the system variables values cannot be accessed and only the average Bell value over rounds can be estimated from observed statistics, in the worst case analysis, one needs to take the function concave closure of Eq.~\eqref{Eq:SuppPhaseError} for the estimation of the total number of phase errors. To help understand this point, we provide an example which can actually happen. Suppose in an experiment, all the quantum devices and the quantum systems are first implemented according to the ideal setting, which gives the expected Bell value of $2\sqrt{2}$ in every round. After one half of the rounds in the experiment, suppose the measurement devices break down and give the same output of $1$ irrespective of the measurement bases settings. We would expect that the first half of the rounds behave ideally and the users should obtain perfectly secure key bits in these rounds. In average, the users should obtain about $1/2$ bit of secure key in every key generation round. While if the users evaluate the phase-error rate over the rounds, they shall determine a phase-error rate of approximately $1/4$. When using this parameter for key privacy estimation, the average information leakage would be approximately $h(1/4)$ bit in every key generation round~\cite{fung2010practical}.
Though the result is tight in the number of phase errors, it does not appear to reach the optimal privacy estimation result. The convexity problem is not removed if one wishes to average the possible values of the system variable in a round and derive a $\Lambda_i$-independent phase-error probability.


Instead of estimating the number of phase errors, we shall take advantage of the concept of sample entropy, a notion originally appeared in classical Shannon theory~\cite{shannon1948mathematical} (see, Chapter~3 of Ref.~\cite{cover2012elements} and Chapter 14.2 of Ref.~\cite{wilde2011classical}).

\begin{definition}[Sample Entropy]
Given a random variable $X\in\mathcal{X}$ with $|\mathcal{X}|<\infty$ and a probability mass function $p_X$, the sample entropy of a realisation $X=x$ is defined as
\begin{equation}\label{Eq:SampleEntropy}
     \bar{H}(x) =
        -\log[p_X(x)].
\end{equation}
\label{Def:SampleEntropy}
\end{definition}

The sample entropy can be understood as another random variable derived from $X$ with the same probability mass function $p_X$, which we denote as $\bar{H}(X)$. The name `sample entropy' originates from the fact that the expected value of $\bar{H}(X)$ equals the Shannon entropy of the random variable $X$,
\begin{equation}
\begin{split}
  \mathbb{E}_{p_X}[\bar{H}(X)] &= \sum_{x\in\mathcal{X}}p_X(x)\{-\log[p_X(x)]\} \\
  &=H(X).
\end{split}
\end{equation}
In this work, we mainly deal with binary random variables. Consider $X\in\{0,1\}$ with $\Pr{X=1}=p\in(0,1)$, the Shannon entropy of $X$ equals to
\begin{equation}
\begin{split}
  H(X) &= \sum_{x\in\{0,1\}}p_X(x)(-\log[p_X(x)]) \\
  &=-(1-p)\log(1-p) - p\log p \\
  &= h(p).
\end{split}
\end{equation}
For two parameters $p$ and $\theta$ such that $0\leq p<p+\theta\leq 1/2$, the following inequality holds,
\begin{equation} \label{eq:binaryEntOpt}
\begin{split}
h(p)<-p\log(p+\theta) - (1-p)\log[(1-p)-\theta]<h(p+\theta).
\end{split}
\end{equation}
The proof is rather straightforward and hence is skipped here.

The number of possible realisations of a random variable can be estimated from upper-bounding the sample entropy. Here, we present it as the following lemma.
\begin{lemma}
Given $n$ bits $x_1,\cdots,x_n\in\mathcal{X}=\{0,1\}$, $n$ constants $t_1,\cdots,t_n\in(0,1)$, and $C_U>0$, define the set $\mathcal{T}\subseteq\mathcal{X}^n$
\begin{equation}
\mathcal{T} \equiv \left\{(x_1,\cdots,x_n)\big|-\sum_{i=1}^n\log[q(x_i)] \leq C_U\right\},
\end{equation}
where
\begin{eqnarray}
q(x_i) \equiv
\left\{
\begin{tabular}{ll} $1-t_i$, & $x_i=0$, \\
$t_i$, & $x_i=1$.
\end{tabular}
\right.
\end{eqnarray}
  The cardinality of $\mathcal{T}$ can be upper bounded,
  \begin{equation} \label{eq:cardCU}
    |\mathcal{T}|\leq 2^{C_U}.
  \end{equation}
\label{Lemma:PatternCount}
\end{lemma}

\begin{proof}
By the fact that,
\begin{equation}
    1 = \sum_{(x_1,\cdots,x_n)\in\{0,1\}^n} 2^{\sum_{i=1}^n \log[q(x_i)]},
\end{equation}
and $\mathcal{T}\subseteq\{0,1\}^n$, we have
  \begin{equation}
  \begin{split}
    1 &= \sum_{(x_1,\cdots,x_n)\in\{0,1\}^n} 2^{\sum_{i=1}^n \log[q(x_i)]} \\
    &\geq \sum_{(x_1,\cdots,x_n)\in\mathcal{T}} 2^{\sum_{i=1}^n \log[q(x_i)]} \\
    &\geq \sum_{(x_1,\cdots,x_n)\in\mathcal{T}} 2^{-C_U} \\
    &=|\mathcal{T}|2^{-C_U}.
  \end{split}
  \end{equation}
Therefore, Eq.~\eqref{eq:cardCU} holds.
\end{proof}

We remark that the function $q(x_i)$ can be understood as the underlying probability corresponding to a random variable. Then, $-\log[q(x_i)]$ becomes the sample entropy of the realization $x_i$. The lemma holds for the general case of non-i.i.d.~random variables, which is later used to tackle the non-i.i.d.~condition in the finite-size analysis with Azuma's concentration inequality.

\begin{definition}[Filtration]
Let $(\Omega,\mathcal{A},P)$ be a probability triplet. For every natural number $i\in \mathbb{N}$, let $\mathcal{F}_i$ be a sub-$\sigma$-algebra of $\mathcal{A}$. Then the set
\begin{equation}
  \mathbb{F} \equiv \{\mathcal{F}_i\}_{i\in\mathbb{N}}
\end{equation}
is called a filtration if $\mathcal{F}_k\subseteq\mathcal{F}_l,\,\forall k\leq l$.
\end{definition}

\begin{definition}[Martingale]
Let $(\Omega,\mathcal{A},P)$ be a probability triplet and let $\mathbb{F} = \{\mathcal{F}_i\}_{i\in\mathbb{N}}$ be a filtration. For a sequence of random variables $X_0,X_1,X_2,\cdots$ such that $\forall i$, $X_i$ is an $\mathcal{F}_i$-measurable function, the sequence is called a martingale with respect to the filtration $\mathbb{F}$ if $\forall i$,
\begin{equation}
\begin{split}
\mathbb{E}(|X_i|) &<\infty, \\
\mathbb{E}(X_{i+1}|\mathcal{F}_i) &= X_i.
\end{split}
\end{equation}
\end{definition}

\begin{lemma}[Azuma's Inequality]
Given a martingale $\{X_0,X_1,\cdots\}$ with respect to filtration $\mathbb{F}=\{\mathcal{F}_i\}_{i\in\mathbb{N}}$ and finite constants $ c_i\in[0,\infty),i\in\mathbb{N}^+$, if there exist two sets of predictable processes $\{A_i\}_{i\in\mathbb{N}^+}$ and $\{B_i\}_{i\in\mathbb{N}^+}$ where $A_i$ and $B_i$ are $\mathcal{F}_{i-1}$-measurable functions satisfying,
\begin{equation} \label{eq:boundedmartingale}
  \begin{split}
    A_t\leq X_t-X_{t-1}\leq B_t, \\
    B_t - A_t \leq c_t,
  \end{split}
\end{equation}
then for all $\delta>0$ and $n\in\mathbb{N}^+$,
  \begin{equation}
    \Pr(X_n - X_0\geq\delta)\leq \exp\left(-\dfrac{2\delta^2}{\sum_{t=1}^n c_t^2}\right).
  \end{equation}
  \label{Lemma:Azuma}
\end{lemma}

In essence, Azuma's inequality relates the expected values with frequencies. Together with the single-round result in Lemma~\ref{lemma:Correlation}, we take the procedure shown in Figure \ref{Fig:FiniteSizeFlowchart} to derive an upper bound to Problem~\ref{OriginalProblem}.

Because Azuma's inequality is relatively simple in its expression and well-known in the literature, we shall mainly explain how our finite-size analysis works based on this inequality. While for a tighter key rate, one can apply refined versions of Azuma's inequality. In our work, we consider the use of recently developed Kato's inequality~\cite{kato2020concentration}.

\begin{lemma}[Kato's Inequality~\cite{kato2020concentration}]
Let $\{X_i\}$ be a list of random variables, and $\{\mathcal{F}_i\}$ be a filtration that identifies random variables, including those variables $\{X_1,\cdots,X_i\}$, i.e., $\mathcal{F}_i$ is a $\sigma$-algebra that satisfies $\mathcal{F}_i\subseteq\mathcal{F}_{i+1}$ and $E(X_{i'}|\mathcal{F}_i)=X_{i'}$ for $i'\leq i$. Suppose that the relation $0\leq X_i\leq1$ holds for any $i$. In this case, for any $n\in\mathbb{N},\delta\in[0,\infty)$ and $\delta'\in\mathbb{R}$, the following inequality holds,
  \begin{equation}\label{eq:Kato}
    \Pr{\sum_{i=1}^n \frac{\mathbb{E}(X_i|\mathcal{F}_{i-1})}{n} \geq(1+2\delta')\sum_{i=1}^n\frac{X_i}{n}+\delta-\delta'}\leq \exp{-\frac{2n(\delta^2-\delta'^2)}{\left(1+\frac{4\delta'}{3}\right)^2}}.
  \end{equation}
\label{Lemma:Kato}
\end{lemma}

When applying Kato's inequality in our analysis, we shall take the following substitution in Eq.~\eqref{eq:Kato},
\begin{equation}
  X_i\rightarrow1-X_i,
\end{equation}
and use the following concentration result,
\begin{equation}\label{eq:reverseKato}
  \Pr{\sum_{i=1}^n \frac{\mathbb{E}(X_i|\mathcal{F}_{i-1})}{n} \leq(1+2\delta')\sum_{i=1}^n\frac{X_i}{n}-\delta-\delta'}\leq \exp{-\frac{2n(\delta^2-\delta'^2)}{\left(1+\frac{4\delta'}{3}\right)^2}}.
\end{equation}

If one has empirical knowledge of the behaviour of random variables, then one can optimise the values of $\delta$ and $\delta'$ to obtain better concentration results. To be specific, given failure probability $\varepsilon$ and the number of rounds $n$, if one has an empirical estimation of the random variables $\{X_i\}$, where $\sum_{i=1}^n\frac{X_i}{n}$ shall be close to $\tilde{X}$, then the following optimisation gives the optimal values of $\delta$ and $\delta'$ in Eq.~\eqref{eq:reverseKato},
\begin{equation}\label{Eq:MaxOpt}
\begin{split}		\arg\max_{\delta,\delta'}&\left(1+2\delta'\right)\tilde{X}-\delta-\delta', \\
\text{s.t. }\varepsilon&=\exp{-\frac{2n(\delta^2-\delta'^2)}{\left(1+\frac{4\delta'}{3}\right)^2}}, \\
\delta&>0.
\end{split}
\end{equation}
Similarly, if one has empirical estimation of the conditional expectation variables $\{\mathbb{E}(X_i|\mathcal{F}_{i-1})\}$, where $\sum_{i=1}^n \frac{\mathbb{E}(X_i|\mathcal{F}_{i-1})}{n}$ shall be close to $\tilde{\chi}$, then the following optimisation gives the optimal values of $\delta$ and $\delta'$ in Eq.~\eqref{eq:reverseKato},
\begin{equation}\label{Eq:MinOpt}
\begin{split}		\arg\min_{\delta,\delta'}&\frac{\tilde{\chi}+\delta+\delta'}{1+2\delta'}, \\
\text{s.t. }\varepsilon&=\exp{-\frac{2n(\delta^2-\delta'^2)}{\left(1+\frac{4\delta'}{3}\right)^2}}, \\
\delta&>0.
\end{split}
\end{equation}
The solution to Eq.~\eqref{Eq:MaxOpt} is given by the following function,
\begin{equation}
  \delta'^*=\dfrac{3\left(-16t^2+72t\tilde{X}-72t\tilde{X}^2-9\sqrt{2}\sqrt{2t^2-9t\tilde{X}-8t^2\tilde{X}+45t\tilde{X}^2+8t^2\tilde{X}^2-72t\tilde{X}^3+36t\tilde{X}^4}\right)}{4(-18t+16t^2+81\tilde{X}-72t\tilde{X}-81\tilde{X}^2+72t\tilde{X}^2)},
\end{equation}
and the solution to Eq.~\eqref{Eq:MinOpt} is given by the following function,
\begin{equation}
  \delta'^*=\dfrac{3\left(-27t-4t^2+72t\tilde{\chi}-72t\tilde{\chi}^2-9\sqrt{t^2-18t\tilde{\chi}-4t^2\tilde{\chi}+90t\tilde{\chi}^2+4t^2\tilde{\chi}^2-144t\tilde{\chi}^3+72t\tilde{\chi}^4}\right)}{4(36t+4t^2+81\tilde{\chi}-72t\tilde{\chi}-81\tilde{\chi}^2+72t\tilde{\chi}^2)},
\end{equation}
where we denote
\begin{equation}
  t=\frac{\ln\varepsilon}{n}.
\end{equation}
The optimal values of $\delta$ can be solved by the constraints in Eq.~\eqref{Eq:MaxOpt} and~\eqref{Eq:MinOpt}.

\subsection{Phase-error probability regularisation}\label{Sec:epRelax}
In the finite-size analysis, we shall convert the phase-error probability (as a function of the underlying Bell value given in Eq.~\eqref{Eq:SuppPhaseError}) to the sample entropy. A direct conversion would result in the divergence of the logarithm function in Eq.~\eqref{Eq:SampleEntropy} when the phase-error rate approaches zero. Then, the requirement of Eq.~\eqref{eq:boundedmartingale} cannot be satisfied. In order to solve this divergence problem, we employ a regularised version of the relation between the phase-error probability and the expected Bell value, following Eq.~\eqref{Eq:SuppPhaseError}.

\begin{definition}
Given $\xi\in(0,\frac{1}{2})$, the $\xi$-regularised phase-error probability function with respect to the expected Bell value is defined as
  \begin{eqnarray}
    e_{p}^{\xi}(S) &\equiv& \frac12\cdot\dfrac{e_{p}(S) + \xi}{\frac12+\xi} \nonumber \\ &=& \left\{
    \begin{tabular}{ll}
        $\dfrac{1 + 2\xi -\sqrt{(S/2)^2-1}}{2(1+2\xi)}$, &$2<S\leq2\sqrt{2},$\\
        $\dfrac{1}{2},$ &$0\leq S\leq2$,
        \end{tabular}
          \right.
\label{Eq:SmoothPhaseErrorProb}
\end{eqnarray}
where $e_{p}(S)$ is the phase-error probability function with respect to $S\in[0,2\sqrt{2}]$ given in Eq.~\eqref{Eq:SuppPhaseError}. Here, we use the fact that $\max_S{e_{p}(S)}=\frac{1}{2}$.
\label{Def:RelaxedSampleEntropy}
\end{definition}

Here are some properties of $e_{p}^{\xi}(S)$ that will be used for later analysis.

\begin{lemma}[Properties of $e_{p}^{\xi}(S)$]
Given $\xi\in(0,\frac{1}{2})$, the function $e_{p}^{\xi}(S)$ defined over the interval $S\in[0,2\sqrt{2}]$ in Eq.~\eqref{Eq:SmoothPhaseErrorProb} has the following properties:
\begin{enumerate}
\item Monotonicity: $e_{p}^{\xi}(S)$ is monotonically decreasing with respect to $S$.
\item Finite range strictly above $0$:
\begin{equation} \label{eq:xiephrange}
\begin{split}
e_{p}^{\xi}(S)\in\left[\dfrac{\xi}{1+2\xi},\dfrac{1}{2}\right].
\end{split}
\end{equation}
\item Upper bound of $e_{p}(S)$:
\begin{equation} \label{eq:xiephrelax}
\begin{split}
e_{p}^{\xi}(S)\geq e_{p}(S).
\end{split}
\end{equation}
\item
Concavity: $\forall S_1,S_2$,
\begin{equation} \label{eq:concaveepS}
\begin{split}
h\left[e_{p}^{\xi}\left(\dfrac{1}{2}(S_1+S_2)\right)\right]\geq \dfrac{1}{2}\left(h\left[e_{p}^{\xi}(S_1)\right] + h\left[e_{p}^{\xi}(S_2)\right]\right).
\end{split}
\end{equation}
\end{enumerate}
\label{Lemma:SmoothProperty}
\end{lemma}

The proofs are rather straightforward. Eq.~\eqref{eq:concaveepS} can be proved by taking the second derivative of $S$. The last property that $h\left[e_{p}^{\xi}(S)\right]$ is concave in $S$ shall play a vital role in Lemma \ref{Lemma:SampleEntropyBound}. In Figure \ref{Fig:RelaxedSampleEnt}, we plot the function over the interval $S\in[2,2\sqrt{2}]$ for $\xi=1\times10^{-7}$.

\begin{figure}[htbp]
\centering \includegraphics[width=8cm]{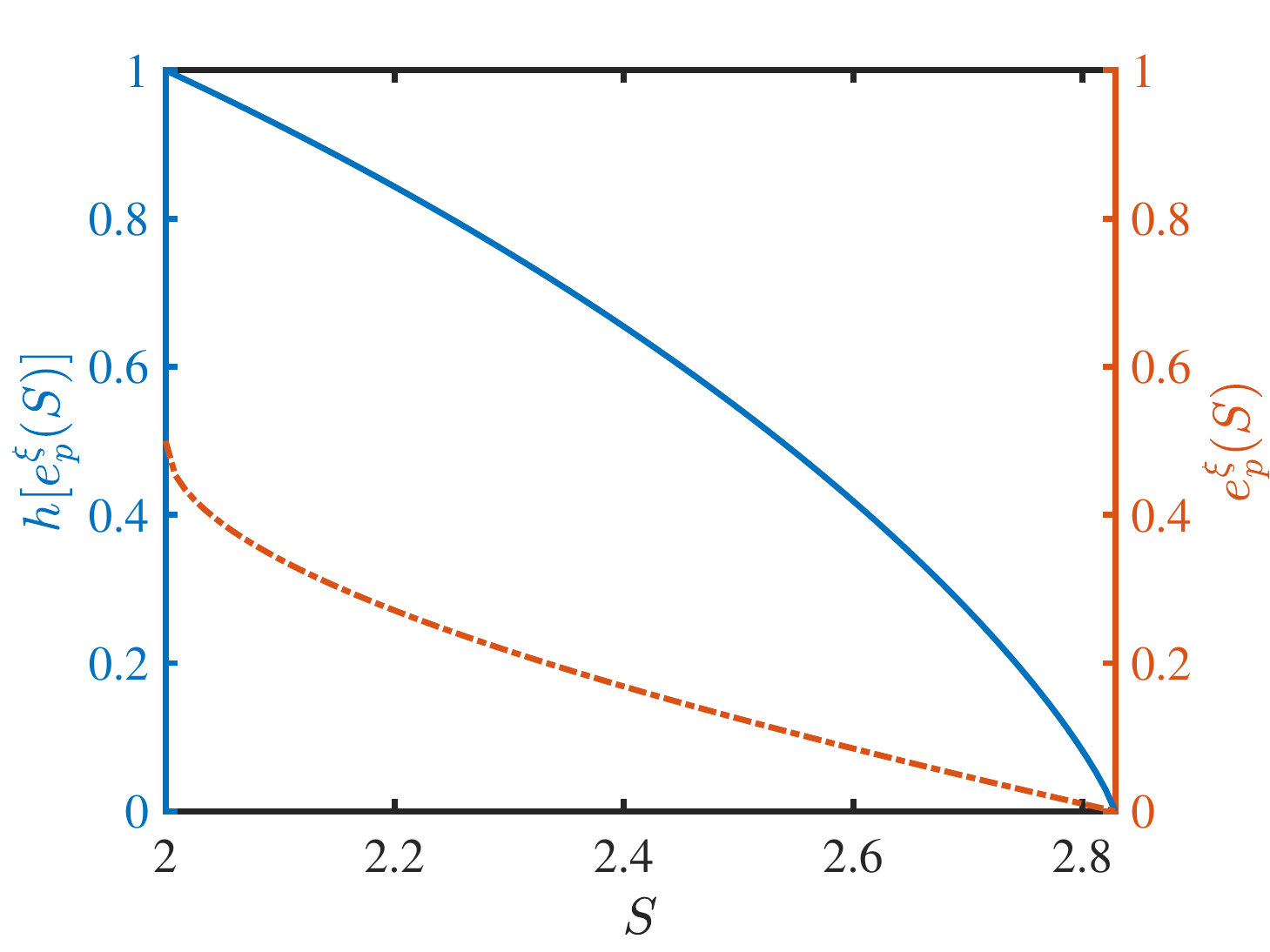}
\caption{A diagram of the functions $h[e_{p}^{\xi}(S)],e_{p}^{\xi}(S)$ on the interval $S\in[2,2\sqrt{2}]$. The functions are given in Eq.~\eqref{Eq:SmoothPhaseErrorProb}\eqref{Eq:SuppPhaseError}, respectively, and we omit the interval $S\in[0,2]$ in the figure where $h[e_{p}^{\xi}(S)]\equiv1,e_{p}^{\xi}(S)\equiv\frac{1}{2}$.
In the showcase, the parameter $\xi=1\times10^{-7}$. It can be seen that the function $h[e_{p}^{\xi}(S)]$ is concave in $S$, while $e_{p}^{\xi}(S)$ is convex in $S$.} \label{Fig:RelaxedSampleEnt}
\end{figure}

We emphasize that $\xi$ is not a security parameter but a value one can choose freely. In real experiments, $\xi$ could be optimised given the number of rounds and security parameters for a better key rate. Normally, $\xi$ is a small value in comparison to the phase-error probability so that the key rate is not affected much. Furthermore, note that the regularised phase-error probability function needs not necessarily be of the form in Eq.~\eqref{Eq:SmoothPhaseErrorProb}. It can be chosen in other forms, as long as it has similar properties as Eq.~\eqref{Eq:SmoothPhaseErrorProb} shown in Lemma \ref{Lemma:SmoothProperty}. For instance, $\xi$ does not necessarily be a constant. One can let $\xi$ be a function with respect to $S$. We leave the optimization of the regularisation forms for future research.

\subsection{Phase-error sample entropy}\label{Sec:PhaseErrorSample}
Now let us construct the phase-error sample entropy for each round of experiment based on Eq.~\eqref{Eq:SmoothPhaseErrorProb} and show that it satisfies the criteria of Eq.~\eqref{eq:boundedmartingale} so that we can apply Azuma's inequality.

Consider a DIQKD experiment of $n$ rounds. Before the $i$th round, the experiment runs the same as in the actual protocol. In the $i$th round, the projection $\{\hat{\Pi}_{\lambda_i}^A\otimes\hat{\Pi}_{\lambda_i}^B\}_{\lambda_i}$ is first applied to the system. Construct the filtration $\mathbb{F} = \{\mathcal{F}_j\}_{j=0}^{n}$ with $\mathcal{F}_{j-1}\subseteq\mathcal{F}_{j}$, where $\mathcal{F}_{j-1}$ is the $\sigma$-algebra generated by $\{\Lambda_1,\cdots,\Lambda_j\}$. By our definition of the system variable, the random variable $\Lambda_j$ contains all the information in the system variables with a smaller subscript, hence we can take $\mathcal{F}_{j-1}=\Lambda_j$ for simplicity. For the sake of completeness, define
$\Lambda_{n+1}$ to be an arbitrary quantum operation on the leftover quantum system after the experiment is finished. For the key generation rounds $(X_i,Y_i)=(2,0)$, Alice and Bob need to learn whether a phase error occurs, \emph{had} they applied the phase-error measurement to the system. For other basis choices, they just apply the original measurements.

On the sub-$\sigma$-algebra $\mathcal{F}_{i-1},i=1,\cdots,n$, define the following random variable
\begin{eqnarray}\label{Eq:RelaxSampleEntropy}
\zeta_{p}^{(i)} =
\left\{
\begin{tabular}{ll} $0$, & if $(X_i,Y_i)\neq(2,0)$, \\
$-\log\left[e_{p}^{\xi}(S^{\lambda_i})\right]$, & if $(X_i,Y_i)=(2,0),\,E_{p}^{\lambda_i}=1$, \\
$-\log\left[1 - e_{p}^{\xi}(S^{\lambda_i})\right]$, & if $(X_i,Y_i)=(2,0),\,E_{p}^{\lambda_i}=0$,
\end{tabular}
\right.
\end{eqnarray}
where $S^{\lambda_i}$ is the expected Bell value in the $i$th round given in Eq.~\eqref{def:UnderBellValue}. For the key generation rounds $(X_i,Y_i)=(2,0)$, the random variable $\zeta_{p}^{(i)}$ becomes the $\xi$-regularised phase-error sample entropy random variable,
\begin{eqnarray}
{\zeta_{p}^{(i)}}_{|(X_i,Y_i)=(2,0)} = \bar{H}^{\xi}(E_{p}^{\lambda_i}) =
\left\{
\begin{tabular}{ll}
$-\log\left[e_{p}^{\xi}(S^{\lambda_i})\right]$, &$E_{p}^{\lambda_i}=1,$\\
$-\log\left[1-e_{p}^{\xi}(S^{\lambda_i})\right]$, &$E_{p}^{\lambda_i}=0.$
\end{tabular}
\right.
\end{eqnarray}
Now, we can construct a sequence of random variables,
\begin{eqnarray} \label{eq:DeltaRV}
  \Delta_{p}^{(l)} = \left\{
        \begin{tabular}{ll} $0$, & $l=0$, \\
        $\sum_{i=1}^{l}\left(\zeta_{p}^{(i)} - \mathbb{E}_{\Pr{X_i,Y_i,E_{p}^{\lambda_i}}}\left[\zeta_{p}^{(i)}|\mathcal{F}_{i-1}\right]\right)$, & $l=1,\cdots,n$.
        \end{tabular}
        \right.
\end{eqnarray}

\begin{lemma} \label{Lemma:MartingaleConstruct}
The random variables $\{\Delta_{p}^{(l)}\}_{l\in[n]\cup\{0\}}$ form a bounded martingale with respect to the filtration $\mathbb{F}$. For all $\xi\in(0,\frac{1}{2})$ and $l\geq1$, there are two sets of constants $\{C_L^{(l)}\}_{l\in[n]}$ and $\{C_U^{(l)}\}_{l\in[n]}$, such that
\begin{equation}
\begin{split}
  C_L^{(l)}\leq\Delta_{p}^{(l)}-\Delta_{p}^{(l-1)}&\leq C_U^{(l)}, \\
  C_U^{(l)}-C_L^{(l)}&\leq c_\xi,
\end{split}
\end{equation}
where
\begin{equation}\label{Eq:CXi}
c_\xi = -\log\left(\dfrac{\xi}{1+2\xi}\right).
\end{equation}
\end{lemma}

\begin{proof}
According to the definition in Eq.~\eqref{eq:DeltaRV}, the random variable $\Delta_{p}^{(l)}$ has the following recurrence relationship,
\begin{equation}
  \begin{split}
    \Delta_{p}^{(l)} &= \Delta_{p}^{(l-1)} + \zeta_{p}^{(l)} - \mathbb{E}_{\Pr{X_l,Y_l,E_{p}^{\lambda_l}}}\left[\zeta_{p}^{(l)}|\mathcal{F}_{l-1}\right].
  \end{split}
\end{equation}
It follows that
\begin{equation}
\begin{split}
\mathbb{E}_{\Pr{X_l,Y_l,E_{p}^{\lambda_l}}}\left[\Delta_{p}^{(l)}|\mathcal{F}_{l-1}\right] &= \mathbb{E}_{\Pr{X_l,Y_l,E_{p}^{\lambda_l}}}\left[\Delta_{p}^{(l-1)}|\mathcal{F}_{l-1}\right] + \mathbb{E}_{\Pr{X_l,Y_l,E_{p}^{\lambda_l}}}\left[\zeta_{p}^{(l)}|\mathcal{F}_{l-1}\right] - \mathbb{E}_{\Pr{X_l,Y_l,E_{p}^{\lambda_l}}}\left[\zeta_{p}^{(l)}|\mathcal{F}_{l-1}\right] \\
  &= \Delta_{p}^{(l-1)},
\end{split}
\end{equation}
and the value of $\Delta_{p}^{(l)}$ is bounded, as
\begin{equation}
\begin{split}
  \Delta_{p}^{(l)} &\geq \Delta_{p}^{(l-1)} - p_X(2)p_Y(0)\left\{-e_{p}^{\lambda_l}\log\left[e_{p}^{\xi}(S^{\lambda_l})\right] - \left(1-e_{p}^{\lambda_l}\right)\log\left[1 - e_{p}^{\xi}(S^{\lambda_l})\right]\right\} \\ &\geq \Delta_{p}^{(l-1)} - p_X(2)p_Y(0)h\left[e_{p}^{\xi}(S^{\lambda_l})\right],
\end{split}
\end{equation}
and
\begin{equation}
\begin{split}
\Delta_{p}^{(l)}&\leq \Delta_{p}^{(l-1)} - \log\left(\dfrac{\xi}{1+2\xi}\right) - p_X(2)p_Y(0)\left\{-e_{p}^{\lambda_l}\log\left[e_{p}^{\xi}(S^{\lambda_l})\right] - \left(1-e_{p}^{\lambda_l}\right)\log\left[1 - e_{p}^{\xi}(S^{\lambda_l})\right]\right\} \\ &\leq \Delta_{p}^{(l-1)} - \log\left(\dfrac{\xi}{1+2\xi}\right) - p_X(2)p_Y(0)h\left(e_{p}^{\lambda_l}\right).
\end{split}
\end{equation}
Here, we denote the actual phase-error probability of the underlying system as $e_{p}^{\lambda_l}$, which satisfies $e_{p}^{\lambda_l}\leq e_{p}^{\xi}(S^{\lambda_l})$. With $\Delta_{p}^{(0)}=0$, we can see that the random process $\{\Delta_{p}^{(l)}\}_l$ form a bounded martingale.

For any $l\in[n]$, the difference between two adjacent random variables in the process is given by
\begin{equation}
\begin{split}
  \Delta_{p}^{(l)} - \Delta_{p}^{(l-1)} &= \zeta_{p}^{(l)} - \mathbb{E}_{\Pr{X_l,Y_l,E_{p}^{\lambda_l}}}\left[\zeta_{p}^{(l)}|\mathcal{F}_{l-1}\right] \\
  &= \zeta_{p}^{(l)} - p_X(2)p_Y(0)\left\{-e_{p}^{\lambda_l}\log\left[e_{p}^{\xi}(S^{\lambda_l})\right] - \left(1-e_{p}^{\lambda_l}\right)\log\left[1 - e_{p}^{\xi}(S^{\lambda_l})\right]\right\}.
\end{split}
\end{equation}
Let
\begin{equation}
\begin{split}
C_L^{(l)} &= - p_X(2)p_Y(0)\left\{-e_{p}^{\lambda_l}\log\left[e_{p}^{\xi}(S^{\lambda_l})\right] - \left(1-e_{p}^{\lambda_l}\right)\log\left[1 - e_{p}^{\xi}(S^{\lambda_l})\right]\right\}, \\
C_U^{(l)} &= - p_X(2)p_Y(0)\left\{-e_{p}^{\lambda_l}\log\left[e_{p}^{\xi}(S^{\lambda_l})\right] - \left(1-e_{p}^{\lambda_l}\right)\log\left[1 - e_{p}^{\xi}(S^{\lambda_l})\right]\right\}-\log\left(\dfrac{\xi}{1+2\xi}\right),
\end{split}
\end{equation}
we have
\begin{equation}
C_L^{(l)}\leq\Delta_{p}^{(l)}-\Delta_{p}^{(l-1)}\leq C_U^{(l)}.
\end{equation}
By construction, $C_L^{(l)},C_U^{(l)}$ are constants, and
\begin{equation}
C_U^{(l)}-C_L^{(l)} = c_\xi = -\log\left(\dfrac{\xi}{1+2\xi}\right).
\end{equation}

\end{proof}

\begin{lemma}
Given $\varepsilon\in(0,1)$, with a probability no smaller than $1-\varepsilon$, the sum of $\xi$-regularised phase-error sample entropy in $n$ rounds can be upper bounded by the following expression,
\begin{equation}
  \sum_{i=1}^n\zeta_{p}^{(i)} < np_X(2)p_Y(0) h\left[ e_{p}^{\xi}\left(\dfrac{1}{n}\sum_{i=1}^n S^{\lambda_i}\right)\right] + \sqrt{-\dfrac{nc_\xi^2\ln\varepsilon}{2}},
\end{equation}
where $c_{\xi}$ is given in Eq.~\eqref{Eq:CXi}.
\label{Lemma:SampleEntropyBound}
\end{lemma}

\begin{proof}
According to Lemma \ref{Lemma:MartingaleConstruct}, we can apply Azuma's concentration result in Lemma~\ref{Lemma:Azuma} to the bounded martingale $\{\Delta_{p}^{(l)}\}_l$,
\begin{equation}
  \sum_{i=1}^n\zeta_{p}^{(i)} \leq \sum_{i=1}^n\mathbb{E}_{\Pr{X_i,Y_i,E_{p}^{\lambda_i}}}\left[\zeta_{p}^{(i)}|\mathcal{F}_{i-1}\right] + \sqrt{-\dfrac{nc_\xi^2\ln\varepsilon}{2}},
\end{equation}
where Eq.~\eqref{eq:DeltaRV} is applied. According to the definition of $\zeta_{p}^{(i)}$ in Eq.~\eqref{Eq:RelaxSampleEntropy} and applying Eq.~\eqref{eq:binaryEntOpt}, we have
\begin{equation}\label{Eq:SmoothRelation}
\begin{split}
  \mathbb{E}_{\Pr{X_i,Y_i,E_{p}^{\lambda_i}}}\left[\zeta_{p}^{(i)}|\mathcal{F}_{i-1}\right] &= p_X(2)p_Y(0) \left\{-e_{p}^{\lambda_i}\log\left[e_{p}^{\xi}(S^{\lambda_i})\right]-\left(1-e_{p}^{\lambda_i}\right)\log\left[1-e_{p}^{\xi}(S^{\lambda_i})\right]\right\} \\ &< p_X(2)p_Y(0) h\left[e_{p}^{\xi}(S^{\lambda_i})\right].
\end{split}
\end{equation}
Taking the summation of Eq.~\eqref{Eq:SmoothRelation} over rounds and further applying Eq.~\eqref{eq:concaveepS}, one can prove the lemma.
\end{proof}

\subsection{Estimating the average expected bell value}\label{Sec:BellValueEst}
The result in the last subsection shows that the total $\xi$-regularised phase-error sample entropy can be estimated from the average expected Bell value over the rounds. Still, the expected Bell value cannot be accessed directly.
In this subsection we show how to estimate this value from the Bell test results via Azuma's concentration inequality.

Given the basis choice $(x,y)$ with $x,y\in\{0,1\}$, for any $i\in[n]$, define the following random variables conditioned on $\sigma$-algebra $\mathcal{F}_{i-1}$,
\begin{eqnarray}\label{Eq:TestRoundsRV}
  \zeta_{xy}^{(i)} = \left\{
        \begin{tabular}{ll} $0$, & if $(X_i,Y_i)\neq(x,y)$ or $(X_i,Y_i,A_i,B_i)=(x,y,a,b),(-1)^{xy}\neq ab$, \\
  $1$, & if $(X_i,Y_i,A_i,B_i)=(x,y,a,b),(-1)^{xy}= ab$,
  \end{tabular}
          \right.
\end{eqnarray}
where $a,b\in\{+1,-1\}$. From the parameter estimation in the protocol we have $\sum_{i=1}^n\zeta_{xy}^{(i)} = m_{xy} - q_{xy}$, where $m_{xy}$ is the number of rounds with the inputs $(X_i,Y_i) = (x,y)$ and $q_{xy}$ is the number of rounds where $(-1)^{x\cdot y}A_i\cdot B_i=-1$. Using the random variables defined in Eq.~\eqref{Eq:TestRoundsRV}, the expected Bell value in each round can be expressed as
\begin{equation}\label{Eq:ExpectedBellValue}
\begin{split}
S^{\lambda_i} &= \sum_{x,y\in\{0,1\}}(-1)^{xy}\mathbb{E}(A_iB_i|X_i=x,Y_i=y,\mathcal{F}_{i-1}) \\ &= \sum_{x,y\in\{0,1\}}\left(\dfrac{2\Pr{\zeta_{xy}^{(i)}=1|\mathcal{F}_{i-1}}}{p_X(x)p_Y(y)}-1\right).
\end{split}
\end{equation}
For any fixed $(x,y)$, we can construct a sequence of random variables from $\{\zeta_{xy}^{(i)}\}_i$,
\begin{eqnarray}\label{eq:DeltaBellRV}
  \Delta_{xy}^{(l)} = \left\{
        \begin{tabular}{ll} $0$, & $l=0$, \\
        $\sum_{i=1}^{l}\left(\zeta_{xy}^{(i)} - \Pr{\zeta_{xy}^{(i)} = 1|\mathcal{F}_{i-1}}\right)$, & $l\geq1$.
        \end{tabular}
        \right.
\end{eqnarray}

Similar to the proof of Lemma~\ref{Lemma:MartingaleConstruct}, one can prove that for each $(x,y)$, the sequence of random variables $\{\Delta_{xy}^{(l)}\}_l$ forms a martingale and has a bounded difference,
\begin{equation}
\begin{split}
  \Delta_{xy}^{(l)} - \Delta_{xy}^{(l-1)} &= \zeta_{xy}^{(i)} - \Pr{\zeta_{xy}^{(i)} = 1|\mathcal{F}_{l-1}} \\ &\in\left[-\max\Pr{\zeta_{xy}^{(i)} = 1|\mathcal{F}_{l-1}},1-\min\Pr{\zeta_{xy}^{(i)} = 1|\mathcal{F}_{l-1}}\right] \\ &\subseteq\left[-p_X(x)p_Y(y),1\right).
\end{split}
\end{equation}
Let
\begin{equation}
c_{xy} = 1+p_X(x)p_Y(y),
\end{equation}
by applying Azuma's inequality, we have the following concentration result:

\begin{lemma}
Given $\in(0,1)$, with a probability no smaller than $1-\varepsilon$, the following inequality holds
\begin{equation}\label{Eq:TestRoundBound}
  \sum_{i=1}^n\zeta_{xy}^{(i)} \leq \sum_{i=1}^n \Pr{\zeta_{xy}^{(i)} = 1|\mathcal{F}_{i-1}} + \sqrt{-\dfrac{nc_{xy}^2\ln\varepsilon}{2}}.
\end{equation}
\label{Lemma:BellValueBound}
\end{lemma}

Taking this concentration result into Eq.~\eqref{Eq:ExpectedBellValue}, we can obtain an estimation of the average expected Bell value in the experiment, as presented in the following corollary.

\begin{corollary}
Given $\varepsilon_{xy}\in(0,1)$ such that $\varepsilon=\sum_{x,y\in\{0,1\}}\varepsilon_{xy}\in(0,1)$, with probability no smaller than $1-\varepsilon$, the average expected Bell value over $n$ rounds can be lowered bounded by,
\begin{equation}\label{eq:avgexpS}
  S_{est}\equiv\dfrac{1}{n}\sum_{i=1}^nS^{\lambda_i} \geq \sum_{x,y\in\{0,1\}}\left(\dfrac{2\left( m_{xy}-q_{xy}-\sqrt{-\dfrac{nc_{xy}^2\ln\varepsilon_{xy}}{2}}\right)}{np_X(x)p_Y(y)}-1\right),
\end{equation}
where $S^{\lambda_i}$ is the expected Bell value of the $i$th round with system variable value $\Lambda_i=\lambda_i$, as given in Eq.~\eqref{Eq:ExpectedBellValue}.
\label{Corollary:EstBellValue}
\end{corollary}

\subsection{Key privacy estimation with observed statistics}\label{Sec:KeyPrivacyEst}
In this subsection, we combine the results in the above subsections and present the final privacy estimation result from the observed statistics. The martingale random variables associating with the phase-error sample entropy and the Bell test measurements are defined over the same filtration. Therefore, by combining the results of Lemma~\ref{Lemma:SampleEntropyBound} and Corollary \ref{Corollary:EstBellValue}, we can obtain an estimation of the total $\xi$-regularised phase-error sample entropy via the observed statistics. Afterwards, by applying Lemma~\ref{Lemma:PatternCount}, we can derive an upper bound on the $\varepsilon$-smallest probable set cardinality and hence the cost of privacy amplification. In Figure \ref{Fig:RandomSample}, we summarise the route of the finite-size analysis to the final key privacy. We present the formal statement of the estimation of the privacy amplification cost.

\begin{figure*}[!hbt]
\centering \resizebox{15cm}{!}{\includegraphics{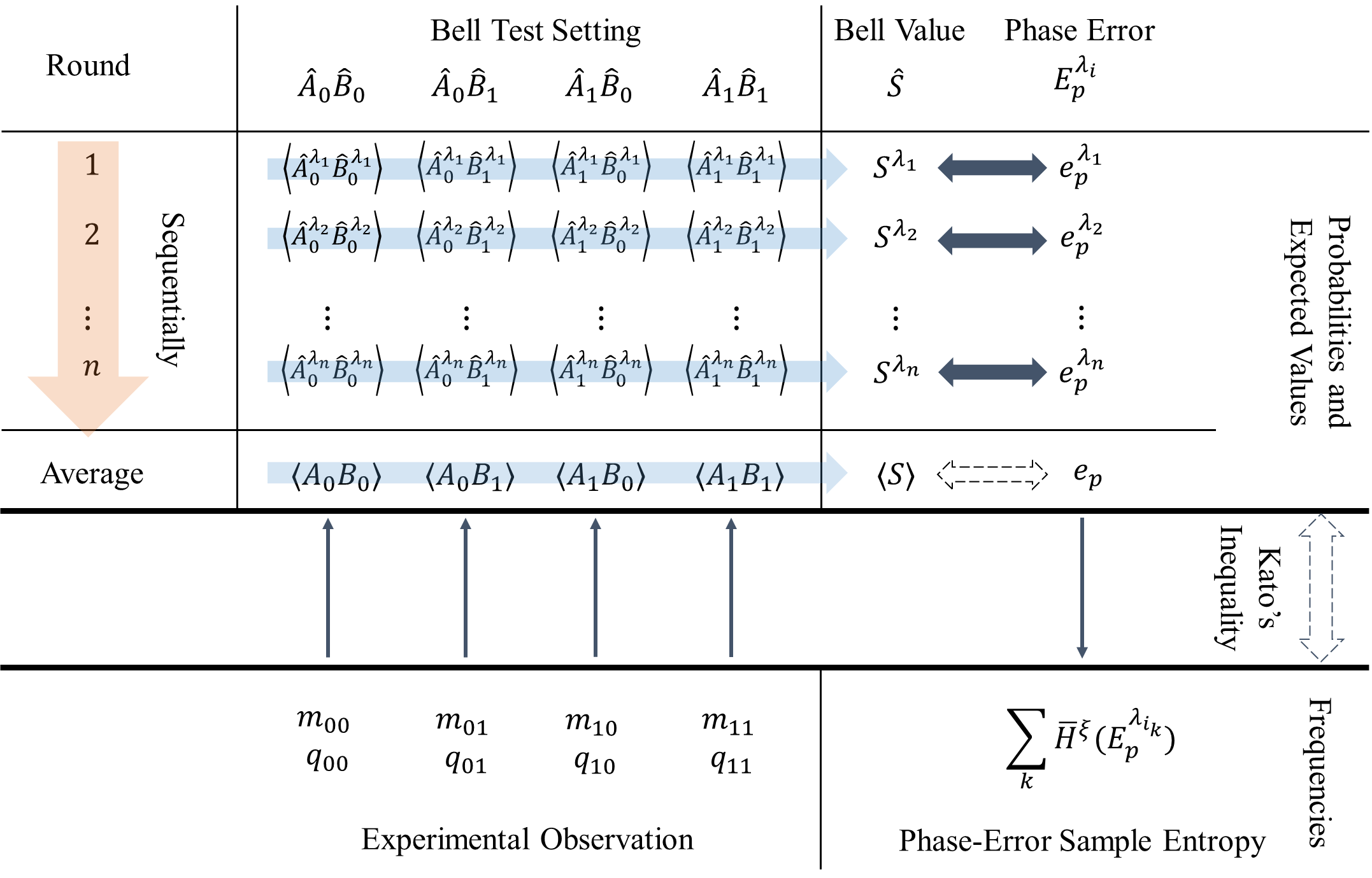}}
\caption{A diagram showing the route to the finite size analysis. The measurements are taken in a time sequence on the subsystems in each round (indicated by the orange arrow). On the upper part shows the expected values or probabilities of various variables (observable expectations in Bell tests, the Bell observable expectation, and the phase-error probability). The expected Bell value is constructed by the expected values of the possible observables in the Bell test (indicated by the blue arrow), which is equivalent to the measurement of the Bell operator in Eq.~\eqref{Eq:BellOperator}. Lemma~\ref{lemma:Correlation} bridges the expected Bell value and the phase-error probability (indicated by the two-way grey arrow). The lower part of the diagram corresponds to frequencies. They can be related to the expected values or probabilities via martingale-based concentration inequalities.}
\label{Fig:RandomSample}
\end{figure*}

\begin{theorem}[Privacy Amplification Cost]
For a DIQKD protocol of $n$ rounds, given the failure probability in the phase-error sample entropy estimation $\varepsilon_{pe}\in(0,1)$, failure probabilities for the parameter estimation of each Bell test setting $\varepsilon_{xy}\in(0,1)$, and the failure probability in privacy amplification $\varepsilon_{pc}\in(0,1)$, such that $\varepsilon = \varepsilon_{pe} + \sum_{x,y\in\{0,1\}}\varepsilon_{xy} + \varepsilon_{pc}\in(0,1)$, then for any   $\xi\in\left(0,\frac{1}{2}\right)$, with a failure probability no larger than $\varepsilon$, the cost for privacy amplification is upper bounded by
\begin{equation}\label{Eq:PrivacyAmplificationCost}
    I_{pa} \leq n\left\{p_X(2)p_Y(0) h\left[e_{p}^{\xi} \left(S_{est}\right) \right] + c_{\xi}\delta_{h} \right\} -\log\varepsilon_{pc},
\end{equation}
where $\delta_h$ is the deviation term in the estimation of $\xi$-regulraised phase-error sample entropy,
\begin{equation}
  \delta_h=\sqrt{-\dfrac{\ln\varepsilon_{pe}}{2n}},
\end{equation}
$S_{est}$ is the estimated expected Bell value,
\begin{equation}
  S_{est}=\sum_{x,y\in\{0,1\}}\left[\dfrac{2\left( m_{xy}-q_{xy}-nc_{xy}\delta_{xy}\right)}{np_X(x)p_Y(y)}-1\right],
\end{equation}
with $\delta_{xy}$ a deviation term in the estimation of expected Bell value,
\begin{equation}
  \delta_{xy}=\sqrt{-\dfrac{\ln\varepsilon_{xy}}{2n}},
\end{equation}
$c_{\xi}$ and $\forall x,y\in\{0,1\},c_{xy}$ are constants,
\begin{equation}
\begin{split}
  c_\xi &= \log\left(\dfrac{1+2\xi}{\xi}\right), \\
  c_{xy} &= 1 + p_X(x)p_Y(y).
\end{split}
\end{equation}
\label{thm:PrivacyAmpCost}
\end{theorem}

We remark that unlike $\varepsilon_{pe},\varepsilon_{pc},\varepsilon_{xy}$, the parameter $\xi$ is not a security parameter. Theorem~\ref{thm:PrivacyAmpCost} holds valid for all values of $\xi\in(0,\frac{1}{2})$. Nevertheless, when all the other protocol parameters are fixed, $\xi$ can be optimised for better key generation performance. In particular, it converges to zero with increasing data size. We leave the detailed numerical discussion for its optimisation in Sec.~\ref{Sec:Regular}.

\subsection{Refined result using Kato's inequality}\label{Sec:KatoResult}
In this section, we show a refined result after using Kato's inequality. In general, the deductions are very much similar to those using Azuma's inequality. Here, we state the modifications in each step.

\begin{lemma}[Refined version of Lemma~\ref{Lemma:SampleEntropyBound}]
Given $n\in\mathbb{N}^+,\varepsilon\in(0,1),\delta\in[0,\infty)$, and $\delta'\in\mathbb{R}$ that satisfy
\begin{equation}
  \varepsilon=\exp{-\frac{2n(\delta^2-\delta'^2)}{\left(1+\frac{4\delta'}{3}\right)^2}},
\end{equation}
with a probability no smaller than $1-\varepsilon$, the sum of $\xi$-regularised phase-error sample entropy in $n$ rounds can be upper bounded by the following expression,
  \begin{equation}
    \sum_{i=1}^n\zeta_{p}^{(i)}<\frac{np_X(2)p_Y(0) h\left[ e_{p}^{\xi}\left(\dfrac{1}{n}\sum_{i=1}^n S^{\lambda_i}\right)\right]+nc_{\xi}(\delta+\delta')}{1+2\delta'},
  \end{equation}
where $c_{\xi}$ is given in Eq.~\eqref{Eq:CXi}.
\label{Lemma:RefinedSamEnt}
\end{lemma}

\begin{proof}
By Eq.~\eqref{eq:xiephrange}, $\forall i$, we have
\begin{equation}
  c_{\xi}^{-1}\zeta_{p}^{(i)}\in[0,1].
\end{equation}
Then, we can apply Eq.~\eqref{eq:reverseKato} to the above random variables, such that with a probability no smaller than $1-\varepsilon$,
\begin{equation}
  \sum_{i=1}^n\zeta_{p}^{(i)}<\frac{\sum_{i=1}^n \mathbb{E}_{\Pr{X_i,Y_i,E_{p}^{\lambda_i}}}[\zeta_{p}^{(i)}|\mathcal{F}_{i-1}]+nc_{\xi}(\delta+\delta')}{1+2\delta'}.
\end{equation}
According to the definition of $\zeta_{p}^{(i)}$ in Eq.~\eqref{Eq:RelaxSampleEntropy} and applying Eq.~\eqref{eq:binaryEntOpt}, we have
\begin{equation}
\begin{split}
  \mathbb{E}_{\Pr{X_i,Y_i,E_{p}^{\lambda_i}}}\left[\zeta_{p}^{(i)}|\mathcal{F}_{i-1}\right] &= p_X(2)p_Y(0) \left\{-e_{p}^{\lambda_i}\log\left[e_{p}^{\xi}(S^{\lambda_i})\right]-\left(1-e_{p}^{\lambda_i}\right)\log\left[1-e_{p}^{\xi}(S^{\lambda_i})\right]\right\} \\ &< p_X(2)p_Y(0) h\left[e_{p}^{\xi}(S^{\lambda_i})\right].
\end{split}
\end{equation}
Taking the summation of the above inequality over rounds and further applying Eq.~\eqref{eq:concaveepS}, one can prove the lemma.
\end{proof}

\begin{lemma}[Refined version of Lemma~\ref{Lemma:BellValueBound}]
Given $n\in\mathbb{N}^+,\varepsilon\in(0,1),\delta\in[0,\infty)$, and $\delta'\in\mathbb{R}$ that satisfy
\begin{equation}
  \varepsilon=\exp{-\frac{2n(\delta^2-\delta'^2)}{\left(1+\frac{4\delta'}{3}\right)^2}},
\end{equation}
with a probability no smaller than $1-\varepsilon$, the following inequality holds,
  \begin{equation}
    \frac{1}{n}\sum_{i=1}^n \Pr{\zeta_{xy}^{(i)} = 1|\mathcal{F}_{i-1}}>(1+2\delta')\frac{1}{n}\sum_{i=1}^n\zeta_{xy}^{(i)}-\delta-\delta'.
  \end{equation}
\end{lemma}

The proof for this lemma is basically the same as the one for Lemma~\ref{Lemma:RefinedSamEnt}. Note that $\forall x,y\in\{0,1\}$, with respect to the filtration $\{\mathcal{F}_i\}$, the random variables $\zeta_{xy}^{(i)}$ naturally suffice the condition in Lemma~\ref{Lemma:Kato}.

\begin{corollary}[Refined version of Corollary~\ref{Corollary:EstBellValue}]
Given $n\in\mathbb{N}^+,\varepsilon_{xy},\varepsilon\in(0,1),\delta_{xy}\in[0,\infty)$, and $\delta_{xy}'\in\mathbb{R}$ that satisfy
\begin{equation}
  \begin{split}
     \varepsilon&=\sum_{x,y\in\{0,1\}}\varepsilon_{xy}, \\
     \varepsilon_{xy}&=\exp{-\frac{2n(\delta_{xy}^2-\delta_{xy}'^2)}{\left(1+\frac{4\delta_{xy}'}{3}\right)^2}},
  \end{split}
\end{equation}
with probability no smaller than $1-\varepsilon$, the average expected Bell value over $n$ rounds can be lowered bounded by,
\begin{equation}
  \dfrac{1}{n}\sum_{i=1}^nS^{\lambda_i} \geq \sum_{x,y\in\{0,1\}}\left\{\dfrac{2\left[ \left(1+2\delta_{xy}'\right)\left(m_{xy}-q_{xy}\right)-n\left(\delta_{xy}+\delta_{xy}'\right)\right]}{np_X(x)p_Y(y)}-1\right\},
\end{equation}
where $S^{\lambda_i}$ is the expected Bell value of the $i$th round with system variable value $\Lambda_i=\lambda_i$, as given in Eq.~\eqref{Eq:ExpectedBellValue}.
\end{corollary}

\begin{theorem}[Refined version of Theorem~\ref{thm:PrivacyAmpCost}]
For the DIQKD protocol given in Box~\ref{box:ProtocolDetail}, suppose the number of experimental rounds is $n$. Consider failure probabilities in phase-error sample entropy estimation $\varepsilon_{pe}\in(0,1)$, in parameter estimation of each Bell test setting $\forall x,y\in\{0,1\}, \varepsilon_{xy}\in(0,1)$, and in privacy amplification $\varepsilon_{pc}\in(0,1)$, such that they satisfy $\varepsilon = \varepsilon_{pe} + \sum_{x,y\in\{0,1\}}\varepsilon_{xy} + \varepsilon_{pc}\in(0,1)$. Choose parameters $\xi\in\left(0,\frac{1}{2}\right)$, $\delta_{p}\in [0,\infty),\delta_{p}'\in\mathbb{R}$, and $\forall x,y\in\{0,1\}, \delta_{xy}\in [0,\infty),\delta_{xy}'\in\mathbb{R}$ that satisfy
\begin{equation}
\begin{split}
   \varepsilon_{pe} &= \exp[-\frac{2n(\delta_{p}^2-\delta_{p}'^2)}{\left(1+\frac{4\delta_{p}'}{3}\right)^2}], \\
   \varepsilon_{xy} &= \exp[-\frac{2n(\delta_{xy}^2-\delta_{xy}'^2)}{\left(1+\frac{4\delta_{xy}'}{3}\right)^2}]. \\
\end{split}
\end{equation}
With a failure probability no larger than $\varepsilon$, the cost for privacy amplification can be upper bounded by
\begin{equation}
  I_{pa}\leq \frac{1}{1+2\delta_{p}'}n\left\{p_X(2)p_Y(0)h\left[e_{p}^{\xi} (S_{est})\right] + c_{\xi}(\delta_{p}+\delta_{p}') \right\}-\log\varepsilon_{pc},
\end{equation}
where $S_{est}$ is the estimated expected Bell value,
\begin{equation}\label{Eq:EstExpBell}
\begin{split}
   S_{est}&= \sum_{x,y\in\{0,1\}}\left\{\dfrac{2\left[ (1+2\delta_{xy}')(m_{xy}-q_{xy})-n(\delta_{xy}+\delta_{xy}')\right]}{np_X(x)p_Y(y)}-1\right\}, \\
\end{split}
\end{equation}
the $\xi$-regularised phase-error probability function $e_p^{\xi}(S)$ is given in Eq.~\eqref{Eq:SmoothPhaseErrorProb}, and $c_{\xi}$ is a constant originated from the regularisation parameter $\xi$,
\begin{equation}
  c_{\xi}= \log(\frac{1+2\xi}{\xi}).
\end{equation}
\label{Thm:SuppMainResult}
\end{theorem}

Here, $\delta_{p},\delta_{p}'$ and $c_{\xi}$ come from the estimation of the phase-error sample entropy and $\delta_{xy}$ and $\delta_{xy}'$ come from the estimation of the expected Bell value. The deviation terms $\delta_p,\delta_p',\delta_{xy}$ and $\delta_{xy}'$ scale in the order of $O(1/\sqrt{n})$. Compared with Azuma's inequality, $\delta_{p}'$ and $\delta_{xy}'$ are new modification terms introduced in Kato's inequality. If  $\delta_{p}'=\delta_{xy}'=0$, Kato's inequality becomes Azuma's inequality. In correspondence to the informal version of the theorem in the main text, for the symmetric case where the input settings are chosen with equal probabilities for Bell tests and the values $q_{xy}$ are expected to be almost identical, one can choose $\forall x,y\in\{0,1\}, \varepsilon_{xy}=\varepsilon_{S}/4$ and fix identical values for $\delta_{xy}\equiv\delta$ and $\delta_{xy}'\equiv\delta'$. In this case, Eq.~\eqref{Eq:EstExpBell} can be simplified to,
\begin{equation}
  S_{est}\approx(1+2\delta')\bar{S}+8\delta'-\sum_{x,y\in\{0,1\}}\frac{2(\delta+\delta')}{p_X(x)p_Y(y)}.
\end{equation}
For simplicity, we take the approximation $m_{xy}\approx np_X(x)p_Y(y)$.
In the asymptotic limit of infinite data size, the average cost per round for privacy amplification converges to
\begin{equation}
  \lim_{n\rightarrow\infty}\frac{1}{n}I_{pa}=p_X(2)p_Y(0)h\left[e_{p} (\bar{S})\right],
\end{equation}
with the regularisation parameter $\xi$ converging to zero, which is tight in the asymptotic limit. In the numerical simulation, we take the information reconciliation efficiency to be $1$, reaching the Shannon limit.

In the refined result, the terms $\delta_p,\delta_p',\delta_{xy}$, and $\delta_{xy}'$ can be optimised according to Eqs.~\eqref{Eq:MaxOpt} and \eqref{Eq:MinOpt}. Here, we provide an optimisation protocol. Note that this protocol only aims at tighter finite-size performance. Non-optimal results do not cause any security problems as long as they satisfy the requirements in Theorem~\ref{Thm:SuppMainResult}.

\begin{mybox}[label={box:KatoOpt}]{{Optimisation protocol for the use of Kato's inequality}}
\textbf{Arguments: }\\
$n$: the number of rounds of quantum measurements \\
$\{p_X(x)\}_x,\,\{p_Y(y)\}_y$: probability distributions for the basis choices \\
$\varepsilon_{pe}\in(0,1)$: failure probability in phase-error sample entropy estimation \\
$\varepsilon_{xy}\in(0,1),\forall x,y\in\{0,1\}$: failure probability in Bell value estimation for the setting $(x_i,y_i)=(x,y)$ \\
$\xi\in(0,\frac{1}{2})$: regularisation parameter

\tcblower
\begin{enumerate}[leftmargin=*]
\item\emph{Prediction for experimental data: }Before the DIQKD experiment, for each Bell test setting, predict the number of rounds, $\tilde{q}_{xy}$, where $(-1)^{x\cdot y}a_ib_i=-1$.
\item \emph{Optimisation for Bell-value estimation: }In Eq.~\eqref{Eq:MaxOpt}, $\forall x,y\in\{0,1\}$, take the replacement
    \begin{equation}\nonumber
    \begin{split}
       \tilde{X}\leftarrow\frac{\tilde{q}_{xy}}{np_X(x)p_Y(y)}, \varepsilon\leftarrow\varepsilon_{xy}, \\
    \end{split}
  \end{equation}
  solve the optimisation and then take
  \begin{equation}\nonumber
    \begin{split}
    \delta_{xy}\leftarrow\delta^*,
    \delta_{xy}'\leftarrow\delta'^*.
    \end{split}
  \end{equation}

\item \emph{Prediction of the expected Bell value: }In Eq.~\eqref{Eq:EstExpBell}, take the replacement
    \begin{equation}\nonumber
      m_{xy}\leftarrow np_X(x)p_Y(y),q_{xy}\leftarrow\tilde{q}_{xy},
    \end{equation}
    calculate the formula and then take
    \begin{equation}\nonumber
      \tilde{S}_{est}\leftarrow S_{est}.
    \end{equation}
\item \emph{Optimisation for regularised phase-error sample entropy estimation: }In Eq.~\eqref{Eq:MinOpt}, take the replacement
    \begin{equation}\nonumber
    \begin{split}
       \tilde{\chi}\leftarrow c_{\xi}^{-1}h[e_p^{\xi}(\tilde{S}_{est})], \varepsilon\leftarrow\varepsilon_{pe}, \\
    \end{split}
    \end{equation}
    solve the optimisation and then take
    \begin{equation}\nonumber
      \delta_{p}\leftarrow\delta^*, \delta_{p}'\leftarrow\delta'^*.
    \end{equation}
\end{enumerate}
\end{mybox}

\section{Advantage Key Distillation}\label{Supp:AdvKeyDistill}
So far, we have restricted our discussions to the standard protocol in Box~\ref{box:ProtocolDetail}, which has the following features.
\begin{enumerate}
\item
The devices always generate binary measurement outputs. That is, the measurement results of each observable range in $\{+1,-1\}$. If not, output values other than in this set are merged to $+1$ or $-1$.
\item
The measurement results under the key generation setting are taken directly as the raw keys.
\item
The users carry out one-way classical communication in the post-processing procedure.
\end{enumerate}

For common device-dependent QKD protocols, there exist methods to obtain better key generation performance than a similar `standard protocol'. For instance, one could utilise the noisy processing~\cite{kraus2005lower}, the detection tag information~\cite{ma2012improved}, and two-way classical communication such as the B-step~\cite{gottesman2003proof} to improve the tolerance on transmittance and noise. In Sec.~\ref{Sec:NumericalDiscussion} we shall explain the specific meanings of `loss' and `noise' in experiments.
On the contrary, little is known about the feasibility of the advantage key distillation methods in DIQKD so far. An exceptional work in~\cite{tan2020advantage} shows that the repetition code can be adopted in DIQKD. Nevertheless, the analysis is restricted to the i.i.d.~scenario.


Fortunately, under the complementarity-based security analysis, the advantage key distillation methods mentioned above can be adopted in DIQKD under the most general attacks --- coherent attacks. In this section, we present the procedures for carrying out methods and analyse their validity.


\subsection{Detection tag}
In a real experiment, detection outcomes are not necessarily binary. Typically, a measurement device with two single-photon detectors might output inconclusive results of no-clicks and double-clicks denoted by $\phi$, in addition to conclusive results of single-clicks denoted by $+1$ and $-1$. In DIQKD, in order to apply Jordan's Lemma, Alice and Bob would coarse-grain all other detection results to a bit, 0 or 1. For instance, in our simulation as shown in Sec.~\ref{Supp:EAT}, we project $\phi$ to $+1$. Then, Alice and Bob could attach a \textit{detection tag} to each raw bit to indicate whether it comes from a genuine single-click or coarse-grain. Note that this detection tag could contain any extra information from the measurement device beyond the basis and key bits.

In standard DIQKD, Alice and Bob simply ignore the detection tag in postprocessing. Obviously, this is a waste of information. By taking advantage of the detection tag, Alice and Bob can perform information reconciliation more efficiently \cite{ma2012improved}, as shown in Box~\ref{box:Vacancy}. Bob hashes his raw bits and sends the parity information to Alice, who locates the bit errors and correct them. In this procedure, Alice can group her raw key bits by the detection tag, $t_a$, and correct errors for each group separately. Since the bit error rate for each group with different tags is different, Alice can achieve a higher information reconciliation rate,
\begin{equation}
\sum_{t_a} H(\kappa^{B}|\kappa^{A,t_a}) \geq H(\kappa^{B}|\kappa^{A}),
\end{equation}
where $\kappa^{A}$ is the concatenation of $\kappa^{A,t_a}$ with different $t_a$ and $H(\cdot)$ is the Shannon entropy.


\begin{mybox}[label={box:Vacancy}]{{Usage of Loss Information}}
\begin{enumerate}[leftmargin=*]
\item\emph{Measurement: }For every round $i\in[n]$, Alice and Bob
\begin{enumerate}[leftmargin=*]
\item randomly set measurement bases $x_i\in \{0,1,2\},y_i\in \{0,1\}$;
\item record the outputs of the devices $a_i'\in = \{+1,-1,\phi\},b_i' \in \{+1,-1,\phi\}$;
\item assign the events $\phi$ to the value $+1$ and obtain $\{a_i\}_i,\{b_i\}_i\in\{+1,-1\}^n$;
\item Alice records the loss information in an $n$-bit-string tag $t_a$, where the $i$th bit is $1$ if $a_i'=\phi$.
\end{enumerate}
\item \emph{Parameter Acquisition: }Alice and Bob obtain the following parameters in the experiment
\begin{enumerate}[leftmargin=*]
\item the number of rounds $m_{xy}$ with the inputs $(x_i,y_i) = (x,y)\in\{0,1,2\}\times\{0,1\}$;
\item for $(x_i,y_i)=(2,0)$, the difference between $\{a_i'\}_i,\{b_i\}_i$;
\item for $(x_i,y_i) = (x,y)\in\{0,1\}\times\{0,1\}$, the number of positions $q_{xy}$ where $(-1)^{x\cdot y}a_i\cdot b_i=-1$.
\end{enumerate}
\item \emph{Raw Key Acquisition: }Alice and Bob obtain the strings $\kappa^{A},\kappa^B$ from $\{a_i\},\{b_i\}$ with $(x_i,y_i)=(2,0)$.
\item \emph{Privacy Estimation: }Alice and Bob estimate the key privacy with statistics acquired in Step 2, (c).
\item \emph{Key Distillation: }Alice and Bob
\begin{enumerate}[leftmargin=*]
\item reconcile the bit strings to $\kappa^B$ through an encrypted classical channel with the tag $t_a$;
\item  apply a privacy amplification to $\kappa^B$ and derive the final key.
\end{enumerate}
\end{enumerate}
\end{mybox}


\subsection{Adding-noise pre-processing}
Instead of directly taking the measurement results in key generation rounds as for information reconciliation and privacy amplification, Alice and Bob can locally process their raw key data before classical communication. We employ the \textit{adding-noise} pre-processing technique \cite{kraus2005lower}, as shown in Box~\ref{box:NoisyProcess}. Bob randomly flips some of his raw key bits. Afterwards, the users perform the same procedures of information reconciliation and privacy amplification as in the standard protocol, in which they reconcile the keys to Bob's flipped key-bit string. The random flips can be interpreted as noise deliberately added by Bob, which unavoidably increases the quantum bit error rate and hence the information reconciliation cost. On the other hand, as the positions of flipped key bits are not controlled or known by Eve, this added noise would effectively reduce the amount of information leakage and hence the cost for privacy amplification. Overall, the users might acquire a net increase in key generation.

\begin{mybox}[label={box:NoisyProcess}]{{Adding-noise pre-processing}}
\begin{enumerate}[leftmargin=*]
\item\emph{Measurement: }For every round $i\in[n]$, Alice and Bob
    \begin{enumerate}[leftmargin=*]
    \item randomly set measurement bases $x_i\in \{0,1,2\},y_i\in \{0,1\}$ with probability distributions $p_X,p_Y$;
    \item record the outputs $a_i \in \{\pm1\},b_i \in \{\pm1\}$.
    \end{enumerate}
\item \emph{Parameter Acquisition: }Alice and Bob obtain the following parameters in the experiment
    \begin{enumerate}[leftmargin=*]
    \item the number of rounds $m_{xy}$ with inputs $(x_i,y_i) = (x,y)\in\{0,1,2\}\times\{0,1\}$;
    \item for $(x_i,y_i)=(2,0)$, the number of positions $q$ where $a_i\neq b_i$;
    \item for $(x_i,y_i) = (x,y)\in\{0,1\}\times\{0,1\}$, the number of positions $q_{xy}$ where $(-1)^{x\cdot y}a_i\cdot b_i=-1$.
    \end{enumerate}
\item
\emph{Raw Key Acquisition:} Alice and Bob obtain the bit strings $\kappa^A,\kappa^B$ from $a_i,b_i$ with $(x_i,y_i)=(2,0)$.
\item
\emph{Adding-noise:} Bob randomly flips the bits in $\kappa^B$ according to a conditional probability distribution $\Pr(u|b)$ and obtain $\kappa^U$.
\item \emph{Privacy Estimation: }Alice and Bob estimate the key privacy with statistics acquired in Step 2(c).
\item \emph{Key Distillation: }Alice and Bob
\begin{enumerate}[leftmargin=*]
\item reconcile the bit strings to $\kappa^U$ through an encrypted classical channel;
\item  apply a privacy amplification to $\kappa^U$ and derive the final key.
\end{enumerate}
\end{enumerate}
\end{mybox}

In the adding-noise step of the modified protocol, Bob flips his original raw key bits at random. This step can be equivalently described by a quantum operation as shown in Figure \ref{Fig:AddingNoise}. In the key generation rounds, Bob holds an ancillary system $\mathcal{H}_{B_i'}$ initialised in the state $\ket{\varphi}\in\mathcal{H}_{B_i'}$ and applies a quantum control-not ($CNOT$) operation to the subsystem $\mathcal{H}_{B}^{\lambda_i}$. Denote the $CNOT$ operation as $\hat{C}_{B_i'B_i}$, with the first subscript for the control and the second subscript for the target. The ancillary state takes the form of $\ket{\varphi}=\sqrt{1-q}\ket{0}+\sqrt{q}\ket{1}$, where $q=\Pr(u_i=-b_i|b_i)$. Here, Bob takes the ancillary states to be i.i.d.~over all the key generation rounds, corresponding to an i.i.d.~random flip in a real experiment.


\begin{figure}[!hbtp]
	\centering\resizebox{4.5cm}{!}{
		\Qcircuit @C=1.5em @R=2.25em {
			\lstick{} & \ustick{\mathcal{H}_{A_i}} \qw &    \qw &       \multigate{1}{I.R.} \qw & \measureD{\hat{A}_{\theta_i}^{\lambda_i}} & \cw \\
			\lstick{} & \ustick{\mathcal{H}_{B_i}} \qw &    \targ &     \ghost{I.R.} \qw & \measureD{\hat{B}_{\bot}^{\lambda_i}} & \cw \\
			\lstick{\ket{\varphi}} & \ustick{\mathcal{H}_{B_{i'}}} \qw & \ctrl{-1} & \qw     &  \qw & \qw
			\inputgroupv{1}{2}{0.5em}{1.4em}{\rho_{\lambda_i}^{AB}} \\
		}			
	}
\caption{An effective quantum circuit for the pre-processing of adding noise. We depict the $i$th round subsystem for illustration. The noise can be effectively described by an ancillary quantum system, $\varphi$, which affects Bob's quantum system via a $CNOT$ operation. The quantum gate $I.R.$ refers to the operation of one-way information reconciliation from Bob to Alice. In the circuit, we consider the virtual phase-error measurement where Alice and Bob measure the observables of $\hat{A}_{\theta_i}^{\lambda_i},\hat{B}_{\bot}^{\lambda_i}$ on their own systems, respectively. It is not hard to show that the adding-noise operation would not change the results of phase-error measurement, since they commute with each other and with $I.R.$}
	\label{Fig:AddingNoise}
\end{figure}
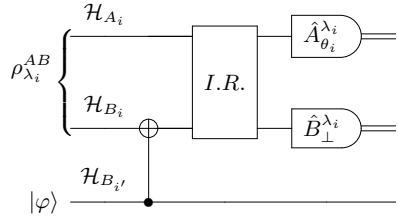

The state of Bob's $i$th ancillary system after $CNOT$ operation is given by,
\begin{equation}\label{Eq:NoiseAncilla}
	\begin{split}
		\sigma^{B_i'}=(1-e_{p}^{\lambda_i})\ketbra{\varphi}+e_{p}^{\lambda_i}\hat{Z}\ketbra{\varphi}\hat{Z}.
	\end{split}
\end{equation}
where $e_{p}^{\lambda_i}$ is the original phase-error probability in the $i$th round without the adding-noise pre-processing as a function of the Bell value $S$ in Eq.~\eqref{Eq:SuppPhaseError}. The eigenvalues of this state are
\begin{equation}
	\begin{split}
		\omega^\pm(e_p^{\lambda_{i}})=\frac{1\pm\sqrt{1-4\left\{q-q^2-\left[1-2e_{p}^{\lambda_i}\sqrt{q(1-q)}\right]^2\right\}}}{2}.
	\end{split}
\end{equation}
We employ the adding-noise analysis \cite{renes2007noisy} and extend it to the complementarity security proof. Here, we utilise the phase-error cardinality argument and apply a martingale-based method similar to that in Sec.~\ref{Supp:FiniteSize}. In the limit of infinite data size, the cost  in privacy amplification is given by, 
\begin{equation}\label{Supp:NoisySave}
\begin{split}
\sum_{i=1}^m h(e_{p}^{\lambda_i}) - H\left(\sigma^{B'}\right) &=\sum_i \left[h(e_{p}^{\lambda_i})- H\left(\sigma^{B_i'}\right)\right] \\
&=\sum_i \left\{h(e_{p}^{\lambda_i})-h\left[\omega^{+}(e_p^{\lambda_{i}})\right] \right\} \\
&\leq m \left(h\left[e_p(\bar{S})\right] - h\left\{\omega^+[e_p(\bar{S})]\right\}\right),
\end{split}
\end{equation}
where the function $H(\cdot)$ represents the von Neumann entropy function, and the summations in the first two lines are taken over the key generation rounds. The last inequality originates from convexity and the fact that in the asymptotic limit of infinite data size, the average expected Bell value in the key generation rounds converges to the observed average Bell value in the Bell test rounds in probability.

In practice, the noise level added to the sifted key, $q$, can be adjusted to maximise the final key generation rate. We present some numerical discussions in Sec.~\ref{Sec:AdvNum}. For a rigorous finite-size analysis, we leave the details to future works.

\subsection{B step}
In all previous discussions, we have restricted the data post-processing to one-way classical communication from Bob to Alice. In general, Alice and Bob can apply two-way classical communication. In common device-dependent QKD protocols, the B-step~\cite{gottesman2003proof}, a two-way information reconciliation method, has been found useful in increasing the tolerance to the quantum bit errors. Here, we discuss the application of this method to DIQKD, as shown in Box~\ref{box:B-step}.

\begin{mybox}[label={box:B-step}]{{B Step}}
\begin{enumerate}[leftmargin=*]
\item\emph{Measurement: }For every round $i\in[n]$, Alice and Bob
    \begin{enumerate}[leftmargin=*]
    \item randomly input $x_i\in\{0,1\},y_i\in\{0,1\}$ to the devices with the probabilities $p_X,p_X$;
    \item record the outputs of the devices $a_i \in \{\pm1\},b_i \in \{\pm1\}$.
    \end{enumerate}
\item \emph{Parameter Acquisition: }Alice and Bob obtain the following parameters in the experiment
    \begin{enumerate}[leftmargin=*]
    \item the number of rounds $m_{xy}$ with the inputs $(x_i,y_i)=(x,y)\in\{0,1\}\times\{0,1\}$;
    \item for $(x_i,y_i)=(x,y)$, the number of positions $q_{xy}$ where $(-1)^{x\cdot y}a_i b_i=-1$.
    \end{enumerate}

\item \emph{Raw Key Acquisition: }Alice and Bob obtain the bit strings $\kappa^A,\kappa^B$ from $a_i,b_i$ with $(x_i,y_i)=(0,0)$.
\item \emph{B Step: }Denote $\kappa^{A(0)}\equiv\kappa^A,\kappa^{B(0)}\equiv\kappa^B$. For $t\in[N_B]$, Alice and Bob repeat the following procedures:
\begin{enumerate}
  \item each randomly permute the bits of $\kappa^{A(t-1)},\kappa^{B(t-1)}$ and obtain $\bar{\kappa}^{A(t)},\bar{\kappa}^{B(t)}$.
  \item apply an XOR operation on the $j$th, $(j+1)$th bits of the permuted bit stings $\bar{\kappa}^{A(t)},\bar{\kappa}^{B(t)}$ ($j\in\{2\mathbb{N}+1\}$), obtain $\mu^{A(t)},\mu^{B(t)}$, where $\mu_{i}^{A(t)}=\bar{\kappa}_{j}^{A(t)}\oplus\bar{\kappa}_{j+1}^{A(t)}, i=(j+1)/2$, and similarly for $\mu^{B(t)}$;
  \item compare the bits of $\mu^{A(t)},\mu^{B(t)}$ via two-way classical communication and discard the positions where $\mu_{i}^{A(t)}\neq\mu_{i}^{B(t)}$, obtain $\bar{\mu}^{A(t)},\bar{\mu}^{B(t)}$, which will decide the bits to be kept;
  \item for $\bar{\mu}_{i}^{A(t)} = \bar{\kappa}_j^{A(t-1)}\oplus\bar{\kappa}_{j+1}^{A(t-1)},\bar{\mu}_{i}^{B(t)} = \bar{\kappa}_j^{B(t-1)}\oplus\bar{\kappa}_{j+1}^{B(t-1)}$, keep $\bar{\kappa}_j^{A(t-1)},\bar{\kappa}_j^{B(t-1)}$ and discard $\bar{\kappa}_{j+1}^{A(t-1)},\bar{\kappa}_{j+1}^{B(t-1)}$, obtain $\kappa^{A(t)},\kappa^{B(t)}$.
\end{enumerate}
One run of the above procedures is called \emph{a B-step}.

\item \emph{Privacy Estimation: }Alice and Bob estimate the key privacy with statistics acquired in Steps 2 - 4.
\item \emph{Key Distillation: }Alice and Bob
\begin{enumerate}[leftmargin=*]
\item reconcile the bit strings to $\kappa^{B(N_B)}$ through an encrypted classical channel;
\item  apply a privacy amplification to $\kappa^{B(N_B)}$ and derive the final key.
\end{enumerate}
\end{enumerate}
\end{mybox}


Different from all the DIQKD protocols discussed in the previous sections, in the protocol shown in Box~\ref{box:B-step}, both Alice and Bob have only two possible measurement observables in each round of the quantum measurements. This modification originates from the difficulty in an effective dimension reduction.
In the B-step method, since Alice and Bob shall use the raw key information in collaborating the key processing, we must have an effective quantum characterisation of Alice's key generation measurement observable. If we use the same design for the quantum measurements as in the standard protocol, the ignorance of $\hat{A}_2$ shall lay an obstacle for further analysis. For this reason, we remove the measurement $\hat{A}_2$ on Alice's side, and the legitimate users shall employ the standard CHSH Bell test setup. They take one pair of the measurement settings, say $(x_i,y_i)=(0,0)$, also for key generation. In this way, we can treat all the quantum systems as effective qubits with the analysis in Sec.~\ref{Supp:EffDimReduction}. It is unavoidable to have a systematical quantum bit error rate, as the users cannot arrange the key generation outcomes to be perfectly correlated while obtaining a large violation of the Bell inequality at the same time. Nevertheless, by applying the B-step, the quantum bit error rate can be reduced, and a net increase in the key generation may occur.

We remark that a full security analysis of the B-step-based DIQKD protocol is not straightforward. First, different from the B-step in common device-dependent QKD protocols, the quantum bit-error probability and the phase-error probability in two adjacent rounds are not equal if the underlying system variables $\Lambda_i$ are different. Second, as discussed in Sec.~\ref{Supp:FiniteSize}, the average phase error probability is not a good measure for privacy estimation. Instead, the phase-error sample entropy should be used. Still, we could generalise the analysis in~\cite{gottesman2003proof} to obtain the refreshed quantum bit-error and phase-error probabilities. We present the following claim and leave the detailed analysis for future works.

\begin{claim}
  In the asymptotic limit of infinite data size, suppose before the application of B-steps, the quantum bit error rate is $e_b$, and the average cost per round for privacy amplification is $h[e_p(\bar{S})]$, where the function $e_p(\bar{S})$ is given in Eq.~\eqref{Eq:SuppPhaseError}. After one B-step, the key rate in the key generation rounds is
  \begin{equation}
    r'=\frac{1}{2}\left[e_b^2+(1-e_b)^2\right]\left\{1-h(e_b')-h[e_p'(\bar{S})]\right\},
  \end{equation}
  where $e_b'$ is the effective quantum bit rate after the B-step, $e_p'(S)$ is the effective quantum phase error probability with respect to the Bell value $S$, and the term $\frac{1}{2}\left[e_b^2+(1-e_b)^2\right]$ represents the fraction of data that is kept after post-selection. The values are upper bounded by
  \begin{equation}
  \begin{split}
  e_b'&\leq\dfrac{e_b^2}{1-2e_b+2e_b^2},\\
  e_p'&\leq\dfrac{2e_p(S)[1-e_b-e_p(S)]}{1-2e_b+2e_b^2}.
  \end{split}
  \end{equation}
\end{claim}

\section{Numerical Discussions and Simulations}\label{Sec:NumericalDiscussion}
In this section, we present some numerical discussions of our analysis and evaluate the key generation performance with specific experimental models. We shall also explain the numerical results presented in the main text in detail. In the beginning, we remark that all the numerical results of our complementarity-based method are obtained by the use of Kato's inequality, namely, the results in Sec.~\ref{Sec:KatoResult}.

\subsection{Numerical simulation model}\label{Supp:DetailNum}
In all the numerical simulations related with the standard model, including also the protocols taking advantage of the loss information and applying a noisy pre-processing step, we shall consider a `noisy and lossy' experimental setup.
The experimental set-up is designed according to the ideal setting. In each round, Alice and Bob share a Bell state
\begin{equation}
\ket{\Phi^+}=\frac{\ket{00}+\ket{11}}{\sqrt{2}},
\end{equation}
where $\{\ket{0},\ket{1}\}$ are the eigenvectors of the Pauli operator $\sigma_z$. The observables of their measurements are
\begin{equation}\label{Eq:IdealObservable}
\begin{split}
\hat{A}_0&=\frac{\hat{\sigma}_z+\hat{\sigma}_x}{\sqrt{2}}, \hat{A}_1=\frac{\hat{\sigma}_z-\hat{\sigma}_x}{\sqrt{2}}, \hat{A}_2=\hat{\sigma}_z, \\ \hat{B}_0&=\hat{\sigma}_z, \hat{B}_1=\hat{\sigma}_x.
\end{split}
\end{equation}
The transmission of the Bell state suffers from depolarising channel noise. When the Bell state reaches the detectors, it evolves into
\begin{equation}
\hat{\mathcal{E}}_d[\rho] = (1-e_d)\rho + e_d \frac{\hat{I}}{4},
\end{equation}
where $e_d\in[0,1]$ is the depolarising factor. The state preparation fidelity is given by
\begin{equation}\label{Eq:Fidelity}
F\left(\Phi^+,\hat{\mathcal{E}}_d[\Phi^+]\right)=\bra{\Phi^+}\hat{\mathcal{E}}_d[\Phi^+]\ket{\Phi^+}=1-\frac{3e_d}{4}.
\end{equation}
There is also loss in both state transmission and detection. We model this effect by considering detectors with finite efficiencies. For a detector with efficiency $\eta$, with probability $\eta$, it works ideally, which shall measure the observables in Eq.~\eqref{Eq:IdealObservable} on the bases settings. With probability $(1-\eta)$ it fails to detect the quantum signal. For simplicity, we take a symmetric detection model, where the detectors of Alice and Bob have the same efficiency, $\eta_A=\eta_B=\eta$. We call $\eta$ the total transmittance in the experiment. If the postprocessing of using the loss information is not applied, the undetected events shall all be assigned with the value $+1$. We remark that though the experimental behaviour is actually i.i.d., we do not assume this in the numerical simulation of the key generation. In information reconciliation, except for the numerical simulation in Sec.~\ref{Sec:IonSim}, we take the efficiency parameter $f_{ec}$ in Eq.~\eqref{Supp:InfoRec} as $f_{ec}=1$ for simplicity, reaching the Shannon limit.
%

\subsection{Numerical discussion on regularisation}\label{Sec:Regular}
In this section, we present some numerical results on the regularisation parameter $\xi$. From the security analysis in Sec.~\ref{Supp:FiniteSize}, it can be seen that the regularisation parameter $\xi$ in the regularised sample entropy can be taken as a free parameter in the range of $(0,\frac{1}{2})$ without compromising the security. Nevertheless, after the other parameters are fixed in the protocol, $\xi$ can be optimised for a better key generation rate. We investigate the relations between the optimal value of $\xi$ and the total transmittance and the number of total experiment rounds, respectively, in Figure \ref{Fig:OptimalXi}. In the simulation, we consider the state preparation fidelity to be unity. The numerical results for the state preparation fidelity are similar to that for the total transmittance, hence we do not discuss it in detail. In Figure \ref{Fig:OptimalXi}(a) we present the optimal values of $\xi$ under different total transmittance. We also consider three different data sizes. The regularisation parameter $\xi$ is optimised to maximise the right-hand-side term in Eq.~\eqref{Eq:PrivacyAmplificationCost}. From the simulation results, we find that with a larger data size and a higher transmittance, the optimal value for the regularisation becomes smaller. In Figure \ref{Fig:OptimalXi}(b), we present the optimal value of $\xi$ with respect to different data sizes. We also consider three different values of total transmittance. As shown by the tendency of the curves, the optimal value of $\xi$ approaches to zero in the limit of infinite data size.

\begin{figure*}[hbt!]
\hfil
\vbox{
\hsize 0.5\textwidth
\hbox{
{\begin{overpic}[width=\hsize]{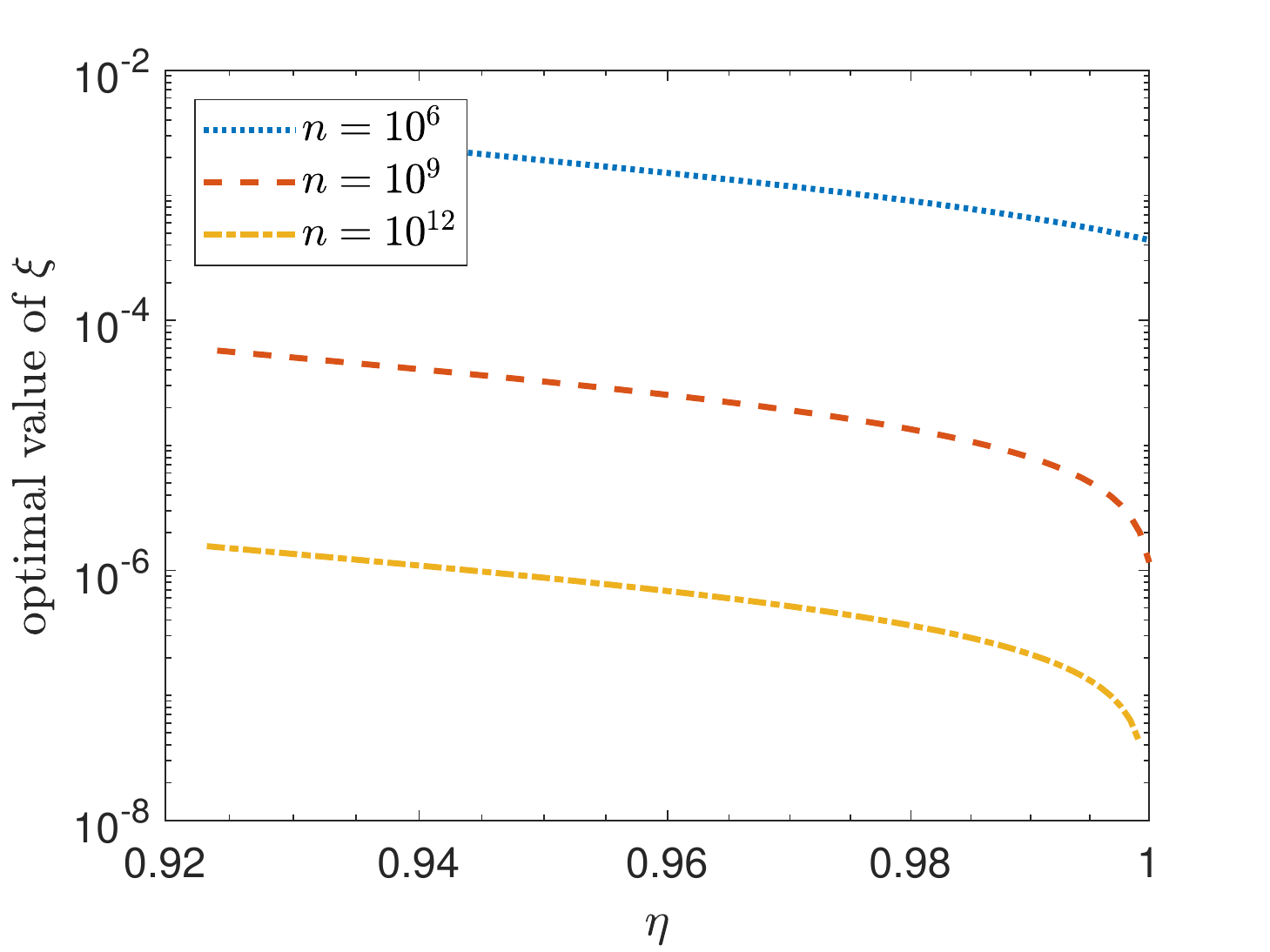}
\put(49,-4)
{\fontsize{10}{10}\selectfont (a)}
\end{overpic}}
\hbox{
{\begin{overpic}[width=\hsize]{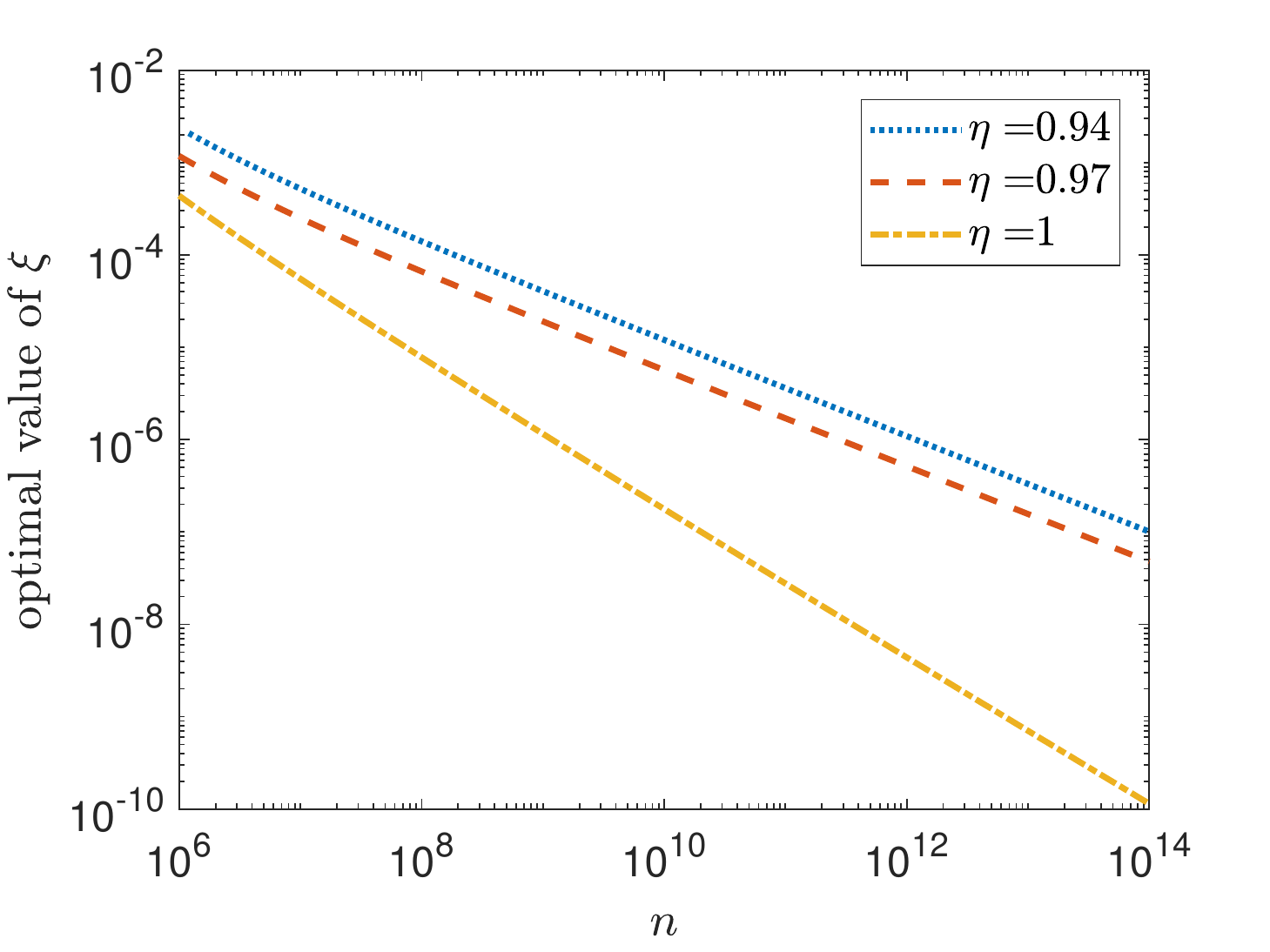}
\put(49,-4)
{\fontsize{10}{10}\selectfont (b)}
\end{overpic}}
}
}
}
\hfil
\caption{Optimisation of the regularisation parameter $\xi$. In the simulation, the failure probabilities in privacy estimation are set as $\varepsilon = \varepsilon_{S} + \varepsilon_{pe}+\varepsilon_{pc}$, with $\varepsilon_{S} = \sum_{x,y\in\{0,1\}} \varepsilon_{xy} = 1\times10^{-7},\varepsilon_{pe} = 1\times10^{-7},\varepsilon_{pc} = 1\times10^{-7}$. We consider a uniform distribution for the choices of measurement settings on both sides. (a) The optimal value of $\xi$ with respect to different total transmittance. Three different data sizes of $n=1\times10^6,1\times10^9,1\times10^{12}$ are considered. (b) The optimal value of $\xi$ with respect to different data sizes. Three different total transmittance of $\eta=0.94,0.97,1$ are considered.}
\label{Fig:OptimalXi}
\end{figure*}


\subsection{Comparison with entropy accumulation theorem}\label{Supp:EAT}
In this section, we give a brief discussion on the comparison between our security analysis and the results in~\cite{arnon2018practical,arnon2019simple,murta2019towards}, where the DIQKD security is established via EAT, an entropy-based method. For brevity, we call our method the complementarity-based method.

\subsubsection{Theoretical basis for comparison}
The statements such as the description of security parameters are different in the complementarity-based and EAT-based approaches. To unify the descriptions, we first explain the languages in two approaches on stating security parameters.

Generally, given a QKD protocol, we can define a success criterion for key generation. When the criterion is met in the experiment, the users shall obtain keys. In DIQKD, before the experiment begins, suppose the users estimate that the underlying quantum system should yield an expected Bell value $S_{exp}$. We can set the protocol success criterion as that after the quantum measurements are finished, the observed average Bell value is above a threshold,
\begin{equation}\label{Supp:SuccCriteria}
\bar{S}\geq S_{exp}-\delta_{s},
\end{equation}
where $\delta_{s}>0$ is a pre-fixed parameter accounting for statistical fluctuations, related to the completeness error. If the success criterion is met, the legitimate users of QKD can proceed to the post-processing steps and obtain the final keys. We say the users has a successful implementation of the protocol. Otherwise, the protocol aborts. In some security analyses, the potential failure of information reconciliation is also considered. Nevertheless, whether information reconciliation succeeds can be checked via error verification~\cite{fung2010practical}. We do not consider this issue in the analysis and assume that the information reconciliation is successful by default.

In a rigorous manner, the completeness error $\varepsilon_{QKD}^c$ means that there exists an honest implementation, of which the probability that the protocol aborts is upper bounded, $p_{abort}\leq\varepsilon_{QKD}^c$. Completeness requires that a protocol does not abort all the time. There should exist an implementation such that the users can obtain final keys with a high probability. If a protocol always aborts, it does not leak information, but it is useless. For QKD, the completeness is usually naturally met. For well-designed experiments that are capable to generate secure keys, the abort probability is typically negligible.

If the legitimate users consider that they have a successful implementation, it is possible that their keys are actually not secure. This issue is covered by the security definition in Def.~\ref{Def:SecurityDef}. In the entropic approaches like the EAT method, it is preferred to use the soundness error for security statements. In the complementarity-based method, the security parameter is often given by the failure probability of the process, of which the operation meaning is given by Eq.~\eqref{Eq:FidelitySec}.

In~\cite{fung2010practical}, it is proven that the soundness error can be quantitatively linked with the total failure probability in a phase-error-correction-based security analysis. If the total failure probability of the phase-error-correction-based analysis is upper bounded by $\varepsilon_f$, then the protocol soundness can be upper bounded by
\begin{equation}\label{Supp:FailProbtoSound}
\varepsilon_{QKD}^s\leq(1-p_{abort})\sqrt{\varepsilon_f(2-\varepsilon_f)}.
\end{equation}
When $\varepsilon_{QKD}^c,\varepsilon_f\approx0$, the protocol is then approximately $\sqrt{2\varepsilon_f}$-sound. This relation shall be the basis for a relatively fair comparison between the complementarity-based analysis and the EAT-based analysis.

In the entropy-based security proofs, the final key security is quantified in terms of the (smooth) min-entropy~\cite{renner2008security}.

\begin{definition}[Smooth Min-Entropy]
For a classical-quantum state $\rho_{AE}=\sum_{a\in\mathcal{A}}p_A(a)\ketbra{a}\otimes\rho_E(a)$ acting on the joint Hilbert space $\mathcal{H}_{AE}$, given $\varepsilon\in(0,1)$, the $\varepsilon$-smooth min-entropy is defined as
\begin{equation}
H_{min}^\varepsilon(A|E)_{\rho_{AE}}=\max_{\tilde{\rho}_{AE}}\left[-\log\max_{\hat{M}_a}\sum_{a\in\mathcal{A}}\tilde{p}(a)\Tr(\hat{M}_a\tilde{\rho}_E(a))\right],
\end{equation}
where the maximisation is taken over all sub-normalised classical-quantum operators $\tilde{\rho}_{AE}=\sum_{a\in\mathcal{A}}\tilde{p}_A(a)\ketbra{a}\otimes\tilde{\rho}_E(a)$ that are $\varepsilon$-close to $\rho_{AE}$ in the purified distance, and $\{\hat{M}_a\}_{a\in\mathcal{A}}$ is a positive operator-valued measure (POVM). For a sub-normalised operator, $\sigma$, $\Tr(\sigma)\leq1$. The purified distance between two sub-normalised positive semi-definite operators $\rho,\sigma$ is defined as
\begin{equation}
\begin{split}
P(\rho,\sigma)=\sqrt{1-\left(\Tr|\sqrt{\rho}\sqrt{\sigma}|+\sqrt{[1-\Tr(\rho)][1-\Tr(\sigma)]}\right)^2}.
\end{split}
\end{equation}
\end{definition}

\subsubsection{EAT protocols in comparison}
In this section, we recap two EAT-based security analyses and express clearly what formulae we are using in evaluating their performance. In both protocols, the CHSH Bell value is evaluated through the winning probability of the CHSH game. If $a_i\cdot b_i=(-1)^{x_i\cdot y_i}$, we say the CHSH game wins in the round. The expected CHSH Bell value $S$ can be converted into the winning probability $\omega$ via a linear transformation,
\begin{equation}
\omega=\frac{1}{8}S+\frac{1}{2}.
\end{equation}
The protocol success criterion is set similar to Eq.~\eqref{Supp:SuccCriteria}, which requires that the observed winning frequency to be high enough. This shall be made rigorous in the statement of protocols.

The first result we consider comes from the original EAT-based DIQKD security analysis~\cite{arnon2018practical,arnon2019simple}. We consider a spot-checking protocol, as shown in Box~\ref{box:OriginEAT}. Alice and Bob shall use an additional trusted i.i.d.~random variable, $t_i$, to determine whether a round is used for key generation, termed as a `spot', or for the Bell test, termed as a `check'. This protocol is similar to the protocol we use, while an advantage is that the additional random variable avoids `useless' experimental rounds, i.e., the rounds with measurement settings $(x_i,y_i)=(2,1)$ in Box~\ref{box:ProtocolDetail}. With the improvement in the basis sifting factor, this protocol design may yield a higher key expansion rate, that is, the value of $k/n$, when the number of experimental rounds is large. The key generation result is given in Theorem~\ref{Thm:OriginEAT}.

\begin{mybox}[label={box:OriginEAT}]{{DIQKD Protocol in Ref.~\cite{arnon2019simple}}}
\textbf{Arguments: }\\
$n$: the number of rounds of quantum measurements \\
$\gamma\in(0,1]$: the ratio of Bell test rounds \\
$\omega_{exp}$: expected CHSH winning probability \\
$\delta_{est}\in(0,1)$: width of the statistical confidence interval for parameter estimation \\
$\chi$: indicator function, taking its value as $1$ for a true statement and $0$ for a false statement

\tcblower
\begin{enumerate}[leftmargin=*]
\item\emph{Measurement: }For every round $i\in[n]$, Alice and Bob
    \begin{enumerate}[leftmargin=*]
    \item randomly set $t_i\in\{0,1\}$ such that $\Pr(t_i=1)=\gamma$;
    \item if $t_i=0$, set measurement bases $(x_i,y_i)=(2,0)$; otherwise, randomly set $x_i,y_i\in\{0,1\}$;
    \item record the outputs of the devices $a_i,b_i\in\{\pm1\}$.
    \end{enumerate}
\item \emph{Error correction: }Alice and Bob apply an error correction protocol. If the protocol succeeds, Alice and Bob obtain raw keys $\kappa^A$ and $\kappa^B$.
\item \emph{Parameter estimation: }Bob sets $c_i=\perp$ if $t_i=0$ and $c_i=\chi[a_i\cdot b_i=(-1)^{x_i\cdot y_i}]$ if $t_i=1$. He aborts the protocol if $\sum_{i=1}^{n}c_i<(\omega_{exp}\gamma-\delta_{est})n$.
\item \emph{Privacy amplification: }Alice and Bob apply privacy amplification on $\kappa^A$ and $\kappa^B$ and obtain final keys of length $k$.
\end{enumerate}
\end{mybox}

\begin{theorem}[Theorem 5.1 in Ref.~\cite{arnon2019simple}]
For the DIQKD protocol in Box~\ref{box:OriginEAT}, suppose the expected CHSH winning probability is $\omega_{exp}\in[0,1]$. Given fixed parameters $\varepsilon_s,\varepsilon_{EC},\varepsilon_{PA},\varepsilon_{EA}\in(0,1)$, with the total soundness error $\varepsilon_{QKD}^s\leq\varepsilon_s+\varepsilon_{EC}+\varepsilon_{PA}+\varepsilon_{EA}$, a $k$-bit final key can be obtained in a successful experiment,
\begin{equation}\label{Eq:EATKeyRate}
\begin{split}
  k=n\cdot\eta_{opt}\left(\frac{\varepsilon_s}{4},\varepsilon_{EA}+\varepsilon_{EC}\right)-\text{leak}_{EC}-3\log\left(1-\sqrt{1-\left(\frac{\varepsilon_s}{4}\right)^2}\right) \\-\gamma n-\sqrt{n}2\log7\sqrt{1-2\log(\frac{\varepsilon_s}{4}\cdot(\varepsilon_{EA}+\varepsilon_{EC}))}+2\log\varepsilon_{PA}.
\end{split}
\end{equation}
Here, $\text{leak}_{EC}$ is the amount of private randomness required for information reconciliation, $\varepsilon_s$ is the smoothing parameter in the smooth min-entropy, and $\varepsilon_{EC},\varepsilon_{PA},\varepsilon_{EA}$ are the failure probabilities of information reconciliation, privacy amplification, and an EAT protocol, respectively.
The function $\eta_{opt}(\cdot,\cdot)$ is defined as follows,
\begin{eqnarray}
    g(p) &=& \left\{
    \begin{tabular}{ll}
        $1-h\left(\frac{1}{2}+\frac{1}{2}\sqrt{16\frac{p}{\gamma}\left(\frac{p}{\gamma}-1\right)+3}\right)$, &$\frac{p}{\gamma}\in\left[\frac{3}{4},\frac{2+\sqrt{2}}{4}\right],$\\
        $1$, &$\frac{p}{\gamma}\in\left[\frac{2+\sqrt{2}}{4},1\right],$
        \end{tabular}
          \right.
\end{eqnarray}
\begin{eqnarray}\label{Eq:CutAndGlueMinTradeoff}
    f_{\min}(p,p_t) &=& \left\{
    \begin{tabular}{ll}
        $g(p)$, &$p\leq p_t,$\\
        $\frac{\mathrm{d}}{\mathrm{d} p}g(p)|_{p_t}\cdot p+\left(g(p_t)-\frac{\mathrm{d}}{\mathrm{d} p}g(p)|_{p_t}\cdot p_t\right)$, &$p> p_t,$
        \end{tabular}
          \right.
\end{eqnarray}
\begin{equation}
\begin{split}
  \eta(p,p_t,\varepsilon_s,\varepsilon_e)=f_{\min}(p,p_t)-\frac{1}{\sqrt{n}}2\left(\log13+\left\lceil\frac{\mathrm{d}}{\mathrm{d}p}g(p)|_{p_t}\right\rceil\right)\sqrt{1-2\log(\varepsilon_s\cdot\varepsilon_e)},
\end{split}
\end{equation}
\begin{equation}\label{Eq:OptEta}
\begin{split}
  \eta_{opt}(\varepsilon_s,\varepsilon_e)&=\max_{p_t}\eta(\omega_{exp}\gamma-\delta_{est},p_t,\varepsilon_s,\varepsilon_e),\\
  \text{s.t. }\frac{3}{4}&<\frac{p_t}{\gamma}<\frac{2+\sqrt{2}}{4}.
\end{split}
\end{equation}
\label{Thm:OriginEAT}
\end{theorem}

In~\cite{arnon2018practical,arnon2019simple}, the authors have also applied EAT for the evaluation of $\text{leak}_{EC}$, which corresponds to the cost of private randomness in information reconciliation. While to have a relatively fair comparison between our analysis and EAT in privacy estimation, in the numerical simulation, we shall take the result in the Shannon limit for information reconciliation. In the current EAT analysis, the cost for information reconciliation converges to the following formula with increasing data size,
\begin{equation}\label{Eq:LeakEC}
  \text{leak}_{EC}=n[\gamma h(\omega_{exp})+(1-\gamma)h(e_b)],
\end{equation}
where $e_b$ is the bit error rate in the rounds with $t_i=0$.

The second result we consider is a modified EAT-based DIQKD security analysis~\cite{murta2019towards}. The protocol is shown in Box~\ref{box:NewEAT}. In this protocol, a statistical technique first introduced in Ref.~\cite{arnon2019device} is applied for better finite-size performance, where the total number of experimental rounds is not fixed in advance. Instead, Alice and Bob take quantum measurements in blocks, where each block either ends with a Bell test, or when a maximal allowed number of rounds inside a block is reached. Apart from the protocol design, Ref.~\cite{murta2019towards} also optimises the key rate formula. The result is given in Theorem~\ref{Thm:NewEAT}.

\begin{mybox}[label={box:NewEAT}]{{DIQKD Protocol in Ref.~\cite{murta2019towards}}}
\textbf{Arguments: }\\
$m$: the number of blocks of quantum measurements \\
$\gamma\in(0,1]$: the ratio of Bell test rounds \\
$s_{\max} = \lceil\frac{1}{\gamma}\rceil$: the maximum number of rounds inside a block \\
$\bar{s}=\frac{1-(1-\gamma)^{s_{\max}}}{\gamma}$: the expected number of rounds in a block \\
$n=m\bar{s}$: the expected number of rounds \\
$\omega_{exp}$: expected CHSH winning probability \\
$\delta_{est}\in(0,1)$: width of the statistical confidence interval for parameter estimation \\
$\chi$: indicator function, taking its value as $1$ for a true statement and $0$ for a false statement

\tcblower
\begin{enumerate}[leftmargin=*]
\item\emph{Measurement: }For every block $j\in[m]$, Alice and Bob
    \begin{enumerate}[leftmargin=*]
    \item set $i=0$;
    \item while $i\leq s_{\max}$,
    \begin{enumerate}
      \item set $i = i+1$;
      \item randomly set $t_i\in\{0,1\}$ such that $\Pr(t_i=1)=\gamma$;
      \item if $t_i=0$, set measurement bases $(x_i,y_i)=(2,0)$; otherwise, randomly set measurement bases $x_i,y_i\in\{0,1\}$;
      \item record the outputs of the devices $a_i,b_i\in\{\pm1\}$;
      \item if $t_i=1$, set $i = s_{\max}+1$; Bob sets $c_i=\chi[a_i\cdot b_i=(-1)^{x_i\cdot y_i}]$ if $t_i=1$.
    \end{enumerate}
    \end{enumerate}
\item \emph{Error correction: }Alice and Bob apply an error correction protocol. If the protocol succeeds, Alice and Bob obtain raw keys $\kappa^A$ and $\kappa^B$.
\item \emph{Parameter estimation: }
    Bob aborts the protocol if $\sum_{i=1}c_i<m(\omega_{exp}\gamma-\delta_{est})[1-(1-\gamma)^{s_{\max}}]$.
\item \emph{Privacy amplification: }Alice and Bob apply privacy amplification on $\kappa^A$ and $\kappa^B$ and obtain final keys of length $k$.
\end{enumerate}
\end{mybox}

\begin{theorem}[Theorem 1 in Ref.~\cite{murta2019towards}]
For the DIQKD protocol in Box~\ref{box:NewEAT}, suppose the expected CHSH winning probability is $\omega_{exp}\in[0,1]$. Given fixed parameters $\varepsilon_s,\varepsilon_{EC},\varepsilon_{PA},\varepsilon_{EA}\in(0,1)$, with the total soundness error $\varepsilon_{QKD}^s\leq2\varepsilon_{EC}+\varepsilon_{PA}+\varepsilon_s$, the DIQKD protocol either aborts with probability higher than $1-(\varepsilon_{EA}+\varepsilon_{EC})$, or it generates a secure key of length
\begin{equation}
\begin{aligned}
  k=&\frac{n}{\bar{s}}\eta_{opt}-\frac{n}{\bar{s}}h(\omega_{exp}-\delta_{est})-\sqrt{\frac{n}{\bar{s}}}\nu_1-\text{leak}_{EC}\\
  &-3\log(1-\sqrt{1-\left[\frac{\varepsilon_s}{4(\varepsilon_{EA}+\varepsilon_{EC})}\right]^2})+2\log(\frac{1}{2\varepsilon_{PA}}),
\end{aligned}
\end{equation}
where
\begin{equation}
  \nu_1=2\left(\log7+\left\lceil\dfrac{h'(\omega_{exp}+\delta_{est})}{1-(1-\gamma)^{s_{\max}}}\right\rceil\right)\sqrt{1-2\log\varepsilon_s},
\end{equation}
$\text{leak}_{EC}$ is the amount of private randomness required for information reconciliation, $\varepsilon_s$ is the smoothing parameter in the smooth min-entropy, and $\varepsilon_{EC},\varepsilon_{PA},\varepsilon_{EA}$ are the failure probabilities of information reconciliation, privacy amplification, and an EAT protocol, respectively.
The term $\eta_{opt}$ is defined as follows,
\begin{equation}
  \eta_{opt}=\max_{\dfrac{3}{4}<\dfrac{p_t}{1-(1-\gamma)^{s_{\max}}}<\dfrac{2+\sqrt{2}}{4}}\left[F_{\min}(\omega_{exp}[1-(1-\gamma)^{s_{\max}}]-\delta_{est},p_t)-\frac{1}{\sqrt{m}}\nu_2\right],
\end{equation}
\begin{equation}
  F_{\min}(p,p_t)=\frac{\mathrm{d}}{\mathrm{d} p}g(p)|_{p_t}\cdot p+\left(g(p_t)-\frac{\mathrm{d}}{\mathrm{d} p}g(p)|_{p_t}\cdot p_t\right),
\end{equation}
\begin{equation}
  g(p)=\bar{s}\left[1-h\left(\frac{1}{2}+\frac{1}{2}\sqrt{16\frac{p}{1-(1-\gamma)^{s_{\max}}}\left(\frac{p}{1-(1-\gamma)^{s_{\max}}}-1\right)+3}\right)\right],
\end{equation}
\begin{equation}
  \nu_2=2\left(\log(1+6\cdot2^{s_{\max}})+\left\lceil\frac{\mathrm{d}}{\mathrm{d} p}g(p)|_{p_t}\right\rceil\right)\sqrt{1-2\log\varepsilon_s}.
\end{equation}
\label{Thm:NewEAT}
\end{theorem}

Similarly, we shall take the Shannon limit for information reconciliation for the simulation of this result, given by Eq.~\eqref{Eq:LeakEC}. Note that since the number of rounds is not fixed, the term $n$ should be interpreted as the expected number of rounds. Nevertheless, with a large number of blocks $m$, $n$ shall be close to $m\bar{s}$.

\subsubsection{Complementarity-based analysis for the spot-checking protocol}
To have a full comparison between the complementarity-based method and the EAT-based method, we also apply the complementarity approach to the spot-checking protocol. The protocol we consider is similar to the one in Box~\ref{box:OriginEAT}. Here, we simply state the difference from the standard protocol listed in Box~\ref{box:ProtocolDetail}:\\
\emph{Measurement: }For every round $i\in[n]$, Alice and Bob
\begin{enumerate}[leftmargin=*]
    \item randomly set $t_i\in\{0,1\}$ such that $\Pr(t_i=1)=\gamma$;
    \item if $t_i=0$, set measurement bases $(x_i,y_i)=(2,0)$; otherwise, randomly set $x_i,y_i\in\{0,1\}$;
    \item record the outputs of the devices $a_i,b_i\in\{\pm1\}$.
\end{enumerate}
The rest steps are the same as in Box~\ref{box:ProtocolDetail}.

Before we commence, we make some remarks.
\begin{enumerate}
  \item\emph{Consumption of pre-shared secret keys for measurement settings: }In the context of QKD, local private randomness is a free resource, while shared randomness, or, secret keys, cannot be taken freely. Different from in the standard protocol, Alice and Bob need to consume pre-shared secret keys to set up a coordinated input probability distribution in the spot-checking protocol, i.e., for the joint random choice of $t_i$ in each round. This results in additional consumption in the key length calculation.
  \item\emph{Cost of pre-shared keys in information reconciliation: }The raw keys are acquired from the rounds with $(x_i,y_i)=(2,0)$, or $t_i=0$ in the spot-checking protocols. The cost of secret keys in the information reconciliation step differs depending on the analysis method. In the complementarity approach, as stated by the security statements in Sec.~\ref{Supp:SecDef}, one can sift the key generation rounds in evaluating the information reconciliation cost. For instance, in the spot-checking protocol, the cost is roughly given by $n(1-\gamma)h(e_b)$. On the contrary, in the current EAT-based method, as the key length is analysed by entropy accumulation under the Markov-chain condition, the information reconciliation cost in terms of entropy has the additional term $n\gamma h(\omega_{exp})$, as shown in Eq.~\eqref{Eq:LeakEC}.
  \item\emph{Pre-determined abort condition: }In the protocol listed in Box~\ref{box:OriginEAT}, there is a pre-determined abort condition assessed by the observed Bell value, which relates to the completeness parameter. The threshold Bell value enters the key-length calculation. In the complementarity approach, we do not explicitly assign the abort condition. The protocol is aborted if the observed Bell statistics cannot support positive key generation. In other words, the abort condition is embedded in parameter estimation. We shall come to this point later in Sec.~\ref{Sec:NumResults} when presenting numerical simulation results.
\end{enumerate}

When adopting the complementarity approach in the spot-checking protocol, the key rate analysis is basically the same as for the standard protocol. One simply needs to reconstruct the random variables for the martingale-based privacy estimation and take account for the additional cost of pre-shared secret keys for establishing the measurement settings. Here, we present the final result.

\begin{theorem}[Key length for the spot-checking protocol, complementarity approach]
For the spot-checking DIQKD protocol, suppose the number of experimental rounds is $n$ and $\Pr(t_i=1)=\gamma\in(0,1]$. Consider failure probabilities in phase-error sample entropy estimation $\varepsilon_{pe}\in(0,1)$, in parameter estimation of the Bell value $\varepsilon_{S}\in(0,1)$, and in privacy amplification $\varepsilon_{pc}\in(0,1)$, such that they satisfy $\varepsilon = \varepsilon_{pe} + \varepsilon_{S} + \varepsilon_{pc}\in(0,1)$. Choose parameters $\xi\in\left(0,\frac{1}{2}\right)$, $\delta_{p}\in [0,\infty),\delta_{p}'\in\mathbb{R}$, and $\delta_{S}\in [0,\infty),\delta_{S}'\in\mathbb{R}$ that satisfy
\begin{equation}
\begin{split}
   \varepsilon_{pe} &= \exp[-\frac{2n(\delta_{p}^2-\delta_{p}'^2)}{\left(1+\frac{4\delta_{p}'}{3}\right)^2}], \\
   \frac{1}{4}\varepsilon_{S} &= \exp[-\frac{2n(\delta_{S}^2-\delta_{S}'^2)}{\left(1+\frac{4\delta_{S}'}{3}\right)^2}]. \\
\end{split}
\end{equation}
With a failure probability no larger than $\varepsilon$, the cost for privacy amplification can be upper bounded by
\begin{equation}
  I_{pa}\leq \frac{1}{1+2\delta_{p}'}n\left\{(1-\gamma)h\left[e_{p}^{\xi} (S_{est})\right] + c_{\xi}(\delta_{p}+\delta_{p}') \right\}-\log\varepsilon_{pc},
\end{equation}
where $S_{est}$ is the estimated expected Bell value,
\begin{equation}
\begin{split}
   S_{est}&= \sum_{x,y\in\{0,1\}}\left\{\dfrac{8\left[ (1+2\delta_{S}')(m_{xy}-q_{xy})-n(\delta_{S}+\delta_{S}')\right]}{n\gamma}-1\right\}, \\
\end{split}
\end{equation}
the $\xi$-regularised phase-error probability function $e_p^{\xi}(S)$ is given in Eq.~\eqref{Eq:SmoothPhaseErrorProb}, and $c_{\xi}$ is a constant originated from the regularisation parameter $\xi$,
\begin{equation}
  c_{\xi}= \log(\frac{1+2\xi}{\xi}).
\end{equation}
The amount of generated key bits is given by
\begin{equation}
  k = m [
  1-f_{ec}h(e_b)]-nh(\gamma)-I_{pa},
\end{equation}
where $e_b$ is the quantum bit error rate, $f_{ec}$ is the efficiency parameter for information reconciliation, and other parameters are the same as in Box~\ref{box:ProtocolDetail}.
\end{theorem}

\subsubsection{Numerical results}\label{Sec:NumResults}
The complementarity and EAT methods are both tight in the asymptotic limit of infinite data size with the difference of a basis sifting factor, converging to the asymptotic key rate in~\cite{pironio2009device}. Hence, we shall focus on the comparison of their finite-size performance.

In the first place, we state the parameter settings in the numerical simulation in this section. We have also taken these settings for numerical simulation in the main text. For the security parameters, the total soundness error is set as $\varepsilon_{QKD}^s=1\times10^{-5}$. By Eq.~\eqref{Supp:FailProbtoSound}, the total failure probability when adopting the complementarity-based analysis is taken as $\varepsilon_f\approx{\varepsilon_{QKD}^s}^2/2=5\times10^{-11}$. The total soundness error or the total failure probabilities is composed of several terms, corresponding to various stages in the security analyses. In general, an optimisation could be employed for the assignment of their values. For simplicity, we take a heuristic optimisation of these terms. We summarise the components in the total soundness error or the total failure probability in each protocol. For the simulation of the complementarity-based approach, the total failure probability is upper-bounded by
\begin{equation}\label{Eq:FailAssign}
\varepsilon_f=\varepsilon_{S}+\varepsilon_{pe}+\varepsilon_{pc},
\end{equation}
where $\varepsilon_{S},\varepsilon_{pe}$ are failure probabilities of parameter estimation, and $\varepsilon_{pc}$ is the failure probability of privacy amplification. In Box~\ref{box:ProtocolDetail}, the failure probability in estimating the Bell value is given by $\varepsilon_{S}=\sum_{x,y\in\{0,1\}}\varepsilon_{xy}$. For the original EAT-based approach in Box~\ref{box:OriginEAT}, the total soundness error is given by
\begin{equation}
  \varepsilon_{QKD}^s=\varepsilon_s+\varepsilon_{EC}+\varepsilon_{PA}+\varepsilon_{EA}.
\end{equation}
For the modified EAT-based approach in Box~\ref{box:NewEAT}, the total soundness error is given by
\begin{equation}
  \varepsilon_{QKD}^s=2\varepsilon_{EC}+\varepsilon_{PA}+\varepsilon_s.
\end{equation}

In simulation, the expected values of random variables shall be determined under a set of experimental parameters. To model the statistical fluctuation in an actual experiment, we use the normal distribution to calculate the deviation term. In the EAT-based protocols, this issue is covered in the completeness error, $\varepsilon_{QKD}^{c}$. We take the typical value considered in~\cite{arnon2019simple},
\begin{equation}\label{Eq:Deviation}
  \delta_{est} = \sqrt{-\frac{\ln{\varepsilon_{QKD}^{c}}}{2\gamma n}},
\end{equation}
and $\varepsilon_{QKD}^{c}=1\times10^{-2}$. In the complementarity-based protocol, one does not need to pre-determine an abort criterion. Instead, one takes the observed values and directly evaluates key generation by Theorem~\ref{Thm:SuppMainResult}. Nevertheless, to derive a relatively fair comparison with EAT-based results, we simulate a worst-case observed Bell value that corresponds to the deviation in Eq.~\eqref{Eq:Deviation}, where the observed Bell value shall be simulated by
\begin{equation}
  \bar{S} = 8(\omega_{exp}-\delta_{est})-4.
\end{equation}
Note that in the complementarity-based analysis, one estimates the expected Bell value for each input setting separately. We consider that the four input settings share a symmetric behaviour. That is, the ratios of $m_{xy}$ to $q_{xy}$ are almost the same.

In addition, there is the issue to optimise the probability distribution of measurement settings. For the spot-checking protocols, we optimise over the ratio of Bell tests, $\gamma$. For the standard protocol, we take the following input probability distributions,
\begin{equation}
\begin{split}
   p_X(0)&=p_X(1)=\frac{\gamma}{2}, p_X(2)=1-\gamma, \\
   p_Y(0)&=p_Y(1)=\frac{1}{2},
\end{split}
\end{equation}
where $\gamma\in(0,1)$. Then, the ratio of Bell-test rounds is $\gamma$ and the ratio of key generation rounds is $(1-\gamma)/2$. We optimise over $\gamma$ in the simulation.

We first study the single-parameter influence on the key generation performance, where we examine the influence of noise and loss individually. In simulation, we take the average number of final secure bits per key-generation round, $k/m$, as the figure of merit here. The simulation result is shown in Figure~\ref{Fig:SingleParameters}. We can see that channel loss and noise affect the key rate in a similar manner; hence, we can focus on either of the parameters. With increasing data sizes, the results converge to the asymptotic result in the limit of infinite data size, consistent with the result assuming i.i.d.~attacks \cite{pironio2009device}.

\begin{figure*}[hbt!]
\hfil
\vbox{
\hsize 0.4\textwidth
\hbox{
\begin{overpic}[width=\hsize]{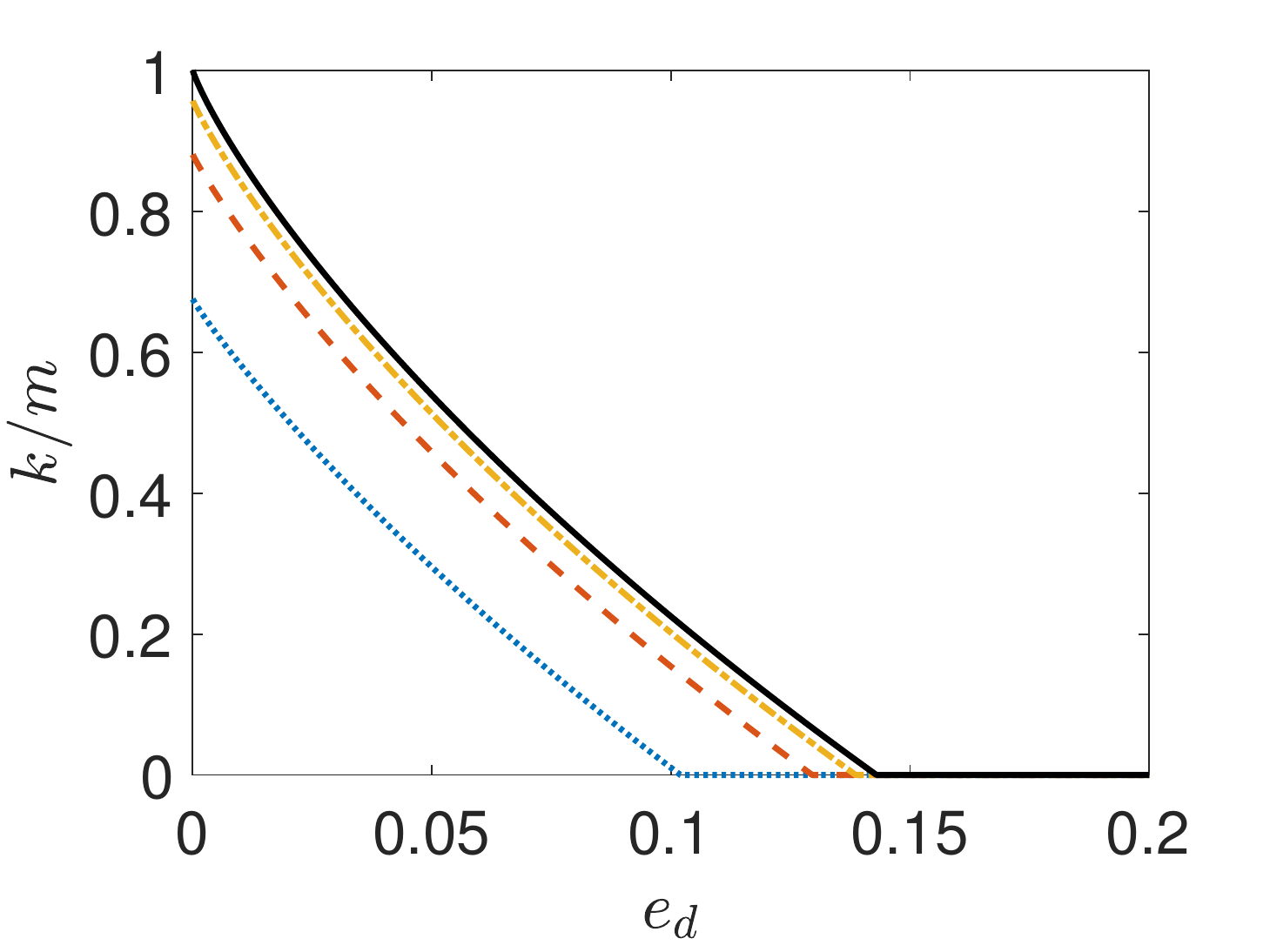}
\put(49,-4)
{\fontsize{10}{10}\selectfont (a)}
\put(45,50)
{\fontsize{10}{10}\selectfont $\eta=100\%$
}
\end{overpic}
\hbox{
{\begin{overpic}[width=\hsize]{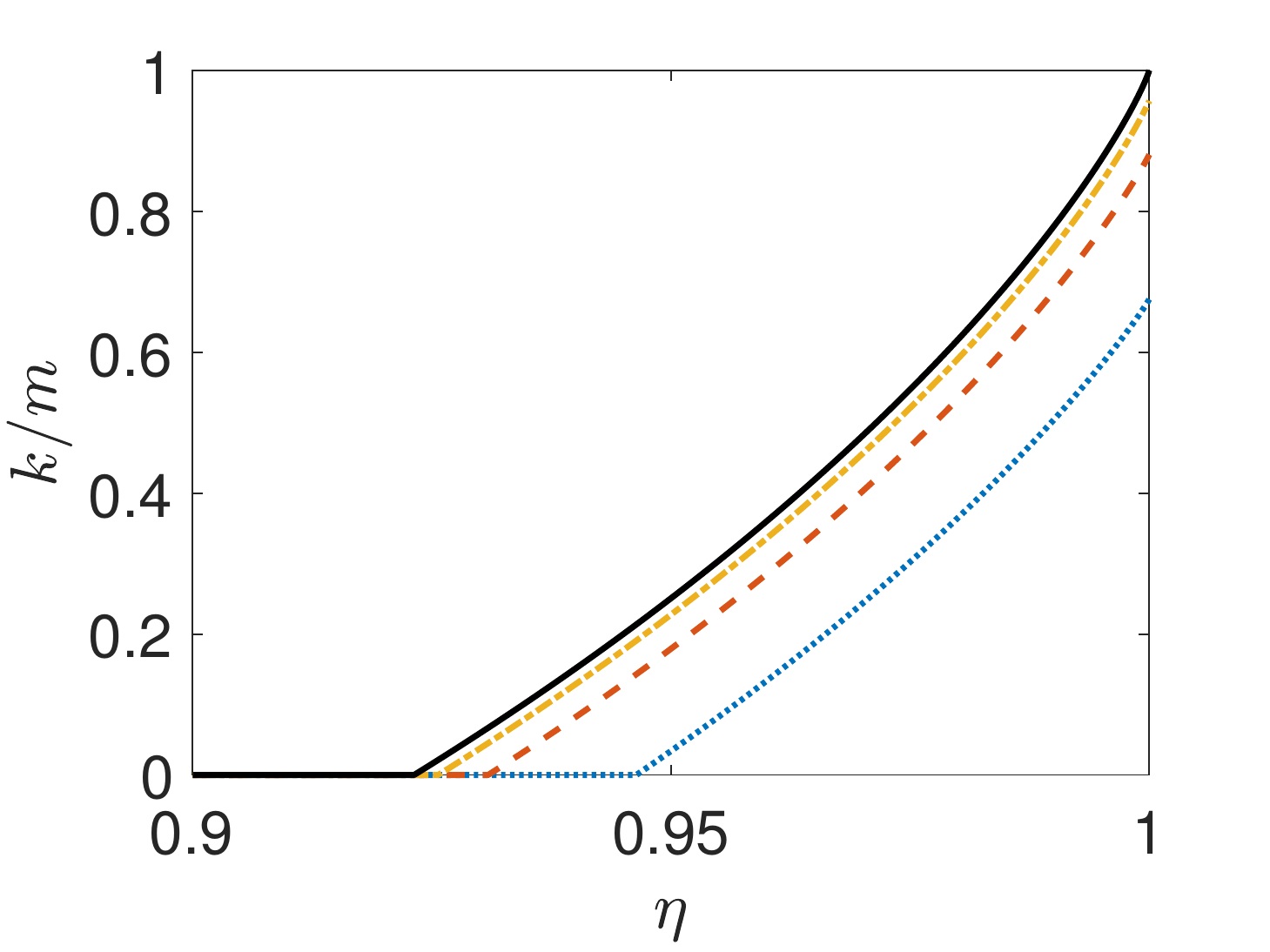}
\put(49,-4)
{\fontsize{10}{10}\selectfont (b)}
\put(35,50)
{\fontsize{10}{10}\selectfont $e_d=0$}
\end{overpic}}
}
}
}
\hfil
\caption{Simulation results of key-generation performance in the standard protocol. We consider the data sizes of $n=10^{6},10^7,10^8$ and $n\rightarrow\infty$, represented by the blue dotted lines, orange dashed lines, yellow dash-dot lines, and black solid lines, respectively. The transmittance is set equal on both sides, $\eta_A=\eta_B=\eta$. (a) The key rate $k/m$ with respect to $e_d$ at $\eta=1$. (b) The key rate $k/m$ with respect to $\eta$ at $e_d=0$.}
\label{Fig:SingleParameters}
\end{figure*}

Similar to the numerical simulation in the main text, we present the smallest data size required by the complementarity approach and the EAT approach with respect to the total transmittance in Figure~\ref{Fig:ThresholdEff}. In this simulation, we set the state preparation fidelity to be unity. The motivation for this simulation comes from the photonic platforms, where the entangled state could be distributed with a high fidelity, while the detection efficiency is relatively limited. For completeness, we consider the smallest required data sizes for all the protocols and methods discussed above: applying the complementarity approach to the standard protocol in Box~\ref{box:ProtocolDetail}, applying the complementarity approach to the spot-checking protocol, applying the original EAT-based method to the spot-checking protocol in Box~\ref{box:OriginEAT}, and applying the modified EAT-based method to the spot-checking protocol in Box~\ref{box:NewEAT}. We label the data sizes required at the efficiencies of $92.6\%$ and $96.0\%$. As a reference, we also label the data points at the unity transmittance. The complementarity approach allows for a more flexible protocol design, where the users do not need to apply a spot-checking protocol and are allowed to set their own measurement settings locally at random. With this advantage, the result of applying the complementarity approach to the standard protocol yields the best result. In comparison to the EAT-based analysis that optimise both the spot-checking protocol design and the key-rate formula, our method cuts down the requirement on the data size and hence the experimental time span by one to two orders of magnitude. Even when we restrict to the same spot-checking protocol, our method is still advantageous to the original EAT-based result in saving the time span for positive key generation.

\begin{figure*}
{\begin{overpic}[width=1\hsize]{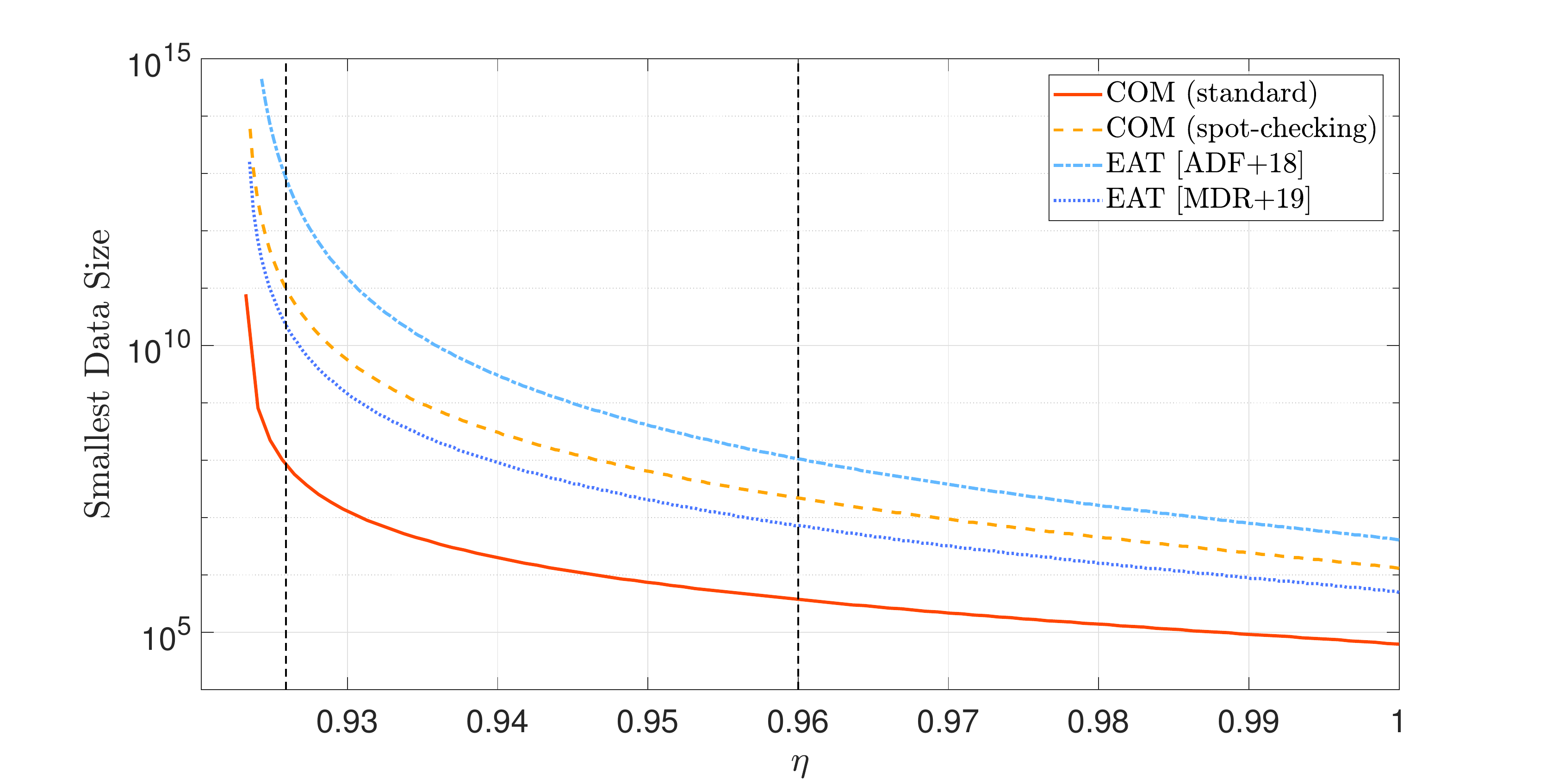}
\put(19,38)
{\fontsize{7}{7}\selectfont
\begin{tabular}{c}
$n=8.51\times10^{12}$
\end{tabular}
}
\put(19,32)
{\fontsize{7}{7}\selectfont
\begin{tabular}{c}
$n=1.01\times10^{11}$
\end{tabular}
}
\put(19,29)
{\fontsize{7}{7}\selectfont
\begin{tabular}{c}
$n=2.40\times10^{10}$
\end{tabular}
}
\put(19,20)
{\fontsize{7}{7}\selectfont
\begin{tabular}{c}
$n=8.32\times10^{7}$
\end{tabular}
}

\put(52,21.5)
{\fontsize{7}{7}\selectfont
\begin{tabular}{c}
$n=1.06\times10^{8}$
\end{tabular}
}
\put(52,19)
{\fontsize{7}{7}\selectfont
\begin{tabular}{c}
$n=2.16\times10^{7}$
\end{tabular}
}
\put(52,17)
{\fontsize{7}{7}\selectfont
\begin{tabular}{c}
$n=7.24\times10^{6}$
\end{tabular}
}

\put(52,12.5)
{\fontsize{7}{7}\selectfont
\begin{tabular}{c}
$n=3.75\times10^{5}$
\end{tabular}
}

\put(79,17.5)
{\fontsize{7}{7}\selectfont
\begin{tabular}{c}
$n=4.07\times10^{6}$
\end{tabular}
}
\put(79,15)
{\fontsize{7}{7}\selectfont
\begin{tabular}{c}
$n=1.30\times10^{6}$
\end{tabular}
}
\put(79,12.5)
{\fontsize{7}{7}\selectfont
\begin{tabular}{c}
$n=4.94\times10^{5}$
\end{tabular}
}
\put(79,10)
{\fontsize{7}{7}\selectfont
\begin{tabular}{c}
$n=6.21\times10^{4}$
\end{tabular}
}
\put(15,4)
{\fontsize{8}{8}\selectfont
\begin{tabular}{c}
$92.6\%$ \\
\end{tabular}
}
\end{overpic}
}
\caption{The smallest data sizes required for successful key generation with respect to transmittance at $e_d=0$. We denote the result by applying the complementarity approach to the standard protocol as `COM (standard)', the result by applying the complementarity approach to the spot-checking protocol as `COM (spot-checking)', the original EAT-based result as `EAT [ADF+18]'~\cite{arnon2018practical}, and the modified EAT-based result as `EAT [MDR+19]'~\cite{murta2019towards}, respectively. The detector efficiencies are set equal on both sides, $\eta_A=\eta_B=\eta$. In the simulation of the standard protocol, we set Bob's input settings to be chosen uniformly at random and optimise over the probability distribution of Alice's input settings. In the simulation of the spot-checking protocols, we optimise over the ratio of Bell test rounds.}
\label{Fig:ThresholdEff}
\end{figure*}
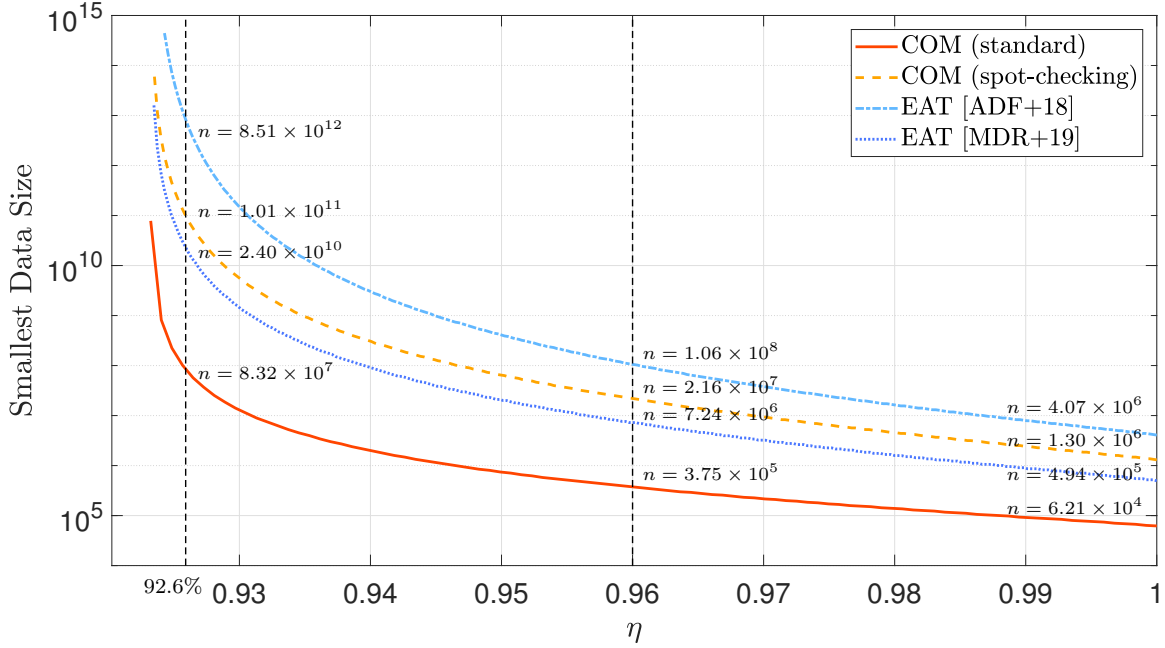

To show the convergence behaviour with an increasing data size, we consider the spot-checking protocol and simulate the key rate given by the complementarity approach, as shown in Figure~\ref{Fig:RateRound}. As a fair comparison, we also simulate the key rate given by the original EAT-based method under the same protocol. We take the expected Bell value to be $S_{exp}=2.47$ and the quantum bit error rate to be $e_b=0.051$, which are readily-implementable in nitrogen-vacancy (NV)-centre experimental set-ups~\cite{murta2019towards}. It can be seen that under realistic experimental parameters of the state-of-the-art NV-centre platforms, the complementarity-based method yields a faster convergence.

\begin{figure}[hbt!]
\centering \resizebox{8cm}{!}{\includegraphics{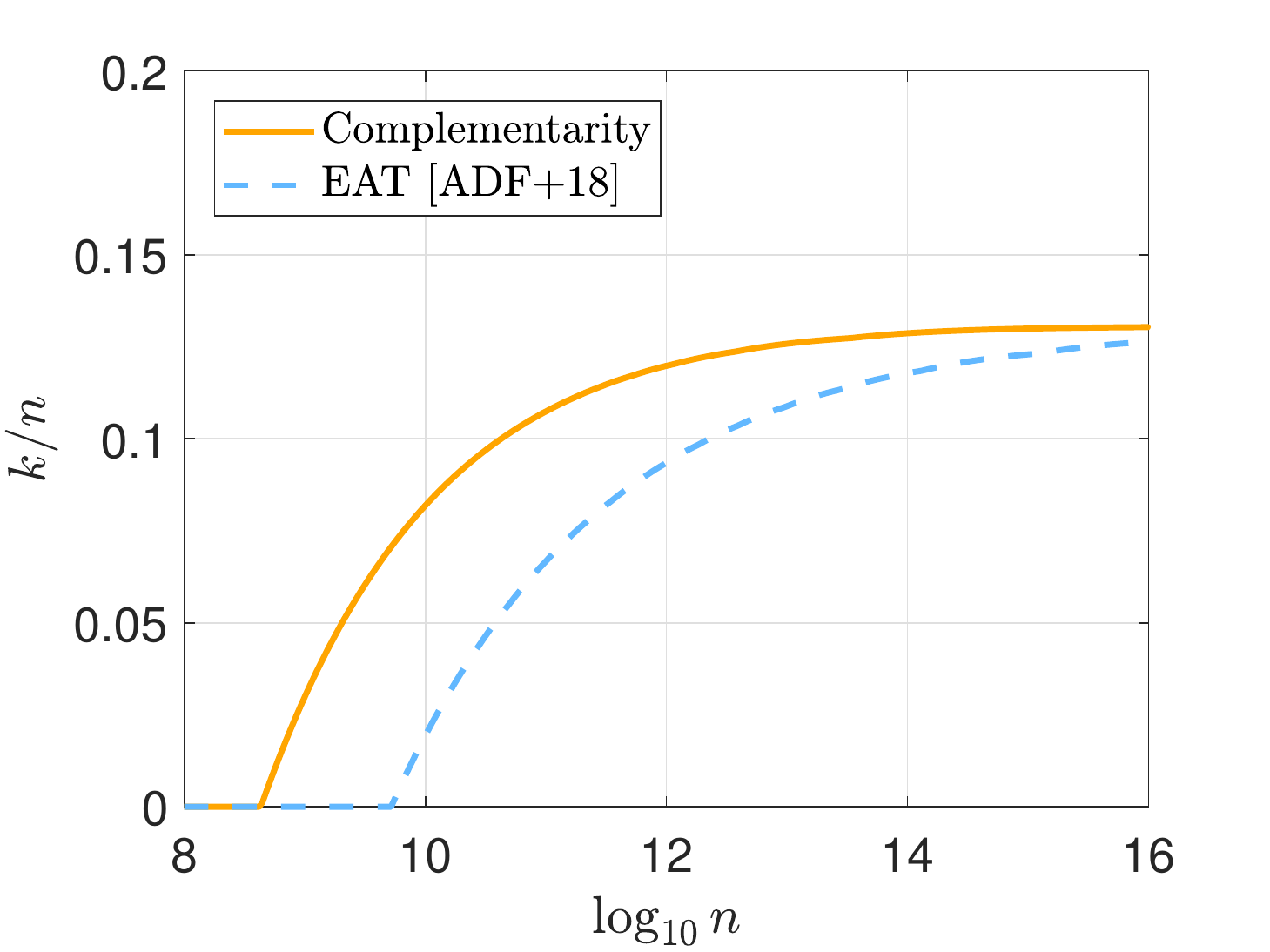}}
\caption{The key rate versus the number of rounds. The key rate refers to the ratio of the key length to the number of total experimental rounds, $k/n$.}
\label{Fig:RateRound}
\end{figure}

We remark that one may perform a more delicate optimisation over all the involved parameters in each protocol, including the full probability distribution of measurement settings and the assignment of failure probabilities in each process. Nevertheless, the advantage shown in these simulation results remains valid in the order of magnitudes. We also validate the results in accordance with the descriptions in Refs.~\cite{arnon2018practical,arnon2019simple,murta2019towards}.

\subsubsection{Data analysis of a recent experiment}\label{Sec:IonSim}
Upon finishing our manuscript, we notice a recent experimental work of a DIQKD experimental realisation on the ion-trap platform~\cite{nadlinger2022experimental}. In this work, a CHSH Bell value of $S=2.64$ and a quantum bit error rate of $e_b=1.8\%$ have been reported. The authors have employed a spot-checking protocol with one-way classical communication and a modified EAT-based result that takes advantage of the second-order analysis~\cite{liu2021device}. They succeed in obtaining $95884$ secure key bits in approximately $n=1.5\times10^6$ rounds with a total soundness error of $\varepsilon_{QKD}^s=1\times10^{-10}$. By simulation, positive key generation can be obtained after approximately $10^6$ rounds under this soundness error (Figure 4 in the cited reference).

Under the reported experimental performance and security parameters, we analyse the potential key generation results when adopting our complementarity-based security analysis. We shall use the protocol in Box~\ref{box:ProtocolDetail}. For key generation simulation,
\begin{enumerate}
  \item In information reconciliation, we use Eq.~\eqref{Supp:InfoRec} and
      \begin{equation}
        H(\kappa^B|\kappa^A)=mh(e_b),e_b=1.8\%.
      \end{equation}
      For the efficiency of information reconciliation, $f_{ec}$, we take the result in~\cite{tang2021shannon}, where $f_{ec}=1.09$ for a wide range of quantum bit error rates. We remark that a more efficient error correction is possible, see, e.g., Ref.~\cite{mao2022high}.
  \item In privacy amplification, we take the following term in correspondence with the reported Bell value,
      \begin{equation}
      \bar{S}\equiv\sum_{x,y\in\{0,1\}}\left[\dfrac{2\left( m_{xy}-q_{xy}\right)}{np_X(x)p_Y(y)}-1\right]=2.64.
      \end{equation}
  \item For the security parameter, we take
      \begin{equation}
        \varepsilon_f\approx\frac{{\varepsilon_{QKD}^s}^2}{2}=5\times10^{-21}.
      \end{equation}
\end{enumerate}

With the above parameter settings, we simulate the potential experimental results if our security analysis is adopted. We find that the least data size required for positive key generation drops down to $n=3.16\times10^5$, which is smaller than the required data size of approximately $10^6$ by the modified EAT-based result. The experimental time span could be shortened by nearly $70\%$~\footnote{We do not reproduce the results in Ref.~\cite{nadlinger2022experimental} due to lack of some optimisation details in the advanced EAT-based results. We refer readers who are interested in this issue to Figure 4 in Ref.~\cite{nadlinger2022experimental}.}. In Figure~\ref{Fig:IonSim}, we also depict the convergence behaviour of the key rate with respect to an increasing data size.

\begin{figure}[hbt!]
\centering \resizebox{8cm}{!}{\includegraphics{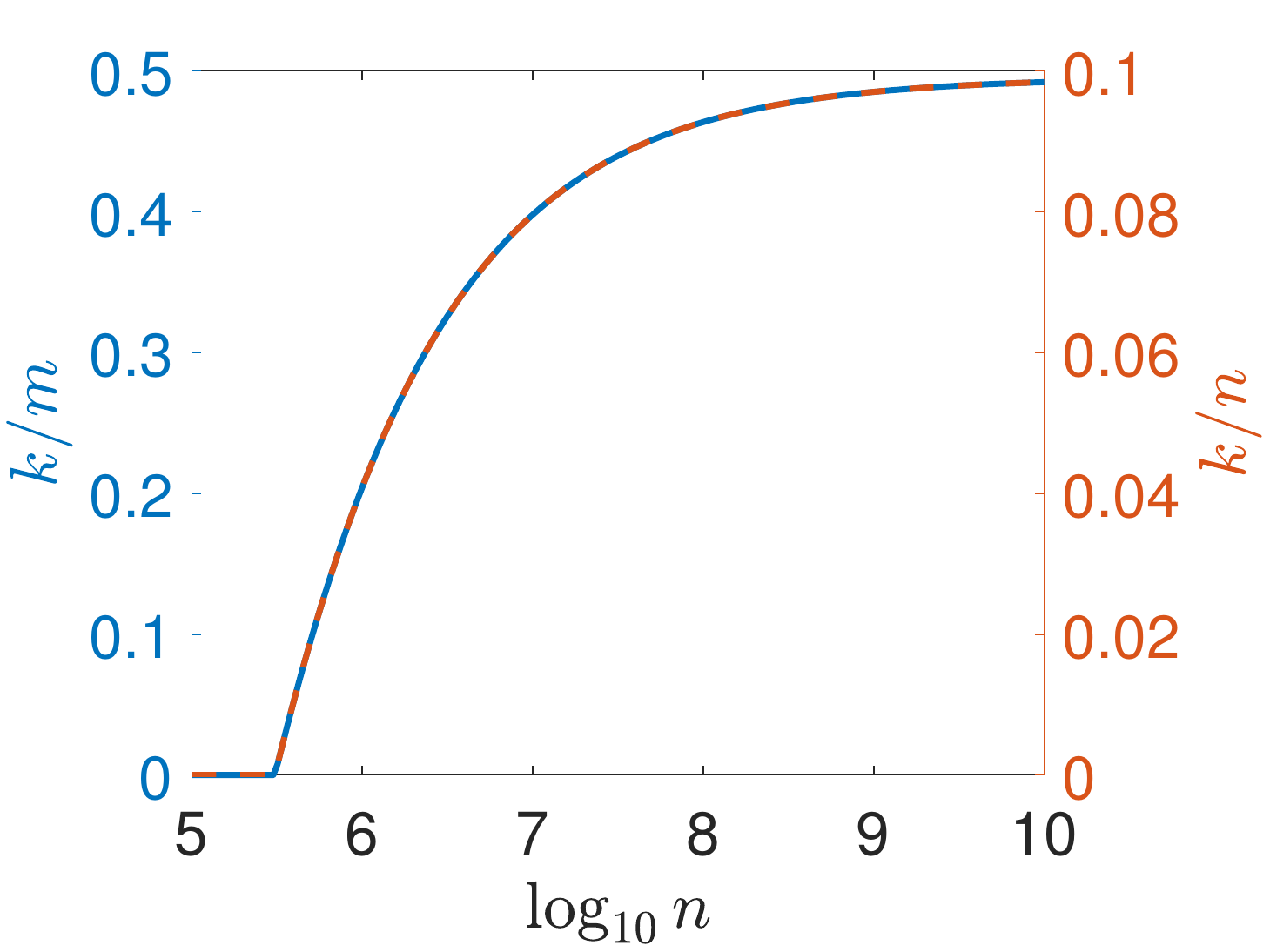}}
\caption{Examination of reported experimental data in Ref.~\cite{nadlinger2022experimental} using the complementarity-based analysis. The curve with respect to the left vertical axis shows the average number of generated key bit in key generation rounds, $k/m$, and the curve with respect to the left vertical axis shows the key rate, $k/n$.}
\label{Fig:IonSim}
\end{figure}

\subsubsection{Summary}\label{Sec:EATSummary}
We end the section of numerical comparison with a brief summary. In Table~\ref{tab:CompareSummary}, we qualitatively summarise the performance of the two analysis methods.

\begin{table}[htb]
\centering
\caption{Qualitative summary of comparison between the complementarity approach and the EAT-based method.}
\begin{tabular}{c|cc|c|c|c}
\hline
 analysis & \multicolumn{2}{|c|}{protocol} & smallest data size & convergence & key expansion, $k/n$ \\
\hline
\multirow{2}{*}{Complementarity}  & local basis choice & $\checkmark$ & small & fast & tight (up to a basis-sifting factor)\\
 & spot-checking & $\checkmark$ & intermediate & intermediate & tight \\
\hline
\multirow{3}{*}{EAT} & local basis choice & $\times$ & - & - & -\\
 & \multirow{2}{*}{spot-checking} & \multirow{2}{*}{$\checkmark$} & large (original) & slow & tight \\
  & & & intermediate (modified) & intermediate & tight \\
\hline
\end{tabular}
\label{tab:CompareSummary}
\end{table}

We remark that the complementarity approach allows for more flexible protocol designs. Using the phase-error-correction security statements, one can sift the key generation rounds and analyse the key privacy therein in an easier and tighter manner. A direct consequence is that we can apply protocols other than the spot-checking protocol with the complementarity analysis. For instance, in the standard protocol in Box~\ref{box:ProtocolDetail}, the users set their measurements settings locally at random. Such flexibility brings more convenience in experiments. If one is interested in generating secret keys quickly, or, a `one-bit key generation' is enough, then applying the complementarity approach to the standard protocol is a good choice. Differently, if one hopes to have a large key expansion, that is, a larger value of $k/n$ with a reasonably large data size, then the spot-checking protocol would do the job, as it has an advantage in the basis sifting.

At the end of this section, we would like to mention some recent progress along the EAT-based line of research. First, it may be plausible to apply the EAT method to more general protocols. For instance, in Ref.~\cite{bhavsar2021calculation}, the authors apply the EAT method to analyse a device-independent randomness expansion protocol, where the nonlocal users set their measurement settings locally at random. This work also exhibits a detailed analysis tailored to the CHSH Bell test. Notwithstanding, one might need to establish a delicate Markov chain from the protocol to use EAT. There are also advanced EAT works that take advantage of the second-order statistics~\cite{dupuis2019entropy,liu2021device}. The recently notable experimental work of full DIQKD is based on such advanced techniques~\cite{nadlinger2022experimental}.
In these works, the additional information of variance between probability distributions is utilised, hence bringing a much tighter finite-size performance. Last but not least, shortly after we posted an earlier version of our work on the preprint platform, we become aware of the result of generalised entropy accumulation~\cite{metger2022generalised}. This work has a similar flavour as EAT, where the final result is to lower-bound the smooth min-entropy by the linear accumulation of conditional von Neumann entropy as the major term. The difference is the basis for entropy accumulation, where the Markov chain condition is relaxed to a relatively general non-signalling condition between the parties. While this technique can be applied to the device-independent setting, its finite-size performance remains to be checked rigorously. A more thorough discussion of these advanced entropy-accumulation-type results and a comparison of our method with them would be both theoretically and practically appealing, yet it is beyond the scope of this work. We would like to leave the study for future research.

\subsection{Key performance with advantage key distillation methods}\label{Sec:AdvNum}
In this section, we explain the details of the numerical results with advantage key distillation methods presented in the main text. All the discussions in this section are restricted to the asymptotic limit of infinite data size. In the figures, for simplicity, we denote the protocol in Box~\ref{box:ProtocolDetail} as (S), the protocol in Box~\ref{box:Vacancy} as (L), the protocol in Box~\ref{box:NoisyProcess} as (N), and the protocol in Box~\ref{box:B-step} as (B). In addition, we also consider a combined protocol utilising both the loss information and a noisy pre-processing step. That is, the sender flips his key bits randomly, and the receiver utilises the loss information on her side for information reconciliation. We denote this combined protocol as (C).

For the protocol utilising a noisy pre-processing step (Box~\ref{box:NoisyProcess}), the strength of noise that is deliberately added, $q=\Pr{u_i={b_i}^c|b_i}$, could be optimised. In Figure \ref{Fig:OptNoise} we depict the optimal noise level with respect to total transmittance. In this simulation, we set the state preparation fidelity to be unity. It can be seen that with a high transmittance, the optimal noise level becomes smaller. For the region where the standard protocol fails to generate secure keys, a high level of added noise could help.

\begin{figure}[hbt!]
\centering
    \resizebox{8cm}{!}{\includegraphics{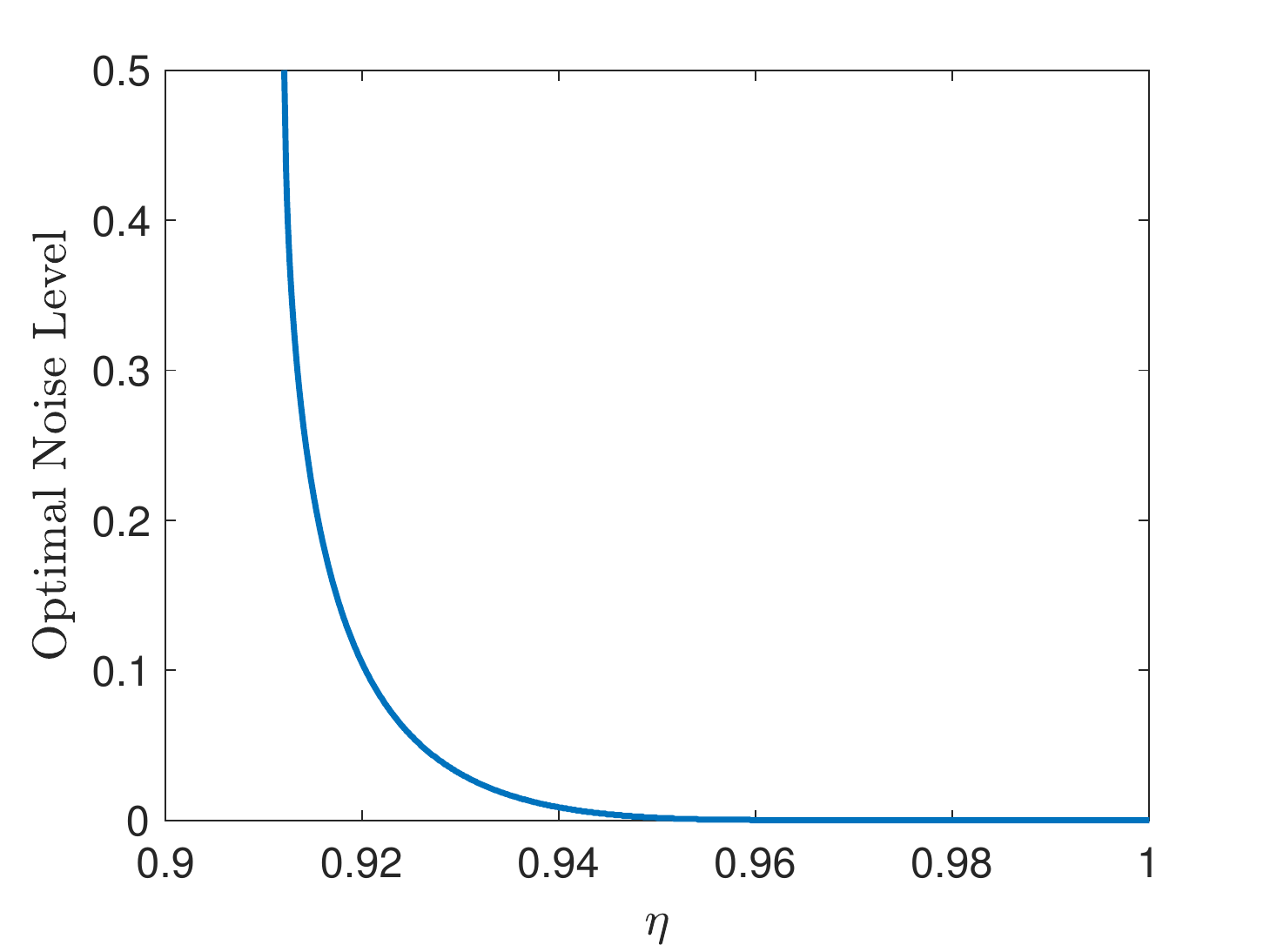}}
\caption{The optimal noise level with respect to total transmittance. In this simulation, we set $e_d=0$, corresponding to a state preparation fidelity of unity.}
\label{Fig:OptNoise}
\end{figure}

In Figure \ref{Fig:AdvThresh}, we depict the key rates $k/m$ when utilising different advantage key distillation methods with respect to transmittance $\eta$ and the channel depolarising factor $e_d$, respectively. When analysing the relation between key rates and transmittance, we set $e_d=0$, and similarly for the analysis of the key rates with respect to the depolarising factor. The protocol utilising the loss information does not change the key generation performance in the presence of channel noise. In Figure \ref{Fig:AdvThresh}(b) we only depict the key rates for the protocols in Box~\ref{box:ProtocolDetail} (S) and Box~\ref{box:NoisyProcess} (N). In the analysis for the channel noise level, the performance of (L) is the same as (S), and (C) is the same as (N). The points where the curves intersect with the horizontal axes correspond to the threshold transmittance and fidelities listed in Table~\ref{tab:ThresholdEff} of the main text, respectively. The correspondence between the depolarising factor and fidelity is given in Eq.~\eqref{Eq:Fidelity}.

\begin{figure*}[hbt!]
\hfil
\vbox{
\hsize 0.5\textwidth
\hbox{
{\begin{overpic}[width=\hsize]{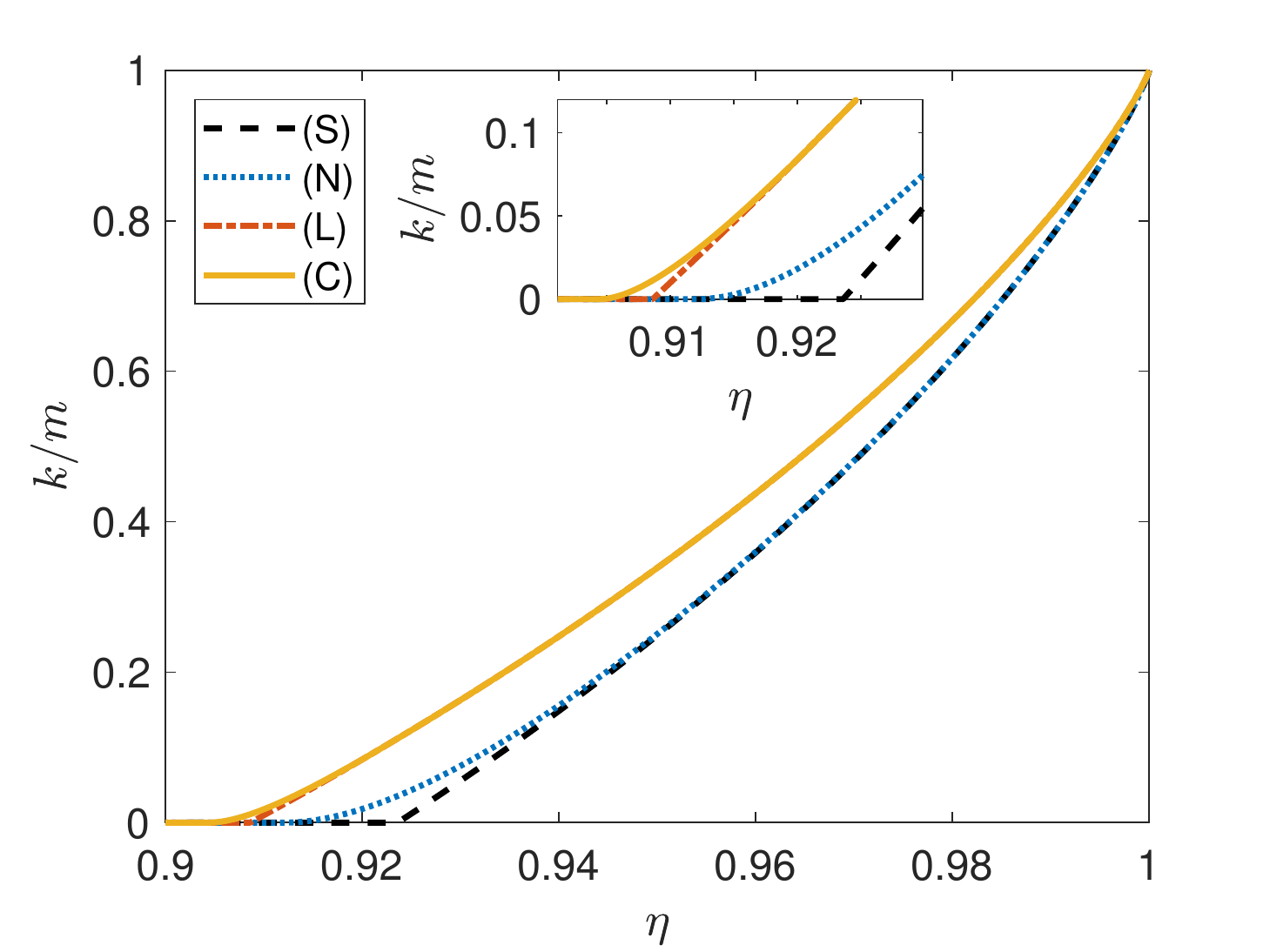}
\put(49,-4)
{\fontsize{10}{10}\selectfont (a)}
\end{overpic}}
\hbox{
{\begin{overpic}[width=\hsize]{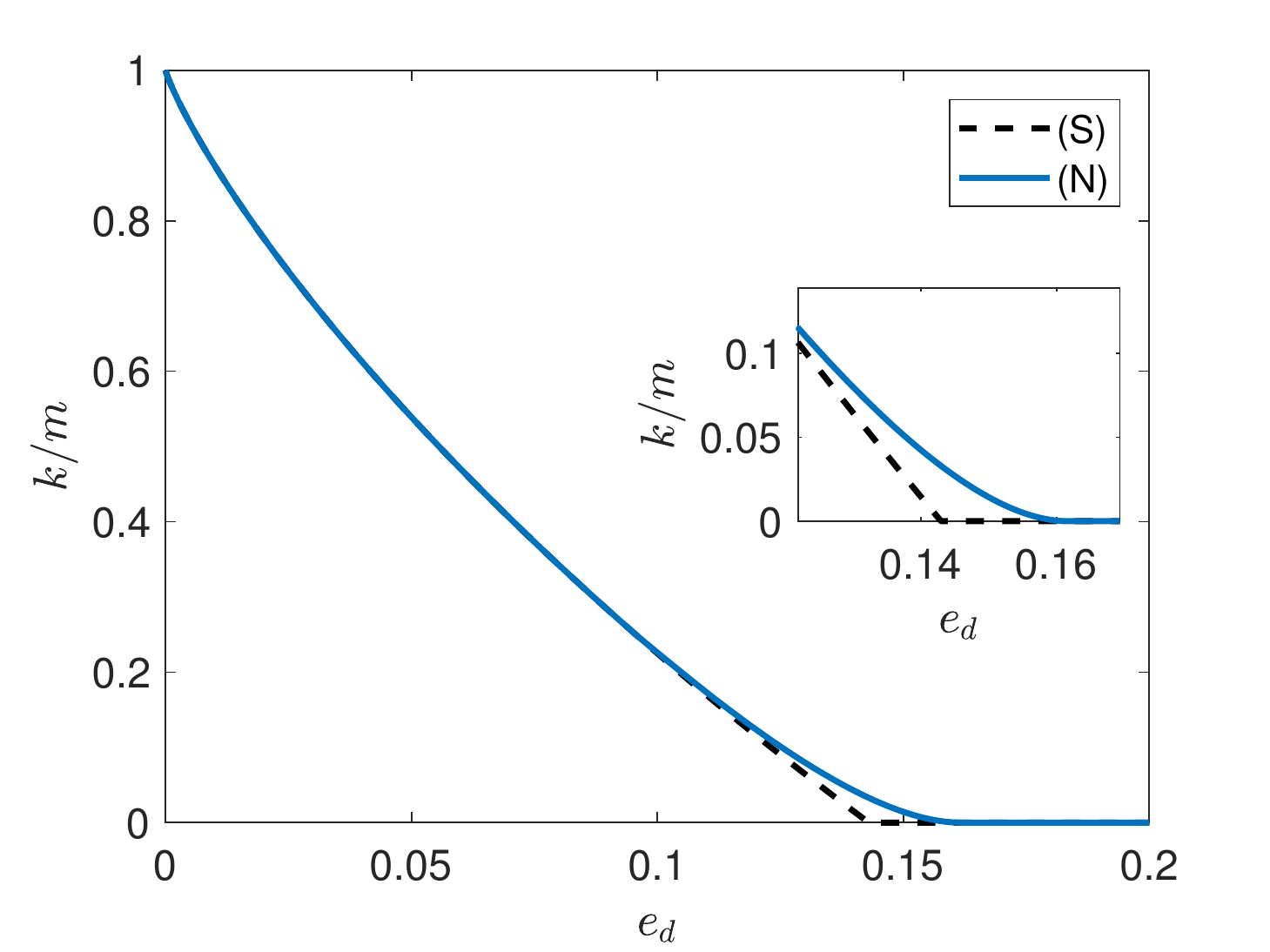}
\put(49,-4)
{\fontsize{10}{10}\selectfont (b)}
\end{overpic}}
}
}
}
\hfil
\caption{DIQKD key rates when using advantage key distillation methods. As a reference, we also depict the key rate of the standard protocol. (a) The key rates under different total transmittance. (b) The key rates under different values of the channel depolarising factor. We zoom in the region where the key rates approach to zero.}
\label{Fig:AdvThresh}
\end{figure*}

For the simulation of the protocol in Box~\ref{box:B-step}, we borrow the protocol design from~\cite{tan2020advantage}. In the cited work, the authors show the possibility of applying a two-way information reconciliation in DIQKD via the use of repetition codes under the condition of collective attacks. In the study of two-way information reconciliation in common device-dependent QKD protocols, the use of repetition code gives similar performance as the B-step method. Hence, we apply the same setting as Setting (vi) in Table I,~\cite{tan2020advantage} here for the analysis of our B-step-based DIQKD protocol. The experiment for this protocol is designed with the following entangled state between Alice and Bob,
\begin{equation}\label{Eq:BstepState}
\ket{\psi}=\cos{\Omega}\ket{00}+\sin{\Omega}\ket{11},\Omega=0.6224,
\end{equation}
and the observables Alice and Bob measure are given by
\begin{equation}\label{Eq:BstepObservable}
\begin{split}
\hat{A}_i&=\ketbra{\theta_{A_i^+}}-\ketbra{\theta_{A_i^-}},\ket{\theta_{A_i^+}}=\cos{\frac{\theta_{A_i}}{2}}\ket{0}+\sin{\frac{\theta_{A_i}}{2}}\ket{1},i=0,1,\\
\theta_{A_0}&=-0.35923,\theta_{A_1}=1.1538,\\
\hat{B}_i&=\ketbra{\theta_{B_i^+}}-\ketbra{\theta_{B_i^-}},\ket{\theta_{B_i^+}}=\cos{\frac{\theta_{B_i}}{2}}\ket{0}+\sin{\frac{\theta_{B_i}}{2}}\ket{1},i=0,1,\\
\theta_{B_0}&=0.35923,\theta_{B_1}=-1.1538.\\
\end{split}
\end{equation}

In Figure \ref{Fig:BSteps}, we depict the key rates adopting B steps for one and two times under different transmittance. In comparison, we depict the key rates for the experimental set-up designed according Eq.~\eqref{Eq:BstepState}\eqref{Eq:BstepObservable} without using B-steps. In the simulation, we set the state preparation fidelity to be unity. Since the experimental configuration has a systematic quantum bit-error rate, even if there is no loss in transmittance, the key rate is lower than one bit per round. Nevertheless, after applying B steps, the threshold transmittance for a positive key rate could be lowered down. In Table~\ref{tab:BstepThreshold}, we list the threshold total transmittance and the tolerable values of the depolarising factors after applying B step. We use MATLAB R2019b for Our numerical calculations. Due to the numerical precision, we only list the simulation data of applying at most 5 B steps. The advantage by applying more B steps is limited and the numerical values suffer from a poor numerical precision, hence we report the threshold efficiency as $88.30\%$ and the threshold fidelity as $88.28\%$ in the main text. The threshold fidelity is calculated by Eq.~\eqref{Eq:Fidelity}, where $F_t=1-0.75\times0.1562\approx88.28\%$.

\begin{figure}[hbt!]
\centering \resizebox{8cm}{!}{\includegraphics{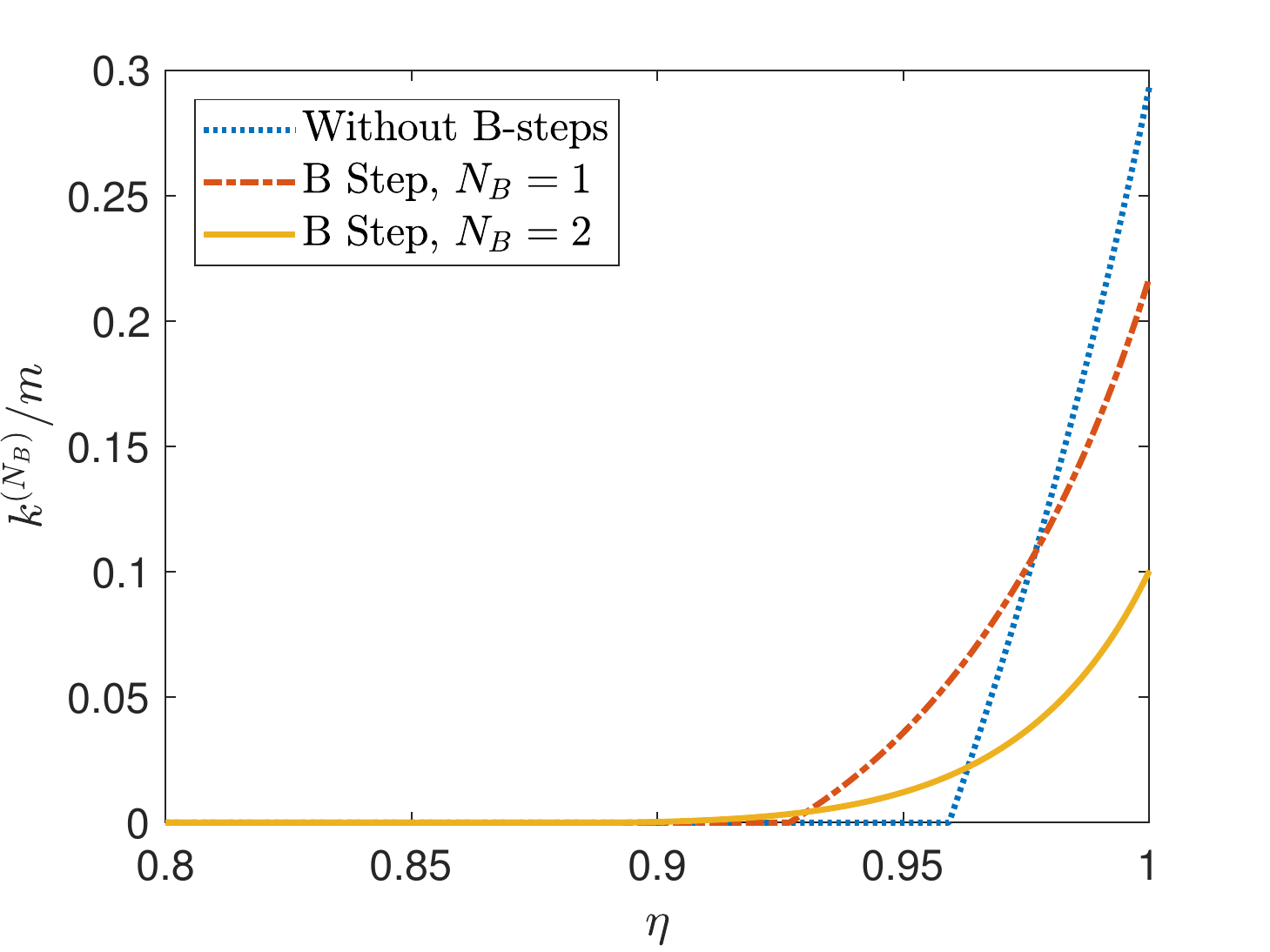}}
\caption{The key rates after applying B steps. The amount of final key bits that are obtained after $N_B$ B steps is denoted as $k^{N_B}$. In comparison, we depict the key rate without using B steps.}
\label{Fig:BSteps}
\end{figure}

\begin{table*}[htbp!]
\centering \caption{Threshold transmittance and tolerable depolarising factors after applying B steps. }
\begin{tabular}{c|ccccc}
\hline
$\#$ B Steps & $1$ & $2$ & $3$ & $4$ & $5$\\
\hline
Threshold Transmittance & $93.54\%$ & $91.22\%$ & $89.70\%$ & $88.81\%$ & $88.30\%$\\
Tolerable Depolarising Factor & $0.0901$ & $0.1197$ & $0.1389$ & $0.1499$ & $0.1562$\\
\hline
\end{tabular}
\label{tab:BstepThreshold}
\end{table*}


We also make some remarks on the cited experimental data in the main text~\cite{hensen2015loophole,liu2021device,li2021experimental}. In these works, both the (effective) state preparation fidelity and detection efficiencies are reported. The reported state preparation fidelities are given by a state tomography. We note that none of the state tomography results actual obeys the depolarising channel model we use for numerical simulation, as described in Sec.~\ref{Supp:DetailNum}. Nevertheless, for a state preparation of high fidelity, these values can be taken as references. The two photonic experiments report the highest detection efficiencies so far. The target of the two experiments is to realise a device-independent quantum randomness expansion. In~\cite{liu2021device}, the locality loophole is not closed with a space-like separation. Instead, it is assumed that both Alice and Bob hold a secure lab preventing any information leakage. We note that this assumption is reasonable in DIQKD. The experiment~\cite{li2021experimental} provides the best loophole-free Bell test data, where the detection loophole, the free-will loophole, and the locality loophole are all closed.

\section{Further Discussions}


\subsection{Security anlysis for parallel DIQKD}
In our security analysis, we consider sequential DIQKD, where the valid users perform quantum measurements one time after another in a definite sequential order. That is, the users shall not start the $(i+1)$th round until they have finished the $i$th round. The measurement results might depend on the previous events, but not the future ones. Another setting is parallel DIQKD, where the users do not wait for the outputs from previous rounds when they input later ones. Then, in an `$n$-round' parallel setting, we can even assume that the users set their own measurement bases with $n$ inputs simultaneously. Afterwards, the measurement devices output all the results together. In this case, there is no definite time order between the inputs and outputs from different rounds. As a result, the model used in Section \ref{Supp:SequentialSetting} is not applicable to the parallel DIQKD setting.

Recently, the security of parallel DIQKD has been shown with the magic-square Bell test \cite{jain2020parallel}. The main drawback is that the loss tolerance is extremely poor, so that an experimental demonstration is practically infeasible.  A natural question is whether we can apply the security technique developed in our work based on complementarity to parallel DIQKD. It might be a good starting point to develop new methods to break down the parallel multiple-round correlations into the analysis of single rounds.

We should emphasise that the sequential setting considered in our work is not more restrictive for the users than the parallel setting. On the one hand, in the parallel setting, the devices can execute joint operations over multiple rounds, which is not allowed in the sequential setting. On the other hand, in the sequential setting shown in Section \ref{Supp:SequentialSetting}, we allow the untrusted devices to communicate with each other in between two adjacent rounds of quantum measurements, which is forbidden in parallel DIQKD.


\subsection{Memory attack and covert channels}
The existing device-independent quantum information processing protocols suffer from the problem of memory attacks~\cite{barrett2013memory}. Let us consider that a DIQKD system is used for multiple times. If we focus on one  session, our security analysis works well. However, if we look at all the sessions together, since the legitimate users cannot monitor the inner operation status of the untrusted devices, the devices could store key information and leak it in the necessary classical communication procedures in later sessions.

One na\"ive countermeasure is to destroy the devices after one session and buy new ones. This is practically infeasible, though. For practical security, we could borrow ideas from secure computation in classical cryptographic engineering~\cite{curty2019foiling}. Nevertheless, one might consider that such methods are contradictory to the tenet of device-independent security. It remains open whether there is a complete solution to memory attacks.

Another issue in the security assumption is that except for the classical communication of post-processing, the devices do not signal to the outside. In the experimental realisation of DIQKD and other device-independent information processing tasks, either implicitly or explicitly, the assumption is taken as guaranteed by some physical shielding. One might argue whether such an assumption is indeed reasonable~\cite{xu2020secure}, as an adversary could employ covert channels in some completely unimaginable ways.

Note that in a different QKD protocol, measurement-device-independent quantum key distribution (MDIQKD)~\cite{lo2012measurement,braunstein2012side}, the problem of memory attacks and covert channels could be avoided naturally on the measurement site. At first glance, MDIQKD requires the valid users to have their quantum information sources fully characterised and thus raises an additional security assumption on the quantum sources. However, different from DIQKD, the untrusted measurement device is assumed to be completely possessed and controlled by Eve.

\subsection{Use of other Bell inequalities}
In this work, we employ the CHSH Bell inequality. Generally, we could use other Bell inequalities in DIQKD. The techniques developed in this work could directly be applied to protocols utilising the family of Bell inequalities with binary inputs and outputs for each party. Some Bell inequalities in this family can help ease the experimental realisation of loophole-free Bell tests, such as lessening the tolerable transmittance and state preparation fidelity~\cite{eberhard1993background}. We also expect similar phenomena in DIQKD. Due to the symmetry features of these inequalities, it may be challenging to obtain an elegant analytic result. We might need to utilise some numerical algorithms for analysis. There are also Bell inequalities with larger input and output alphabets~\cite{brunner2014bell}. For these Bell inequalities, unfortunately, Jordan's lemma cannot be directly applies. It would be interesting to see if a similar effective dimension reduction result exists for these Bell inequalities. The research in characterising the boundary of nonlocal correlations might help~\cite{navascues2007bounding}.

\subsection{Other device-independent quantum information processing tasks}
Our complementarity-based security analysis could apply to other device-independent quantum information processing tasks, such as device-independent quantum randomness generation (DIQRNG) and device-independent entanglement quantification. We note that, unlike DIQKD, the listed two tasks involve only one user. In DIQRNG, the user hopes to obtain private quantum random numbers in a device-independent way~\cite{Colbeck09}. The usual randomness measure is the (smooth) min-entropy~\cite{renner2008security}, while it is closely related to the phase-error probability in the complementarity approach~\cite{tsurumaru2021equivalence}. Therefore, we could say that our results in this work already provides a scheme for DIQRNG analysis. Since our security analysis is in the spirit of quantum-error-correction-based entanglement distillation~\cite{Bennett96Mixed}, it is almost directly applicable in the quantification of distillable entanglement in a device-independent manner. In addition, note that in quantifying the amount of distillable entanglement in a system, one needs only to prove the existence of an entanglement distillation protocol instead of carrying it out in reality. Consequently, we might witness distillable entanglement at regions that do not support key generation in DIQKD.

\subsection{Inspiration from recent progress in entropic methods}
In our work, we make some preliminary comparisons between our method and EAT, a property of min- and max-entropy~\cite{arnon2018practical}. In comparing the performance of our method and the EAT method, we have taken the original EAT results in Refs.~\cite{arnon2018practical,arnon2019simple}. Since EAT was proposed, there has been a series of optimisation works aiming at a better finite-size performance. We have listed a series of notable results along this line of research in Section~\ref{Sec:EATSummary}. Another notable line of research is the R\'enyi-divergence-based analysis. Actually, the first entropy-accumulation-type result in the general non-i.i.d. sequential setting is for the R\'enyi divergence~\cite{miller2016robust}. As one can bound the min-entropy by using the R\'enyi divergence, this result provides one of the first full device-independent security analyses, yet the key-rate formula is not tight. Later results tighten the bound and bring better finite-size performance~\cite{knill2018quantum}. At a high level, all these security analyses provide bounds on the conditional smooth min-entropy, but the approaches are different. It would be interesting to study the links and differences among the variant methods. Moreover, these entropic methods may also inspire further study along the complementarity-based approach.

\end{appendices}

\bibliographystyle{apsrev}
\bibliography{bibDIPh}

\begin{thebibliography}{72}
\expandafter\ifx\csname natexlab\endcsname\relax\def\natexlab#1{#1}\fi
\expandafter\ifx\csname bibnamefont\endcsname\relax
  \def\bibnamefont#1{#1}\fi
\expandafter\ifx\csname bibfnamefont\endcsname\relax
  \def\bibfnamefont#1{#1}\fi
\expandafter\ifx\csname citenamefont\endcsname\relax
  \def\citenamefont#1{#1}\fi
\expandafter\ifx\csname url\endcsname\relax
  \def\url#1{\texttt{#1}}\fi
\expandafter\ifx\csname urlprefix\endcsname\relax\def\urlprefix{URL }\fi
\providecommand{\bibinfo}[2]{#2}
\providecommand{\eprint}[2][]{\url{#2}}

\bibitem[{\citenamefont{Heisenberg}(1927)}]{heisenberg1927anschaulichen}
\bibinfo{author}{\bibfnamefont{W.}~\bibnamefont{Heisenberg}},
  \bibinfo{journal}{Z. Phys.} \textbf{\bibinfo{volume}{43}},
  \bibinfo{pages}{172} (\bibinfo{year}{1927}),
  \urlprefix\url{https://doi.org/10.1007/BF01397280}.

\bibitem[{\citenamefont{Einstein et~al.}(1935)\citenamefont{Einstein, Podolsky,
  and Rosen}}]{einstein1935can}
\bibinfo{author}{\bibfnamefont{A.}~\bibnamefont{Einstein}},
  \bibinfo{author}{\bibfnamefont{B.}~\bibnamefont{Podolsky}}, \bibnamefont{and}
  \bibinfo{author}{\bibfnamefont{N.}~\bibnamefont{Rosen}},
  \bibinfo{journal}{Phys. Rev.} \textbf{\bibinfo{volume}{47}},
  \bibinfo{pages}{777} (\bibinfo{year}{1935}),
  \urlprefix\url{https://link.aps.org/doi/10.1103/PhysRev.47.777}.

\bibitem[{\citenamefont{Bell}(1964)}]{bell1964einstein}
\bibinfo{author}{\bibfnamefont{J.~S.} \bibnamefont{Bell}},
  \bibinfo{journal}{Physics Physique Fizika} \textbf{\bibinfo{volume}{1}},
  \bibinfo{pages}{195} (\bibinfo{year}{1964}).

\bibitem[{\citenamefont{Clauser et~al.}(1969)\citenamefont{Clauser, Horne,
  Shimony, and Holt}}]{clauser1969proposed}
\bibinfo{author}{\bibfnamefont{J.~F.} \bibnamefont{Clauser}},
  \bibinfo{author}{\bibfnamefont{M.~A.} \bibnamefont{Horne}},
  \bibinfo{author}{\bibfnamefont{A.}~\bibnamefont{Shimony}}, \bibnamefont{and}
  \bibinfo{author}{\bibfnamefont{R.~A.} \bibnamefont{Holt}},
  \bibinfo{journal}{Phys. Rev. Lett.} \textbf{\bibinfo{volume}{23}},
  \bibinfo{pages}{880} (\bibinfo{year}{1969}),
  \urlprefix\url{https://link.aps.org/doi/10.1103/PhysRevLett.23.880}.

\bibitem[{\citenamefont{Mayers and Yao}(1998)}]{mayers1998quantum}
\bibinfo{author}{\bibfnamefont{D.}~\bibnamefont{Mayers}} \bibnamefont{and}
  \bibinfo{author}{\bibfnamefont{A.}~\bibnamefont{Yao}}, in
  \emph{\bibinfo{booktitle}{Proceedings of the 39th Annual Symposium on
  Foundations of Computer Science}} (\bibinfo{publisher}{IEEE Computer
  Society}, \bibinfo{address}{Washington, DC, USA}, \bibinfo{year}{1998}), FOCS
  '98, pp. \bibinfo{pages}{503--509}, ISBN \bibinfo{isbn}{0-8186-9172-7},
  \urlprefix\url{http://dl.acm.org/citation.cfm?id=795664.796390}.

\bibitem[{\citenamefont{Ac\'{i}n et~al.}(2007)\citenamefont{Ac\'{i}n, Brunner,
  Gisin, Massar, Pironio, and Scarani}}]{acin2007device}
\bibinfo{author}{\bibfnamefont{A.}~\bibnamefont{Ac\'{i}n}},
  \bibinfo{author}{\bibfnamefont{N.}~\bibnamefont{Brunner}},
  \bibinfo{author}{\bibfnamefont{N.}~\bibnamefont{Gisin}},
  \bibinfo{author}{\bibfnamefont{S.}~\bibnamefont{Massar}},
  \bibinfo{author}{\bibfnamefont{S.}~\bibnamefont{Pironio}}, \bibnamefont{and}
  \bibinfo{author}{\bibfnamefont{V.}~\bibnamefont{Scarani}},
  \bibinfo{journal}{Phys. Rev. Lett.} \textbf{\bibinfo{volume}{98}},
  \bibinfo{pages}{230501} (\bibinfo{year}{2007}),
  \urlprefix\url{https://link.aps.org/doi/10.1103/PhysRevLett.98.230501}.

\bibitem[{\citenamefont{Bennett and Brassard}(1984)}]{bennett1984quantum}
\bibinfo{author}{\bibfnamefont{C.~H.} \bibnamefont{Bennett}} \bibnamefont{and}
  \bibinfo{author}{\bibfnamefont{G.}~\bibnamefont{Brassard}}, in
  \emph{\bibinfo{booktitle}{Proceedings of IEEE International Conference on
  Computers, Systems and Signal Processing}} (\bibinfo{organization}{Bangalore,
  India}, \bibinfo{year}{1984}), pp. \bibinfo{pages}{175--179}.

\bibitem[{\citenamefont{Ekert}(1991)}]{ekert1991quantum}
\bibinfo{author}{\bibfnamefont{A.~K.} \bibnamefont{Ekert}},
  \bibinfo{journal}{Phys. Rev. Lett.} \textbf{\bibinfo{volume}{67}},
  \bibinfo{pages}{661} (\bibinfo{year}{1991}),
  \urlprefix\url{https://link.aps.org/doi/10.1103/PhysRevLett.67.661}.

\bibitem[{\citenamefont{Lo and Chau}(1999)}]{lo1999unconditional}
\bibinfo{author}{\bibfnamefont{H.-K.} \bibnamefont{Lo}} \bibnamefont{and}
  \bibinfo{author}{\bibfnamefont{H.~F.} \bibnamefont{Chau}},
  \bibinfo{journal}{Science} \textbf{\bibinfo{volume}{283}},
  \bibinfo{pages}{2050} (\bibinfo{year}{1999}).

\bibitem[{\citenamefont{Shor and Preskill}(2000)}]{shor2000simple}
\bibinfo{author}{\bibfnamefont{P.~W.} \bibnamefont{Shor}} \bibnamefont{and}
  \bibinfo{author}{\bibfnamefont{J.}~\bibnamefont{Preskill}},
  \bibinfo{journal}{Phys. Rev. Lett.} \textbf{\bibinfo{volume}{85}},
  \bibinfo{pages}{441} (\bibinfo{year}{2000}),
  \urlprefix\url{https://link.aps.org/doi/10.1103/PhysRevLett.85.441}.

\bibitem[{\citenamefont{Koashi}(2009)}]{koashi2009simple}
\bibinfo{author}{\bibfnamefont{M.}~\bibnamefont{Koashi}}, \bibinfo{journal}{New
  J. Phys.} \textbf{\bibinfo{volume}{11}}, \bibinfo{pages}{045018}
  (\bibinfo{year}{2009}),
  \urlprefix\url{http://stacks.iop.org/1367-2630/11/i=4/a=045018}.

\bibitem[{\citenamefont{Xu et~al.}(2020)\citenamefont{Xu, Ma, Zhang, Lo, and
  Pan}}]{xu2020secure}
\bibinfo{author}{\bibfnamefont{F.}~\bibnamefont{Xu}},
  \bibinfo{author}{\bibfnamefont{X.}~\bibnamefont{Ma}},
  \bibinfo{author}{\bibfnamefont{Q.}~\bibnamefont{Zhang}},
  \bibinfo{author}{\bibfnamefont{H.-K.} \bibnamefont{Lo}}, \bibnamefont{and}
  \bibinfo{author}{\bibfnamefont{J.-W.} \bibnamefont{Pan}},
  \bibinfo{journal}{Rev. Mod. Phys.} \textbf{\bibinfo{volume}{92}},
  \bibinfo{pages}{025002} (\bibinfo{year}{2020}),
  \urlprefix\url{https://link.aps.org/doi/10.1103/RevModPhys.92.025002}.

\bibitem[{Note1()}]{Note1}
Note1, \bibinfo{note}{note, however, that DIQKD is still vulnerable to memory
  attacks and covert channels. A memory attack may occur when the devices are
  reused for multiple QKD sessions, where the untrusted devices may store the
  key information in one DIQKD session and leak it via the necessary public
  communication required by the protocol in new sessions~\cite
  {barrett2013memory}. A covert channel signals the key information to the
  outside in an unnoticed way other than the public communication allowed by
  the protocol. We do not discuss these issues in this work. Particularly,
  there is no memory attack if we only consider one QKD session. For
  discussions on practical countermeasures, one may refer to Ref.~\cite
  {curty2019foiling}. Note that in a different QKD protocol,
  measurement-device-independent quantum key distribution (MDIQKD)~\cite
  {lo2012measurement,braunstein2012side}, the problem of memory attacks and
  covert channels can be naturally avoided on the measurement site.}

\bibitem[{\citenamefont{Vazirani and Vidick}(2014)}]{vazirani2014fully}
\bibinfo{author}{\bibfnamefont{U.}~\bibnamefont{Vazirani}} \bibnamefont{and}
  \bibinfo{author}{\bibfnamefont{T.}~\bibnamefont{Vidick}},
  \bibinfo{journal}{Phys. Rev. Lett.} \textbf{\bibinfo{volume}{113}},
  \bibinfo{pages}{140501} (\bibinfo{year}{2014}),
  \urlprefix\url{https://link.aps.org/doi/10.1103/PhysRevLett.113.140501}.

\bibitem[{\citenamefont{Miller and Shi}(2016)}]{miller2016robust}
\bibinfo{author}{\bibfnamefont{C.~A.} \bibnamefont{Miller}} \bibnamefont{and}
  \bibinfo{author}{\bibfnamefont{Y.}~\bibnamefont{Shi}}, \bibinfo{journal}{J.
  ACM} \textbf{\bibinfo{volume}{63}}, \bibinfo{pages}{33:1}
  (\bibinfo{year}{2016}), ISSN \bibinfo{issn}{0004-5411},
  \urlprefix\url{http://doi.acm.org/10.1145/2885493}.

\bibitem[{\citenamefont{Arnon-Friedman
  et~al.}(2018)\citenamefont{Arnon-Friedman, Dupuis, Fawzi, Renner, and
  Vidick}}]{arnon2018practical}
\bibinfo{author}{\bibfnamefont{R.}~\bibnamefont{Arnon-Friedman}},
  \bibinfo{author}{\bibfnamefont{F.}~\bibnamefont{Dupuis}},
  \bibinfo{author}{\bibfnamefont{O.}~\bibnamefont{Fawzi}},
  \bibinfo{author}{\bibfnamefont{R.}~\bibnamefont{Renner}}, \bibnamefont{and}
  \bibinfo{author}{\bibfnamefont{T.}~\bibnamefont{Vidick}},
  \bibinfo{journal}{Nat. Commn.} \textbf{\bibinfo{volume}{9}},
  \bibinfo{pages}{459} (\bibinfo{year}{2018}), ISSN \bibinfo{issn}{2041-1723},
  \urlprefix\url{https://doi.org/10.1038/s41467-017-02307-4}.

\bibitem[{\citenamefont{Zhang et~al.}(2020)\citenamefont{Zhang, Fu, and
  Knill}}]{knill2018quantum}
\bibinfo{author}{\bibfnamefont{Y.}~\bibnamefont{Zhang}},
  \bibinfo{author}{\bibfnamefont{H.}~\bibnamefont{Fu}}, \bibnamefont{and}
  \bibinfo{author}{\bibfnamefont{E.}~\bibnamefont{Knill}},
  \bibinfo{journal}{Phys. Rev. Research} \textbf{\bibinfo{volume}{2}},
  \bibinfo{pages}{013016} (\bibinfo{year}{2020}),
  \urlprefix\url{https://link.aps.org/doi/10.1103/PhysRevResearch.2.013016}.

\bibitem[{\citenamefont{Nadlinger et~al.}(2022)\citenamefont{Nadlinger, Drmota,
  Nichol, Araneda, Main, Srinivas, Lucas, Ballance, Ivanov, Tan
  et~al.}}]{nadlinger2022experimental}
\bibinfo{author}{\bibfnamefont{D.}~\bibnamefont{Nadlinger}},
  \bibinfo{author}{\bibfnamefont{P.}~\bibnamefont{Drmota}},
  \bibinfo{author}{\bibfnamefont{B.}~\bibnamefont{Nichol}},
  \bibinfo{author}{\bibfnamefont{G.}~\bibnamefont{Araneda}},
  \bibinfo{author}{\bibfnamefont{D.}~\bibnamefont{Main}},
  \bibinfo{author}{\bibfnamefont{R.}~\bibnamefont{Srinivas}},
  \bibinfo{author}{\bibfnamefont{D.}~\bibnamefont{Lucas}},
  \bibinfo{author}{\bibfnamefont{C.}~\bibnamefont{Ballance}},
  \bibinfo{author}{\bibfnamefont{K.}~\bibnamefont{Ivanov}},
  \bibinfo{author}{\bibfnamefont{E.-Z.} \bibnamefont{Tan}},
  \bibnamefont{et~al.}, \bibinfo{journal}{Nature}
  \textbf{\bibinfo{volume}{607}}, \bibinfo{pages}{682} (\bibinfo{year}{2022}).

\bibitem[{\citenamefont{Liu et~al.}(2022)\citenamefont{Liu, Zhang, Zhen, Li,
  Liu, Fan, Xu, Zhang, and Pan}}]{liu2022toward}
\bibinfo{author}{\bibfnamefont{W.-Z.} \bibnamefont{Liu}},
  \bibinfo{author}{\bibfnamefont{Y.-Z.} \bibnamefont{Zhang}},
  \bibinfo{author}{\bibfnamefont{Y.-Z.} \bibnamefont{Zhen}},
  \bibinfo{author}{\bibfnamefont{M.-H.} \bibnamefont{Li}},
  \bibinfo{author}{\bibfnamefont{Y.}~\bibnamefont{Liu}},
  \bibinfo{author}{\bibfnamefont{J.}~\bibnamefont{Fan}},
  \bibinfo{author}{\bibfnamefont{F.}~\bibnamefont{Xu}},
  \bibinfo{author}{\bibfnamefont{Q.}~\bibnamefont{Zhang}}, \bibnamefont{and}
  \bibinfo{author}{\bibfnamefont{J.-W.} \bibnamefont{Pan}},
  \bibinfo{journal}{Phys. Rev. Lett.} \textbf{\bibinfo{volume}{129}},
  \bibinfo{pages}{050502} (\bibinfo{year}{2022}),
  \urlprefix\url{https://link.aps.org/doi/10.1103/PhysRevLett.129.050502}.

\bibitem[{\citenamefont{Zhang et~al.}(2022)\citenamefont{Zhang, van Leent,
  Redeker, Garthoff, Schwonnek, Fertig, Eppelt, Rosenfeld, Scarani, Lim
  et~al.}}]{zhang2022device}
\bibinfo{author}{\bibfnamefont{W.}~\bibnamefont{Zhang}},
  \bibinfo{author}{\bibfnamefont{T.}~\bibnamefont{van Leent}},
  \bibinfo{author}{\bibfnamefont{K.}~\bibnamefont{Redeker}},
  \bibinfo{author}{\bibfnamefont{R.}~\bibnamefont{Garthoff}},
  \bibinfo{author}{\bibfnamefont{R.}~\bibnamefont{Schwonnek}},
  \bibinfo{author}{\bibfnamefont{F.}~\bibnamefont{Fertig}},
  \bibinfo{author}{\bibfnamefont{S.}~\bibnamefont{Eppelt}},
  \bibinfo{author}{\bibfnamefont{W.}~\bibnamefont{Rosenfeld}},
  \bibinfo{author}{\bibfnamefont{V.}~\bibnamefont{Scarani}},
  \bibinfo{author}{\bibfnamefont{C.~C.-W.} \bibnamefont{Lim}},
  \bibnamefont{et~al.}, \bibinfo{journal}{Nature}
  \textbf{\bibinfo{volume}{607}}, \bibinfo{pages}{687} (\bibinfo{year}{2022}).

\bibitem[{\citenamefont{Tsurumaru and Ichikawa}(2016)}]{tsurumaru2016multi}
\bibinfo{author}{\bibfnamefont{T.}~\bibnamefont{Tsurumaru}} \bibnamefont{and}
  \bibinfo{author}{\bibfnamefont{T.}~\bibnamefont{Ichikawa}},
  \bibinfo{journal}{New J. Phys.} \textbf{\bibinfo{volume}{18}},
  \bibinfo{pages}{103043} (\bibinfo{year}{2016}).

\bibitem[{\citenamefont{Woodhead et~al.}(2021)\citenamefont{Woodhead,
  Ac{\'\i}n, and Pironio}}]{woodhead2021device}
\bibinfo{author}{\bibfnamefont{E.}~\bibnamefont{Woodhead}},
  \bibinfo{author}{\bibfnamefont{A.}~\bibnamefont{Ac{\'\i}n}},
  \bibnamefont{and} \bibinfo{author}{\bibfnamefont{S.}~\bibnamefont{Pironio}},
  \bibinfo{journal}{Quantum} \textbf{\bibinfo{volume}{5}}, \bibinfo{pages}{443}
  (\bibinfo{year}{2021}).

\bibitem[{\citenamefont{Ac{\'\i}n et~al.}(2006)\citenamefont{Ac{\'\i}n, Massar,
  and Pironio}}]{acin2006efficient}
\bibinfo{author}{\bibfnamefont{A.}~\bibnamefont{Ac{\'\i}n}},
  \bibinfo{author}{\bibfnamefont{S.}~\bibnamefont{Massar}}, \bibnamefont{and}
  \bibinfo{author}{\bibfnamefont{S.}~\bibnamefont{Pironio}},
  \bibinfo{journal}{New J. Phys.} \textbf{\bibinfo{volume}{8}},
  \bibinfo{pages}{126} (\bibinfo{year}{2006}).

\bibitem[{\citenamefont{Pironio et~al.}(2009)\citenamefont{Pironio, Ac\'{i}n,
  Brunner, Gisin, Massar, and Scarani}}]{pironio2009device}
\bibinfo{author}{\bibfnamefont{S.}~\bibnamefont{Pironio}},
  \bibinfo{author}{\bibfnamefont{A.}~\bibnamefont{Ac\'{i}n}},
  \bibinfo{author}{\bibfnamefont{N.}~\bibnamefont{Brunner}},
  \bibinfo{author}{\bibfnamefont{N.}~\bibnamefont{Gisin}},
  \bibinfo{author}{\bibfnamefont{S.}~\bibnamefont{Massar}}, \bibnamefont{and}
  \bibinfo{author}{\bibfnamefont{V.}~\bibnamefont{Scarani}},
  \bibinfo{journal}{New J. Phys.} \textbf{\bibinfo{volume}{11}},
  \bibinfo{pages}{045021} (\bibinfo{year}{2009}),
  \urlprefix\url{http://stacks.iop.org/1367-2630/11/i=4/a=045021}.

\bibitem[{\citenamefont{Azuma}(1967)}]{azuma1967weighted}
\bibinfo{author}{\bibfnamefont{K.}~\bibnamefont{Azuma}},
  \bibinfo{journal}{Tohoku Math. J. Second Ser.} \textbf{\bibinfo{volume}{19}},
  \bibinfo{pages}{357} (\bibinfo{year}{1967}).

\bibitem[{\citenamefont{Kato}(2020)}]{kato2020concentration}
\bibinfo{author}{\bibfnamefont{G.}~\bibnamefont{Kato}},
  \bibinfo{journal}{arXiv:2002.04357}  (\bibinfo{year}{2020}).

\bibitem[{\citenamefont{Shannon}(1948)}]{shannon1948mathematical}
\bibinfo{author}{\bibfnamefont{C.~E.} \bibnamefont{Shannon}},
  \bibinfo{journal}{Bell Syst. Tech. J.} \textbf{\bibinfo{volume}{27}},
  \bibinfo{pages}{379} (\bibinfo{year}{1948}).

\bibitem[{\citenamefont{Li et~al.}(2021)\citenamefont{Li, Zhang, Liu, Zhao,
  Bai, Liu, Zhao, Peng, Zhang, Zhang et~al.}}]{li2021experimental}
\bibinfo{author}{\bibfnamefont{M.-H.} \bibnamefont{Li}},
  \bibinfo{author}{\bibfnamefont{X.}~\bibnamefont{Zhang}},
  \bibinfo{author}{\bibfnamefont{W.-Z.} \bibnamefont{Liu}},
  \bibinfo{author}{\bibfnamefont{S.-R.} \bibnamefont{Zhao}},
  \bibinfo{author}{\bibfnamefont{B.}~\bibnamefont{Bai}},
  \bibinfo{author}{\bibfnamefont{Y.}~\bibnamefont{Liu}},
  \bibinfo{author}{\bibfnamefont{Q.}~\bibnamefont{Zhao}},
  \bibinfo{author}{\bibfnamefont{Y.}~\bibnamefont{Peng}},
  \bibinfo{author}{\bibfnamefont{J.}~\bibnamefont{Zhang}},
  \bibinfo{author}{\bibfnamefont{Y.}~\bibnamefont{Zhang}},
  \bibnamefont{et~al.}, \bibinfo{journal}{Phys. Rev. Lett.}
  \textbf{\bibinfo{volume}{126}}, \bibinfo{pages}{050503}
  (\bibinfo{year}{2021}),
  \urlprefix\url{https://link.aps.org/doi/10.1103/PhysRevLett.126.050503}.

\bibitem[{\citenamefont{Murta et~al.}(2019)\citenamefont{Murta, van Dam,
  Ribeiro, Hanson, and Wehner}}]{murta2019towards}
\bibinfo{author}{\bibfnamefont{G.}~\bibnamefont{Murta}},
  \bibinfo{author}{\bibfnamefont{S.~B.} \bibnamefont{van Dam}},
  \bibinfo{author}{\bibfnamefont{J.}~\bibnamefont{Ribeiro}},
  \bibinfo{author}{\bibfnamefont{R.}~\bibnamefont{Hanson}}, \bibnamefont{and}
  \bibinfo{author}{\bibfnamefont{S.}~\bibnamefont{Wehner}},
  \bibinfo{journal}{Quantum Sci. Technol.} \textbf{\bibinfo{volume}{4}},
  \bibinfo{pages}{035011} (\bibinfo{year}{2019}).

\bibitem[{\citenamefont{Rosenfeld et~al.}(2017)\citenamefont{Rosenfeld,
  Burchardt, Garthoff, Redeker, Ortegel, Rau, and
  Weinfurter}}]{rosenfeld2017event}
\bibinfo{author}{\bibfnamefont{W.}~\bibnamefont{Rosenfeld}},
  \bibinfo{author}{\bibfnamefont{D.}~\bibnamefont{Burchardt}},
  \bibinfo{author}{\bibfnamefont{R.}~\bibnamefont{Garthoff}},
  \bibinfo{author}{\bibfnamefont{K.}~\bibnamefont{Redeker}},
  \bibinfo{author}{\bibfnamefont{N.}~\bibnamefont{Ortegel}},
  \bibinfo{author}{\bibfnamefont{M.}~\bibnamefont{Rau}}, \bibnamefont{and}
  \bibinfo{author}{\bibfnamefont{H.}~\bibnamefont{Weinfurter}},
  \bibinfo{journal}{Phys. Rev. Lett.} \textbf{\bibinfo{volume}{119}},
  \bibinfo{pages}{010402} (\bibinfo{year}{2017}),
  \urlprefix\url{https://link.aps.org/doi/10.1103/PhysRevLett.119.010402}.

\bibitem[{\citenamefont{Hensen et~al.}(2015)\citenamefont{Hensen, Bernien,
  Dr{\'e}au, Reiserer, Kalb, Blok, Ruitenberg, Vermeulen, Schouten, Abell{\'a}n
  et~al.}}]{hensen2015loophole}
\bibinfo{author}{\bibfnamefont{B.}~\bibnamefont{Hensen}},
  \bibinfo{author}{\bibfnamefont{H.}~\bibnamefont{Bernien}},
  \bibinfo{author}{\bibfnamefont{A.~E.} \bibnamefont{Dr{\'e}au}},
  \bibinfo{author}{\bibfnamefont{A.}~\bibnamefont{Reiserer}},
  \bibinfo{author}{\bibfnamefont{N.}~\bibnamefont{Kalb}},
  \bibinfo{author}{\bibfnamefont{M.~S.} \bibnamefont{Blok}},
  \bibinfo{author}{\bibfnamefont{J.}~\bibnamefont{Ruitenberg}},
  \bibinfo{author}{\bibfnamefont{R.~F.} \bibnamefont{Vermeulen}},
  \bibinfo{author}{\bibfnamefont{R.~N.} \bibnamefont{Schouten}},
  \bibinfo{author}{\bibfnamefont{C.}~\bibnamefont{Abell{\'a}n}},
  \bibnamefont{et~al.}, \bibinfo{journal}{Nature}
  \textbf{\bibinfo{volume}{526}}, \bibinfo{pages}{682} (\bibinfo{year}{2015}).

\bibitem[{\citenamefont{Liu et~al.}(2021)\citenamefont{Liu, Li, Ragy, Zhao,
  Bai, Liu, Brown, Zhang, Colbeck, Fan et~al.}}]{liu2021device}
\bibinfo{author}{\bibfnamefont{W.-Z.} \bibnamefont{Liu}},
  \bibinfo{author}{\bibfnamefont{M.-H.} \bibnamefont{Li}},
  \bibinfo{author}{\bibfnamefont{S.}~\bibnamefont{Ragy}},
  \bibinfo{author}{\bibfnamefont{S.-R.} \bibnamefont{Zhao}},
  \bibinfo{author}{\bibfnamefont{B.}~\bibnamefont{Bai}},
  \bibinfo{author}{\bibfnamefont{Y.}~\bibnamefont{Liu}},
  \bibinfo{author}{\bibfnamefont{P.~J.} \bibnamefont{Brown}},
  \bibinfo{author}{\bibfnamefont{J.}~\bibnamefont{Zhang}},
  \bibinfo{author}{\bibfnamefont{R.}~\bibnamefont{Colbeck}},
  \bibinfo{author}{\bibfnamefont{J.}~\bibnamefont{Fan}}, \bibnamefont{et~al.},
  \bibinfo{journal}{Nat. Phys.} \textbf{\bibinfo{volume}{17}},
  \bibinfo{pages}{448} (\bibinfo{year}{2021}).

\bibitem[{\citenamefont{Kraus et~al.}(2005)\citenamefont{Kraus, Gisin, and
  Renner}}]{kraus2005lower}
\bibinfo{author}{\bibfnamefont{B.}~\bibnamefont{Kraus}},
  \bibinfo{author}{\bibfnamefont{N.}~\bibnamefont{Gisin}}, \bibnamefont{and}
  \bibinfo{author}{\bibfnamefont{R.}~\bibnamefont{Renner}},
  \bibinfo{journal}{Phys. Rev. Lett.} \textbf{\bibinfo{volume}{95}},
  \bibinfo{pages}{080501} (\bibinfo{year}{2005}),
  \urlprefix\url{https://link.aps.org/doi/10.1103/PhysRevLett.95.080501}.

\bibitem[{\citenamefont{Ma and L{\"u}tkenhaus}(2012)}]{ma2012improved}
\bibinfo{author}{\bibfnamefont{X.}~\bibnamefont{Ma}} \bibnamefont{and}
  \bibinfo{author}{\bibfnamefont{N.}~\bibnamefont{L{\"u}tkenhaus}},
  \bibinfo{journal}{Quantum Inf Comput.} \textbf{\bibinfo{volume}{12}},
  \bibinfo{pages}{203} (\bibinfo{year}{2012}).

\bibitem[{\citenamefont{Gottesman and Lo}(2003)}]{gottesman2003proof}
\bibinfo{author}{\bibfnamefont{D.}~\bibnamefont{Gottesman}} \bibnamefont{and}
  \bibinfo{author}{\bibfnamefont{H.-K.} \bibnamefont{Lo}},
  \bibinfo{journal}{IEEE T. Inform. Theory} \textbf{\bibinfo{volume}{49}},
  \bibinfo{pages}{457} (\bibinfo{year}{2003}),
  \urlprefix\url{https://doi.org/10.1109/TIT.2002.807289}.

\bibitem[{\citenamefont{Chau}(2002)}]{chau2002practical}
\bibinfo{author}{\bibfnamefont{H.~F.} \bibnamefont{Chau}},
  \bibinfo{journal}{Phys. Rev. A} \textbf{\bibinfo{volume}{66}},
  \bibinfo{pages}{060302} (\bibinfo{year}{2002}),
  \urlprefix\url{https://link.aps.org/doi/10.1103/PhysRevA.66.060302}.

\bibitem[{\citenamefont{Ho et~al.}(2020)\citenamefont{Ho, Sekatski, Tan,
  Renner, Bancal, and Sangouard}}]{ho2020noisy}
\bibinfo{author}{\bibfnamefont{M.}~\bibnamefont{Ho}},
  \bibinfo{author}{\bibfnamefont{P.}~\bibnamefont{Sekatski}},
  \bibinfo{author}{\bibfnamefont{E.~Y.-Z.} \bibnamefont{Tan}},
  \bibinfo{author}{\bibfnamefont{R.}~\bibnamefont{Renner}},
  \bibinfo{author}{\bibfnamefont{J.-D.} \bibnamefont{Bancal}},
  \bibnamefont{and}
  \bibinfo{author}{\bibfnamefont{N.}~\bibnamefont{Sangouard}},
  \bibinfo{journal}{Phys. Rev. Lett.} \textbf{\bibinfo{volume}{124}},
  \bibinfo{pages}{230502} (\bibinfo{year}{2020}),
  \urlprefix\url{https://link.aps.org/doi/10.1103/PhysRevLett.124.230502}.

\bibitem[{\citenamefont{Tan et~al.}(2020)\citenamefont{Tan, Lim, and
  Renner}}]{tan2020advantage}
\bibinfo{author}{\bibfnamefont{E.~Y.-Z.} \bibnamefont{Tan}},
  \bibinfo{author}{\bibfnamefont{C.~C.-W.} \bibnamefont{Lim}},
  \bibnamefont{and} \bibinfo{author}{\bibfnamefont{R.}~\bibnamefont{Renner}},
  \bibinfo{journal}{Phys. Rev. Lett.} \textbf{\bibinfo{volume}{124}},
  \bibinfo{pages}{020502} (\bibinfo{year}{2020}),
  \urlprefix\url{https://link.aps.org/doi/10.1103/PhysRevLett.124.020502}.

\bibitem[{\citenamefont{Jain et~al.}(2020)\citenamefont{Jain, Miller, and
  Shi}}]{jain2020parallel}
\bibinfo{author}{\bibfnamefont{R.}~\bibnamefont{Jain}},
  \bibinfo{author}{\bibfnamefont{C.~A.} \bibnamefont{Miller}},
  \bibnamefont{and} \bibinfo{author}{\bibfnamefont{Y.}~\bibnamefont{Shi}},
  \bibinfo{journal}{IEEE Trans. Inf. Theory} \textbf{\bibinfo{volume}{66}},
  \bibinfo{pages}{5567} (\bibinfo{year}{2020}).

\bibitem[{\citenamefont{Eberhard}(1993)}]{eberhard1993background}
\bibinfo{author}{\bibfnamefont{P.~H.} \bibnamefont{Eberhard}},
  \bibinfo{journal}{Phys. Rev. A} \textbf{\bibinfo{volume}{47}},
  \bibinfo{pages}{R747} (\bibinfo{year}{1993}),
  \urlprefix\url{https://link.aps.org/doi/10.1103/PhysRevA.47.R747}.

\bibitem[{\citenamefont{Barrett et~al.}(2013)\citenamefont{Barrett, Colbeck,
  and Kent}}]{barrett2013memory}
\bibinfo{author}{\bibfnamefont{J.}~\bibnamefont{Barrett}},
  \bibinfo{author}{\bibfnamefont{R.}~\bibnamefont{Colbeck}}, \bibnamefont{and}
  \bibinfo{author}{\bibfnamefont{A.}~\bibnamefont{Kent}},
  \bibinfo{journal}{Phys. Rev. Lett.} \textbf{\bibinfo{volume}{110}},
  \bibinfo{pages}{010503} (\bibinfo{year}{2013}),
  \urlprefix\url{https://link.aps.org/doi/10.1103/PhysRevLett.110.010503}.

\bibitem[{\citenamefont{Curty and Lo}(2019)}]{curty2019foiling}
\bibinfo{author}{\bibfnamefont{M.}~\bibnamefont{Curty}} \bibnamefont{and}
  \bibinfo{author}{\bibfnamefont{H.-K.} \bibnamefont{Lo}},
  \bibinfo{journal}{Npj Quantum Inf.} \textbf{\bibinfo{volume}{5}},
  \bibinfo{pages}{1} (\bibinfo{year}{2019}).

\bibitem[{\citenamefont{Ben-Or et~al.}(2005)\citenamefont{Ben-Or, Horodecki,
  Leung, Mayers, and Oppenheim}}]{ben2005universal}
\bibinfo{author}{\bibfnamefont{M.}~\bibnamefont{Ben-Or}},
  \bibinfo{author}{\bibfnamefont{M.}~\bibnamefont{Horodecki}},
  \bibinfo{author}{\bibfnamefont{D.~W.} \bibnamefont{Leung}},
  \bibinfo{author}{\bibfnamefont{D.}~\bibnamefont{Mayers}}, \bibnamefont{and}
  \bibinfo{author}{\bibfnamefont{J.}~\bibnamefont{Oppenheim}}, in
  \emph{\bibinfo{booktitle}{Proceedings of the Second International Conference
  on Theory of Cryptography}} (\bibinfo{publisher}{Springer-Verlag},
  \bibinfo{address}{Berlin, Heidelberg}, \bibinfo{year}{2005}), TCC'05, pp.
  \bibinfo{pages}{386--406}, ISBN \bibinfo{isbn}{3-540-24573-1,
  978-3-540-24573-5},
  \urlprefix\url{http://dx.doi.org/10.1007/978-3-540-30576-7_21}.

\bibitem[{\citenamefont{Renner and K\"{o}nig}(2005)}]{renner2005Security}
\bibinfo{author}{\bibfnamefont{R.}~\bibnamefont{Renner}} \bibnamefont{and}
  \bibinfo{author}{\bibfnamefont{R.}~\bibnamefont{K\"{o}nig}}, in
  \emph{\bibinfo{booktitle}{Proceedings of the Second International Conference
  on Theory of Cryptography}} (\bibinfo{publisher}{Springer-Verlag},
  \bibinfo{address}{Berlin, Heidelberg}, \bibinfo{year}{2005}), TCC'05, pp.
  \bibinfo{pages}{407--425}, ISBN \bibinfo{isbn}{3-540-24573-1,
  978-3-540-24573-5},
  \urlprefix\url{http://dx.doi.org/10.1007/978-3-540-30576-7_22}.

\bibitem[{Note2()}]{Note2}
Note2, \bibinfo{note}{there is an error in the security parameter in the
  original work of Ref.~\cite {koashi2009simple} due to the definition of
  fidelity measures. This is later corrected in Ref.~\cite
  {fung2010practical}}.

\bibitem[{\citenamefont{Fung et~al.}(2010)\citenamefont{Fung, Ma, and
  Chau}}]{fung2010practical}
\bibinfo{author}{\bibfnamefont{C.-H.~F.} \bibnamefont{Fung}},
  \bibinfo{author}{\bibfnamefont{X.}~\bibnamefont{Ma}}, \bibnamefont{and}
  \bibinfo{author}{\bibfnamefont{H.~F.} \bibnamefont{Chau}},
  \bibinfo{journal}{Phys. Rev. A} \textbf{\bibinfo{volume}{81}},
  \bibinfo{pages}{012318} (\bibinfo{year}{2010}),
  \urlprefix\url{https://link.aps.org/doi/10.1103/PhysRevA.81.012318}.

\bibitem[{\citenamefont{Lo}(2003)}]{lo2003method}
\bibinfo{author}{\bibfnamefont{H.-K.} \bibnamefont{Lo}}, \bibinfo{journal}{New
  J. Phys.} \textbf{\bibinfo{volume}{5}}, \bibinfo{pages}{36}
  (\bibinfo{year}{2003}).

\bibitem[{\citenamefont{Huang et~al.}(2021)\citenamefont{Huang, Zhang, and
  Ma}}]{huang2021stream}
\bibinfo{author}{\bibfnamefont{Y.}~\bibnamefont{Huang}},
  \bibinfo{author}{\bibfnamefont{X.}~\bibnamefont{Zhang}}, \bibnamefont{and}
  \bibinfo{author}{\bibfnamefont{X.}~\bibnamefont{Ma}},
  \bibinfo{journal}{arXiv:2111.14108}  (\bibinfo{year}{2021}).

\bibitem[{\citenamefont{Brassard and Salvail}(1994)}]{brassard1993secret}
\bibinfo{author}{\bibfnamefont{G.}~\bibnamefont{Brassard}} \bibnamefont{and}
  \bibinfo{author}{\bibfnamefont{L.}~\bibnamefont{Salvail}}, in
  \emph{\bibinfo{booktitle}{Advances in Cryptology --- EUROCRYPT '93}}, edited
  by \bibinfo{editor}{\bibfnamefont{T.}~\bibnamefont{Helleseth}}
  (\bibinfo{publisher}{Springer Berlin Heidelberg}, \bibinfo{address}{Berlin,
  Heidelberg}, \bibinfo{year}{1994}), pp. \bibinfo{pages}{410--423}, ISBN
  \bibinfo{isbn}{978-3-540-48285-7}.

\bibitem[{\citenamefont{Tang et~al.}(2021)\citenamefont{Tang, Liu, Yu, and
  Wu}}]{tang2021shannon}
\bibinfo{author}{\bibfnamefont{B.-Y.} \bibnamefont{Tang}},
  \bibinfo{author}{\bibfnamefont{B.}~\bibnamefont{Liu}},
  \bibinfo{author}{\bibfnamefont{W.-R.} \bibnamefont{Yu}}, \bibnamefont{and}
  \bibinfo{author}{\bibfnamefont{C.-Q.} \bibnamefont{Wu}},
  \bibinfo{journal}{Quantum Inf. Process.} \textbf{\bibinfo{volume}{20}},
  \bibinfo{pages}{1} (\bibinfo{year}{2021}).

\bibitem[{\citenamefont{Neumark}(1940)}]{neumark1940spectral}
\bibinfo{author}{\bibfnamefont{M.}~\bibnamefont{Neumark}},
  \bibinfo{journal}{Izv. Ross. Akad. Nauk Ser. Mat.}
  \textbf{\bibinfo{volume}{4}}, \bibinfo{pages}{277} (\bibinfo{year}{1940}).

\bibitem[{\citenamefont{Bennett et~al.}(1993)\citenamefont{Bennett, Brassard,
  Cr\'epeau, Jozsa, Peres, and Wootters}}]{bennett1993teleporting}
\bibinfo{author}{\bibfnamefont{C.~H.} \bibnamefont{Bennett}},
  \bibinfo{author}{\bibfnamefont{G.}~\bibnamefont{Brassard}},
  \bibinfo{author}{\bibfnamefont{C.}~\bibnamefont{Cr\'epeau}},
  \bibinfo{author}{\bibfnamefont{R.}~\bibnamefont{Jozsa}},
  \bibinfo{author}{\bibfnamefont{A.}~\bibnamefont{Peres}}, \bibnamefont{and}
  \bibinfo{author}{\bibfnamefont{W.~K.} \bibnamefont{Wootters}},
  \bibinfo{journal}{Phys. Rev. Lett.} \textbf{\bibinfo{volume}{70}},
  \bibinfo{pages}{1895} (\bibinfo{year}{1993}),
  \urlprefix\url{https://link.aps.org/doi/10.1103/PhysRevLett.70.1895}.

\bibitem[{\citenamefont{Tsirelson}(1993)}]{tsirelson1993some}
\bibinfo{author}{\bibfnamefont{B.~S.} \bibnamefont{Tsirelson}},
  \bibinfo{journal}{Hadronic Journal Supplement} \textbf{\bibinfo{volume}{8}},
  \bibinfo{pages}{329} (\bibinfo{year}{1993}).

\bibitem[{\citenamefont{Berta et~al.}(2010)\citenamefont{Berta, Christandl,
  Colbeck, Renes, and Renner}}]{berta2010uncertainty}
\bibinfo{author}{\bibfnamefont{M.}~\bibnamefont{Berta}},
  \bibinfo{author}{\bibfnamefont{M.}~\bibnamefont{Christandl}},
  \bibinfo{author}{\bibfnamefont{R.}~\bibnamefont{Colbeck}},
  \bibinfo{author}{\bibfnamefont{J.~M.} \bibnamefont{Renes}}, \bibnamefont{and}
  \bibinfo{author}{\bibfnamefont{R.}~\bibnamefont{Renner}},
  \bibinfo{journal}{Nature Physics} \textbf{\bibinfo{volume}{6}},
  \bibinfo{pages}{659} (\bibinfo{year}{2010}).

\bibitem[{\citenamefont{Cover and Thomas}(2012)}]{cover2012elements}
\bibinfo{author}{\bibfnamefont{T.~M.} \bibnamefont{Cover}} \bibnamefont{and}
  \bibinfo{author}{\bibfnamefont{J.~A.} \bibnamefont{Thomas}},
  \emph{\bibinfo{title}{Elements of information theory}}
  (\bibinfo{publisher}{John Wiley \& Sons}, \bibinfo{year}{2012}).

\bibitem[{\citenamefont{Wilde}(2011)}]{wilde2011classical}
\bibinfo{author}{\bibfnamefont{M.~M.} \bibnamefont{Wilde}},
  \bibinfo{journal}{arXiv:1106.1445}  (\bibinfo{year}{2011}).

\bibitem[{\citenamefont{Renes and Smith}(2007)}]{renes2007noisy}
\bibinfo{author}{\bibfnamefont{J.~M.} \bibnamefont{Renes}} \bibnamefont{and}
  \bibinfo{author}{\bibfnamefont{G.}~\bibnamefont{Smith}},
  \bibinfo{journal}{Phys. Rev. Lett.} \textbf{\bibinfo{volume}{98}},
  \bibinfo{pages}{020502} (\bibinfo{year}{2007}),
  \urlprefix\url{https://link.aps.org/doi/10.1103/PhysRevLett.98.020502}.

\bibitem[{\citenamefont{Arnon-Friedman
  et~al.}(2019)\citenamefont{Arnon-Friedman, Renner, and
  Vidick}}]{arnon2019simple}
\bibinfo{author}{\bibfnamefont{R.}~\bibnamefont{Arnon-Friedman}},
  \bibinfo{author}{\bibfnamefont{R.}~\bibnamefont{Renner}}, \bibnamefont{and}
  \bibinfo{author}{\bibfnamefont{T.}~\bibnamefont{Vidick}},
  \bibinfo{journal}{SIAM J. Comput.} \textbf{\bibinfo{volume}{48}},
  \bibinfo{pages}{181} (\bibinfo{year}{2019}).

\bibitem[{\citenamefont{Renner}(2008)}]{renner2008security}
\bibinfo{author}{\bibfnamefont{R.}~\bibnamefont{Renner}},
  \bibinfo{journal}{Int. J. Quantum Inf.} \textbf{\bibinfo{volume}{06}},
  \bibinfo{pages}{1} (\bibinfo{year}{2008}),
  \urlprefix\url{https://doi.org/10.1142/S0219749908003256}.

\bibitem[{\citenamefont{Arnon-Friedman and Bancal}(2019)}]{arnon2019device}
\bibinfo{author}{\bibfnamefont{R.}~\bibnamefont{Arnon-Friedman}}
  \bibnamefont{and} \bibinfo{author}{\bibfnamefont{J.-D.}
  \bibnamefont{Bancal}}, \bibinfo{journal}{New J. Phys.}
  \textbf{\bibinfo{volume}{21}}, \bibinfo{pages}{033010}
  (\bibinfo{year}{2019}).

\bibitem[{\citenamefont{Mao et~al.}(2022)\citenamefont{Mao, Li, Hao,
  Abd-El-Atty, and Iliyasu}}]{mao2022high}
\bibinfo{author}{\bibfnamefont{H.-K.} \bibnamefont{Mao}},
  \bibinfo{author}{\bibfnamefont{Q.}~\bibnamefont{Li}},
  \bibinfo{author}{\bibfnamefont{P.-L.} \bibnamefont{Hao}},
  \bibinfo{author}{\bibfnamefont{B.}~\bibnamefont{Abd-El-Atty}},
  \bibnamefont{and} \bibinfo{author}{\bibfnamefont{A.~M.}
  \bibnamefont{Iliyasu}}, \bibinfo{journal}{Opt. Quantum Electron.}
  \textbf{\bibinfo{volume}{54}}, \bibinfo{pages}{1} (\bibinfo{year}{2022}).

\bibitem[{Note3()}]{Note3}
Note3, \bibinfo{note}{we do not reproduce the results in Ref.~\cite
  {nadlinger2022experimental} due to lack of some optimisation details in the
  advanced EAT-based results. We refer readers who are interested in this issue
  to Figure 4 in Ref.~\cite {nadlinger2022experimental}.}

\bibitem[{\citenamefont{Bhavsar et~al.}(2021)\citenamefont{Bhavsar, Ragy, and
  Colbeck}}]{bhavsar2021calculation}
\bibinfo{author}{\bibfnamefont{R.}~\bibnamefont{Bhavsar}},
  \bibinfo{author}{\bibfnamefont{S.}~\bibnamefont{Ragy}}, \bibnamefont{and}
  \bibinfo{author}{\bibfnamefont{R.}~\bibnamefont{Colbeck}},
  \bibinfo{journal}{arXiv:2103.07504}  (\bibinfo{year}{2021}).

\bibitem[{\citenamefont{Dupuis and Fawzi}(2019)}]{dupuis2019entropy}
\bibinfo{author}{\bibfnamefont{F.}~\bibnamefont{Dupuis}} \bibnamefont{and}
  \bibinfo{author}{\bibfnamefont{O.}~\bibnamefont{Fawzi}},
  \bibinfo{journal}{IEEE Trans. Inf. Theory} \textbf{\bibinfo{volume}{65}},
  \bibinfo{pages}{7596} (\bibinfo{year}{2019}).

\bibitem[{\citenamefont{Metger et~al.}(2022)\citenamefont{Metger, Fawzi,
  Sutter, and Renner}}]{metger2022generalised}
\bibinfo{author}{\bibfnamefont{T.}~\bibnamefont{Metger}},
  \bibinfo{author}{\bibfnamefont{O.}~\bibnamefont{Fawzi}},
  \bibinfo{author}{\bibfnamefont{D.}~\bibnamefont{Sutter}}, \bibnamefont{and}
  \bibinfo{author}{\bibfnamefont{R.}~\bibnamefont{Renner}},
  \emph{\bibinfo{title}{Generalised entropy accumulation}}
  (\bibinfo{year}{2022}), \urlprefix\url{https://arxiv.org/abs/2203.04989}.

\bibitem[{\citenamefont{Lo et~al.}(2012)\citenamefont{Lo, Curty, and
  Qi}}]{lo2012measurement}
\bibinfo{author}{\bibfnamefont{H.-K.} \bibnamefont{Lo}},
  \bibinfo{author}{\bibfnamefont{M.}~\bibnamefont{Curty}}, \bibnamefont{and}
  \bibinfo{author}{\bibfnamefont{B.}~\bibnamefont{Qi}}, \bibinfo{journal}{Phys.
  Rev. Lett.} \textbf{\bibinfo{volume}{108}}, \bibinfo{pages}{130503}
  (\bibinfo{year}{2012}),
  \urlprefix\url{https://link.aps.org/doi/10.1103/PhysRevLett.108.130503}.

\bibitem[{\citenamefont{Braunstein and Pirandola}(2012)}]{braunstein2012side}
\bibinfo{author}{\bibfnamefont{S.~L.} \bibnamefont{Braunstein}}
  \bibnamefont{and}
  \bibinfo{author}{\bibfnamefont{S.}~\bibnamefont{Pirandola}},
  \bibinfo{journal}{Phys. Rev. Lett.} \textbf{\bibinfo{volume}{108}},
  \bibinfo{pages}{130502} (\bibinfo{year}{2012}),
  \urlprefix\url{https://link.aps.org/doi/10.1103/PhysRevLett.108.130502}.

\bibitem[{\citenamefont{Brunner et~al.}(2014)\citenamefont{Brunner, Cavalcanti,
  Pironio, Scarani, and Wehner}}]{brunner2014bell}
\bibinfo{author}{\bibfnamefont{N.}~\bibnamefont{Brunner}},
  \bibinfo{author}{\bibfnamefont{D.}~\bibnamefont{Cavalcanti}},
  \bibinfo{author}{\bibfnamefont{S.}~\bibnamefont{Pironio}},
  \bibinfo{author}{\bibfnamefont{V.}~\bibnamefont{Scarani}}, \bibnamefont{and}
  \bibinfo{author}{\bibfnamefont{S.}~\bibnamefont{Wehner}},
  \bibinfo{journal}{Rev. Mod. Phys.} \textbf{\bibinfo{volume}{86}},
  \bibinfo{pages}{419} (\bibinfo{year}{2014}),
  \urlprefix\url{https://link.aps.org/doi/10.1103/RevModPhys.86.419}.

\bibitem[{\citenamefont{Navascu\'es et~al.}(2007)\citenamefont{Navascu\'es,
  Pironio, and Ac\'{\i}n}}]{navascues2007bounding}
\bibinfo{author}{\bibfnamefont{M.}~\bibnamefont{Navascu\'es}},
  \bibinfo{author}{\bibfnamefont{S.}~\bibnamefont{Pironio}}, \bibnamefont{and}
  \bibinfo{author}{\bibfnamefont{A.}~\bibnamefont{Ac\'{\i}n}},
  \bibinfo{journal}{Phys. Rev. Lett.} \textbf{\bibinfo{volume}{98}},
  \bibinfo{pages}{010401} (\bibinfo{year}{2007}),
  \urlprefix\url{https://link.aps.org/doi/10.1103/PhysRevLett.98.010401}.

\bibitem[{\citenamefont{{Colbeck}}(2006)}]{Colbeck09}
\bibinfo{author}{\bibfnamefont{R.}~\bibnamefont{{Colbeck}}}, Ph.D. thesis,
  \bibinfo{school}{Trinity College, University of Cambridge, arXiv: 0911.3814}
  (\bibinfo{year}{2006}).

\bibitem[{\citenamefont{Tsurumaru}(2022)}]{tsurumaru2021equivalence}
\bibinfo{author}{\bibfnamefont{T.}~\bibnamefont{Tsurumaru}},
  \bibinfo{journal}{IEEE Trans. Inf. Theory} \textbf{\bibinfo{volume}{68}},
  \bibinfo{pages}{1016} (\bibinfo{year}{2022}).

\bibitem[{\citenamefont{Bennett et~al.}(1996)\citenamefont{Bennett, DiVincenzo,
  Smolin, and Wootters}}]{Bennett96Mixed}
\bibinfo{author}{\bibfnamefont{C.~H.} \bibnamefont{Bennett}},
  \bibinfo{author}{\bibfnamefont{D.~P.} \bibnamefont{DiVincenzo}},
  \bibinfo{author}{\bibfnamefont{J.~A.} \bibnamefont{Smolin}},
  \bibnamefont{and} \bibinfo{author}{\bibfnamefont{W.~K.}
  \bibnamefont{Wootters}}, \bibinfo{journal}{Phys. Rev. A}
  \textbf{\bibinfo{volume}{54}}, \bibinfo{pages}{3824} (\bibinfo{year}{1996}),
  \urlprefix\url{https://link.aps.org/doi/10.1103/PhysRevA.54.3824}.

\end{thebibliography}

\end{document}